\documentclass[ecta, nameyear, draft]{econsocart} 
\RequirePackage{xr-hyper}
\RequirePackage[colorlinks, citecolor = blue, linkcolor = blue, urlcolor = blue, pagebackref]{hyperref}

\startlocaldefs
\usepackage{subfiles, subcaption}
\usepackage{natbib}
\usepackage{graphicx, tabularx, booktabs}
\usepackage{enumitem}
\usepackage{setspace}

\makeatletter
\renewcommand{\put@numberlines@box}{}
\renewcommand{\copyright@text}{} 
\makeatother

\theoremstyle{plain}

\newtheorem{theorem}{Theorem}
\newtheorem{property}{Property}

\newtheorem{lemma}{Lemma}

\theoremstyle{definition}

\newtheorem*{assumptionA}{Assumption A}

\newtheorem*{assumptionB1}{Assumption B.1}
\newtheorem*{assumptionB2}{Assumption B.2}
\newtheorem*{assumptionB3}{Assumption B.3}

\newcommand{\mathbbm}[1]{\text{\usefont{U}{bbm}{m}{n}#1}}
\DeclareMathOperator{\detr}{det}

\DeclareMathOperator{\diag}{diag}
\DeclareMathOperator{\mean}{E}
\DeclareMathOperator{\pr}{P}
\DeclareMathOperator{\var}{Var}
\DeclareMathOperator{\cov}{Cov}

\DeclareMathOperator{\sgn}{sgn}
\DeclareMathOperator{\argmax}{argmax}
\DeclareMathOperator{\argmin}{argmin}
\endlocaldefs
\renewcommand{\thetable}{\arabic{table}}

\begin{document}
\patchcmd{\subsection}
{\bigskipamount}
{8pt plus 1pt minus 1pt}
{}{}
\widowpenalty=300
\begin{frontmatter}
\title{Discrete Choice with Generalized Social Interactions}
\runtitle{Discrete Choice with Generalized Social Interactions}

\begin{aug}
\author[add1]{\fnms{Oscar}~\snm{Volpe}\ead[label=e1]{ovolpe@fas.harvard.edu}}
\address[add1]{\orgdiv{Department of Economics}, \orgname{Harvard University}}
\end{aug}

\begin{center}
	\vspace{-0.35cm}
	September 30, 2026
	\vspace{0.75cm}
\end{center}

\begin{funding}
I thank Isaiah Andrews, St\'{e}phane Bonhomme, Ozan Candogan, Arun Chandrasekhar, Michael Dinerstein, Ste-\\ven Durlauf, Ali Horta\c{c}su, Robert Moffitt, Magne Mogstad, Francesca Molinari, Phil Reny, and Alex Torgovitsky for helpful comments and suggestions, as well as seminar participants at the North American UEA Meeting, SITE (Empirical Implementation of Theoretical Models), and the Applied Micro, I/O, and 3rd-year PhD seminars at the University of Chicago. I am grateful to three anonymous referees for invaluable feedback. All errors are my own.
\end{funding}

\begin{abstract}
	This paper studies social interactions in discrete choice models in which indivi-\\duals can respond differently to different network members, conforming to some while distinguishing from others.\hspace{-0.03cm} In this generalized framework,\hspace{-0.02cm} I show how to use within-network variation in individual choices to draw inference about differences in endogenous interaction effects, while accounting for unobserved contextual effects. I also show how to construct internal instruments to correct for measurement error bias, a key source of endogeneity in network models with incomplete information.\hspace{-0.03cm} I apply these results to data from two empirical settings:\hspace{-0.01cm} classroom peer ef-\\[-0.65cm]
	
	\noindent
	fects in Tennessee primary schools and spillovers in deworming treatment uptake in Kenya. In both applications, I find that differences in social interaction effects, \\ with individuals conforming more strongly to similar peers, play a key role in sha-\\ping people's choices, with implications for the targeting of policy interventions. \\
\end{abstract}

\begin{keyword}
	\kwd{social interaction effects}
	\kwd{discrete choice}
	\kwd{strategic complements and substitutes}
	\kwd{contextual effects}
	\kwd{measurement error bias}
	\kwd{partial identification}
\end{keyword}

\vspace{-0.6cm}
\end{frontmatter}

\section{Introduction}\label{s1}

\noindent
The role of social interactions in shaping individual behavior is an active research topic in economics.\ Most empirical work has focused on settings with uniform positive interactions, where all agents conform in the same way to the average action of others in their network. Meanwhile, less attention has been given to settings with nonuniform or mixed-sign interactions, where agents are affected differently by different network members\textemdash conforming to some while distinguishing from others.\ These richer forms of interaction can be harder to evaluate empirically, yet they are important for understanding behavior in many economic settings, where social identity shapes how people's choices depend on those around them.

In this paper, I develop empirical tools to study discrete choice models with \textit{generalized} social interactions, where individuals can experience both positive and negative spillovers. I build an equilibrium framework that extends the classical model of \cite{brdu2001} by allowing the sign and magnitude of interaction effects to vary among agents in a network. Under this framework, I show how to recover key economic quantities from indiv-\\idual choice data, and I construct an estimator for endogenous interaction effects that can be readily applied to standard datasets. I then apply these methods to data from two large-scale education experiments to draw inference about the role of nonuniform interaction effects.

In Section 2, I begin by modeling a network of agents who make binary choices subject to social interaction effects and private payoffs. As in \cite{brdu2001}, I assume agents act with incomplete information, responding to expected rather than realized choices of others. Yet, unlike previous work, I allow interaction effects to vary among agents and to take negative values. To achieve this, I extend the model to allow for a matrix of interaction effects $\mathbf{J} = [J_{k\ell}]_{k,\ell}$,\hspace{-0.01cm} where $J_{k\ell}$ measures how agents of type $k$ respond to agents of type $\ell$. \\[-0.55cm]

By imposing no restrictions on $\mathbf{J}$, I show that the model can be used to study a range of economic decision problems across different network structures, featuring both conformity and differentiation. Under this generalized setting, I establish conditions for the equilibrium to be unique and locally stable, characterizing how these properties depend on the signs and \\[-0.75cm]

\noindent
magnitudes of the interaction effects.\footnote{While these properties are well-studied for uniform, positive interactions, they do not readily transfer to cases with negative interactions; see \cite{jaze2015} for discussion. My analysis contributes by showing how key results associated with positive interactions extend to a broader set of models involving negative interactions.} I also show that, when multiple equilibria exist, negative interactions can create welfare tradeoffs, in which no equilibrium is broadly favorable.

Next, in Section 3, I turn to the question of identification. That is, given data on individual choices across multiple networks, what can be learned about the role of social interactions? The identification of discrete choice network models with incomplete information has been studied extensively by \cite{brdu2001, brdu2007} and more recently by \cite{bhduka2023}. \cite{gr2015}, \cite{pa2017}, and \cite{klta2020} give surveys of \\ the literature. My analysis builds on this work by tackling two key identification challenges.

First,\hspace{-0.03cm} I address the issue of unobserved network effects\textemdash contextual factors that obscure the effect of social interactions on group-level outcomes.\ For example, from data on student achievement, it is hard to isolate the role of peer effects from unobserved class characteristics like teacher quality. Much of the applied literature rules out or places strong parametric restrictions on these unobservables.\footnote{For example, \cite{brdu2001, brdu2007} assume that the contextual effect is a constant linear function of observed variables, whereas \cite{bhduka2023} assume that it has a specific linear factor structure. Others, including \cite{ho2000}, propose using a panel to difference-out the fixed effects between time periods. However, this approach requires both access to panel data and the assumption that contextual effects are invariant over time.} However, this risks introducing omitted variable bias.

To handle this issue, I propose a partial identification strategy to recover key features of $\mathbf{J}$ without imposing any restrictions on unobserved network effects. This method exploits a panel structure inherent to the model, where agents in the same network interact in varying ways. Since these agents are exposed to the same contextual factors, I can control for network unobservables by comparing outcomes of different agents in the same network. This method allows me to draw inference about systematic variation in social interaction effects.

Second, I address a core issue in network models with incomplete information: agents' expectations are not observed by the researcher. Instead, a researcher only observes agents' realized choices in finite networks. Since the average choice in the network converges to an expectation as the network size grows large, previous work has treated this issue as an esti-\\mation problem rather than a threat to identification.\footnote{Previously, estimation has relied on double asymptotic arguments, requiring both the number and size of the networks to tend to infinity; e.g., see \cite{lelili2014, gumo2022}, and \cite{bhduka2023}.}\ I argue that this approach is inadvis-\\able, since, for any finite network, the average action is endogenous in the model due to me-\\[-0.75cm]

\noindent
asurement error. I show that, by failing to account for this endogeneity, standard estimators generally produce highly biased estimates, even as the number of networks tends to infinity.

I approach this issue as a classical measurement error problem, noting that realized average choices are noisy proxies for true expectations and are therefore endogenous. To correct for endogeneity, I construct internal instruments: randomly partitioning networks into two \\ parts, and using the average action in one part as an instrument for the average action in the other. Under the model’s assumptions, this method yields an IV estimator for social interaction effects that is consistent as the number of networks grows, even when network sizes remain fixed. I demonstrate that this estimator performs well in Monte Carlo simulations.

Finally, in Section 4, I apply these results to data from two empirical settings.\ First, I an-\\alyze Project STAR\textemdash a large-scale education experiment in Tennessee that randomly assi-\\gned students to classrooms\textemdash to study how peer effects vary by gender. I find that boys and girls are much more likely to conform to peers of their own gender than to peers of the other gender. Next, I use data from the Primary School Deworming Project, previously studied by \cite{mikr2004}, to estimate how spillovers in treatment uptake vary by age and gender. Again, I find that students prefer to conform to similar peers. Moreover, this own-group conformity is strongest among younger students. In both applications, I find that the generalized interactions model uncovers systematic differences in interaction effects that standard frameworks would overlook. Accounting for these differences provides a richer understanding of how social interactions shape the effects of targeted policies in networks.

\section{Model with Generalized Interaction Effects}\label{s2}\enlargethispage{\baselineskip}

\noindent
In this section, I introduce a binary choice model with social interactions that extends the cl-\\[-0.75cm]

\noindent
assical model of \cite{brdu2001} by allowing the sign and strength of interaction effects to vary across individuals. Unlike in the classical model, the analysis is nonparametric in the distribution of idiosyncratic preferences. Therefore, the equilibrium properties and subsequent identification results are robust to a wide range of functional form assumptions.

\subsection{Agent Preferences}\label{s2.1}

\noindent
Consider a network of agents indexed by $i$, where an agent belongs to one of $K$ types.\ These types can represent social identities, such as gender, class, race, or political partisanship, or they can represent organizational roles within the network, such as workers and supervisors, suppliers and manufacturers, patients and clinicians, or villagers and community leaders. 

Each agent chooses a binary action $\omega_i \in \{0, 1\}$ at a common time. Let $\bar\omega^k$ be the average action among agents of type $k$, and let $\bar\omega^k_{-i}$ be the average action among agents of type $k$ excluding $i$. The individual utility from choosing an action $\omega_i$ for any agent $i$ of type $k$ is:
\begin{equation}\label{eq:1}
	U_i(\omega_i|k) = v_k(\omega_i) + J_{kk}\omega_i \mean_i(\bar\omega_{-i}^k) + \sum_{\ell \neq k} J_{k\ell}\omega_i \mean_i(\bar\omega^\ell) + \epsilon_i(\omega_i). \vspace{-0.2cm}
\end{equation}
In this utility function, $v_k(\omega_i)$ is the deterministic private payoff, which can depend on the agent's type $k$, and $\epsilon_i(\omega_i)$ is an idiosyncratic payoff, which can vary across individuals.\ The terms $\mean_i(\bar\omega_{-i}^k)$ and $\mean_i(\bar\omega^\ell)$ are agent $i$'s subjective beliefs about $\bar\omega_{-i}^k$ and $\bar\omega^\ell$, respectively.\footnote{Another common way to model spillovers is through a quadratic conformity effect, $-\frac{1}{2}J_{k\ell}[\omega_i - \mean_i(\bar\omega^\ell)]^2$, as studied by \cite{be1994}. This parameterization is equivalent in that it leads agents to make the same choices.}

The parameter $J_{k\ell}$ is the interaction effect,\hspace{-0.01cm} measuring how agents of type $k$ are influenced by the expected average choice of agents of type $\ell$. These terms can be collected in a matrix:
\begin{equation}\label{eq:2}
	\mathbf{J} = \small
	\begin{bmatrix}
		J_{11} & J_{12} & \cdots & J_{1K} \\
		J_{21} & J_{22} & \cdots & J_{2K} \\
		\vdots & \vdots & \ddots & \vdots \\
		J_{K1} & J_{K2} & \cdots & J_{KK}
	\end{bmatrix}. \normalsize
\end{equation}
When analyzing model behavior, I refer to $\mathbf{J}$ as the \textit{interaction matrix}. Its entries characterize the direction, sign, and magnitude of social influence among different types of agents.\footnote{$\mathbf{J}$ can be interpreted as the adjacency matrix of a directed graph, where nodes represent different agent types.}

Since choices are binary, the private payoffs can, without loss of generality, be written as linear functions of the choice:\ $v_k(\omega_i) = h_k\omega_i + \eta_k$ and $\epsilon_i(\omega_i) = \varepsilon_i\omega_i + \xi_i$. Here, $h_k$ is the deterministic payoff from selecting $\omega_i = 1$, and $\varepsilon_i$ is the idiosyncratic taste for this action.\ The tastes are assumed independent across agents, and they can have type-specific distributions: \\[-0.55cm]
\begin{equation}\label{eq:3}
	\pr(\varepsilon_i \leq z|k) = F_{\varepsilon|k}(z), \quad\mbox{for } k = 1, \dots, K,
\end{equation}
where $F_{\varepsilon|k}$ is strictly increasing, smooth, and symmetric about zero.\ No additional parame-\\tric assumptions are maintained throughout the analysis. Therefore, the framework accommodates a variety of common empirical specifications, such as logistic and Gaussian errors.

\subsection{Economic Examples of Generalized Interactions}\label{s2.2}

\noindent
The classical model of social interactions for binary outcomes \citep{brdu2001} imposes two key restrictions on the interaction effects. First, it assumes uniformity ($J_{k\ell} = J$ for all $k$ and $\ell$), so all agents influence each other in the same way, regardless of their types. Second, it assumes that all interaction effects are nonnegative ($J_{k\ell} \geq 0$ for all $k$ and $\ell$), so each agent's payoff weakly rises when others take the same action.\ This restriction rules out strategic substitutability, where agents seek to differentiate from certain network members.

I relax the first restriction by allowing the interaction effects to vary among types and also \\[-0.75cm]

\noindent
by not imposing symmetry ($J_{k\ell} = J_{\ell k}$), so the influence of type $k$ on type $\ell$ can differ from that of type $\ell$ on type $k$. This generalization makes the model well-suited to study a variety of economic decision problems across different network structures. Table 1, Panel A, pres-\\ents four commonly studied network structures in economics, as well as the conditions on $\{J_{k\ell}\}_{k, \ell}$ that define each one. The first is a \textit{tree}, used to study spillovers in organizational hierarchies and supply chains, as well as the spread of information or disease. The second is a \textit{circle}, relevant for analyzing local interactions among close contacts \citep{el1993} and spatial competition between nearby firms \citep{sa1979}. The third is a \textit{bipartite} network, used to study two-sided markets linking workers and firms or buyers and sellers. The fourth is a \textit{complete} network, in which all agents interact with one another, as in classical peer eff-\\ects models. Each of these networks is nested under the generalized interactions framework.

I relax the second restriction by allowing for any combination of positive and negative interaction effects. This extension makes the model applicable to environments where agents seek to differentiate their behavior from that of others.\hspace{-0.035cm} Examples, shown in Table 1, Panel B, \\[-0.75cm]

\noindent
include \textit{free-riding} \citep{brkr2007}, where agents contribute less to a public \\ good when others contribute more; \textit{social distinction} \citep{akkr2000}, where agents differentiate their choices to reinforce identity; \textit{segregation} \citep{sc1971}, where agents prefer actions that result in the exclusion of particular peers;\hspace{-0.01cm} and \textit{polarization} \citep{bogesh2022}, where agents deviate from actions taken by those with opposing views. In each case, an agent's preference to distinguish her choice from those of others can be modeled as a negative interaction effect. Thus, by allowing for a mix of positive and negative interaction effects, the generalized model accommodates any pattern of conformity and differentiation.

Relative to the classical model, the generalized interactions framework admits a new set of target parameters that provide insight into how social interactions vary among different individuals. For example, the difference $J_{k_1 \ell} - J_{k_2 \ell}$ measures how the influence of type $\ell$ differs between agents of types $k_1$ and $k_2$. These terms can be used to construct an economic measure of polarization, $\delta_{k_1k_2} = J_{k_1k_1} + J_{k_2k_2} - J_{k_1k_2} - J_{k_2k_1}$, which quantifies the two types' combined preference for conforming to their own type rather than to the other type.

\begin{singlespace}
	\begin{table}[!htbp]
	\caption{Commonly Studied Economic Networks and Interaction Effects}
	\label{tab:network-examples}
	\centering
	\footnotesize
	\begin{tabularx}{\textwidth}{>{\centering\arraybackslash}m{0.14\textwidth} >{\centering\arraybackslash}m{0.18\textwidth} >{\centering\arraybackslash}m{0.28\textwidth} >{\centering\arraybackslash}m{0.30\textwidth}}
		\toprule
		& Illustration $(K = 7)$ & Structural Representation & Economic Examples \\
		\midrule
		\multicolumn{4}{l}{\textit{Panel A: Network Structures}} \\
		\midrule
		Tree
		& \includegraphics[width=0.12\textwidth]{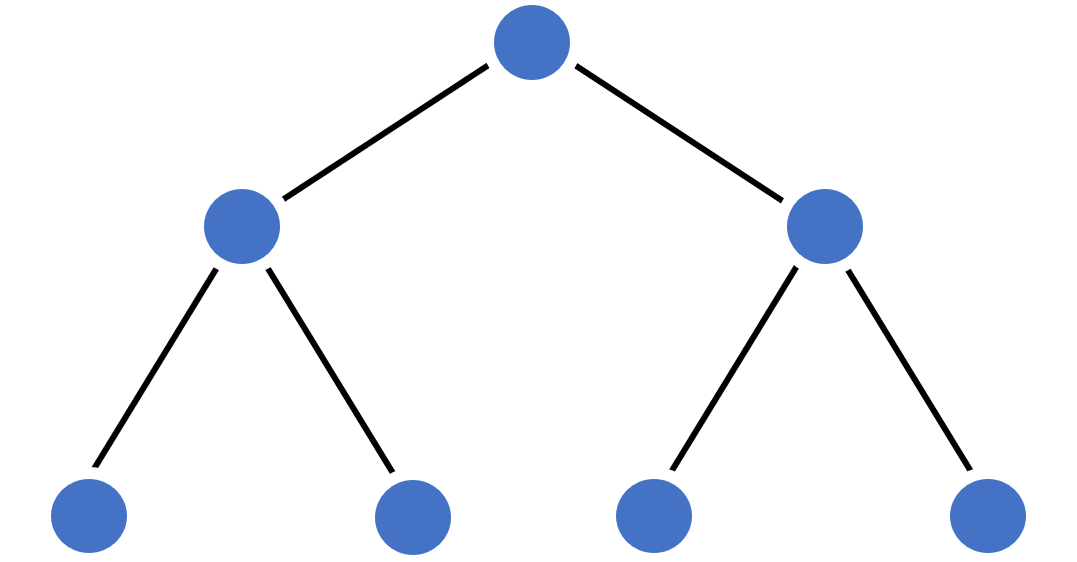}
		& \shortstack{
			$\sum\limits_{k<\ell}\hspace{-0.05cm}
			\mathbbm{1}\{|J_{k\ell}| \!+\! |J_{\ell k}| \!>\! 0\}
			\!=\! K \!-\! 1$ \\[-0.2cm]
			and all types connected
		}
		& hierarchy, vertical production, contagion, info diffusion\\
		\addlinespace
		Circle
		& \includegraphics[width=0.1\textwidth]{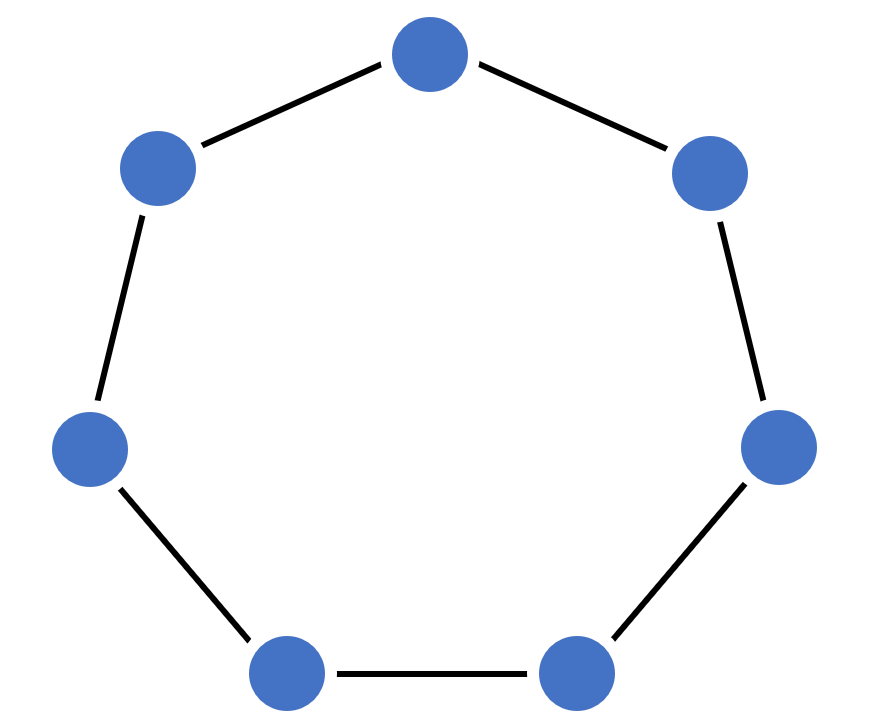}
		& $J_{k\ell} \neq 0$ for $\ell = k \pm 1$ (mod $K$)
		& \multicolumn{1}{>{\centering\arraybackslash}m{4.7cm}}{\shortstack{local interactions, \\ spatial competition}} \\
		\addlinespace
		Bipartite Network
		& \includegraphics[width=0.14\textwidth]{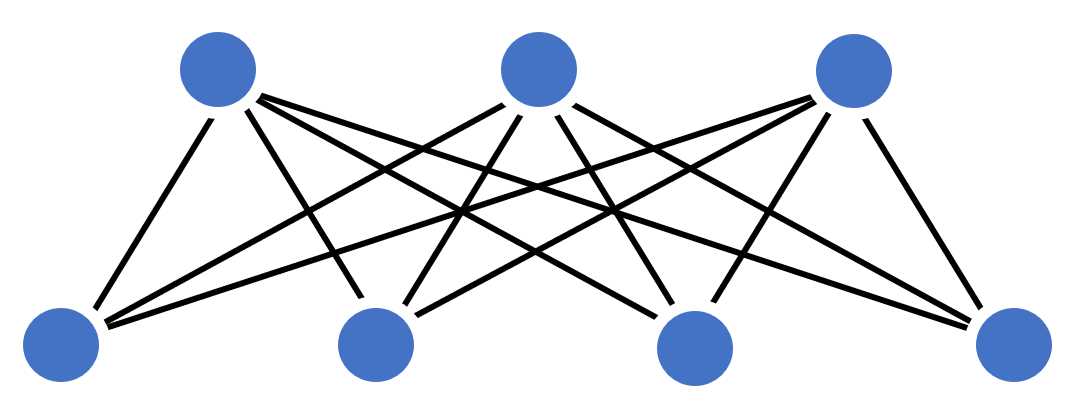}
		& $J_{k\ell} \neq 0$ if $k \in A$ and $\ell \in B$\newline for disjoint $A, B \subseteq \{1, \dots, K\}$
		& labor market, product market, college admission/matching  \\
		\addlinespace
		Complete Network
		& \includegraphics[width=0.09\textwidth]{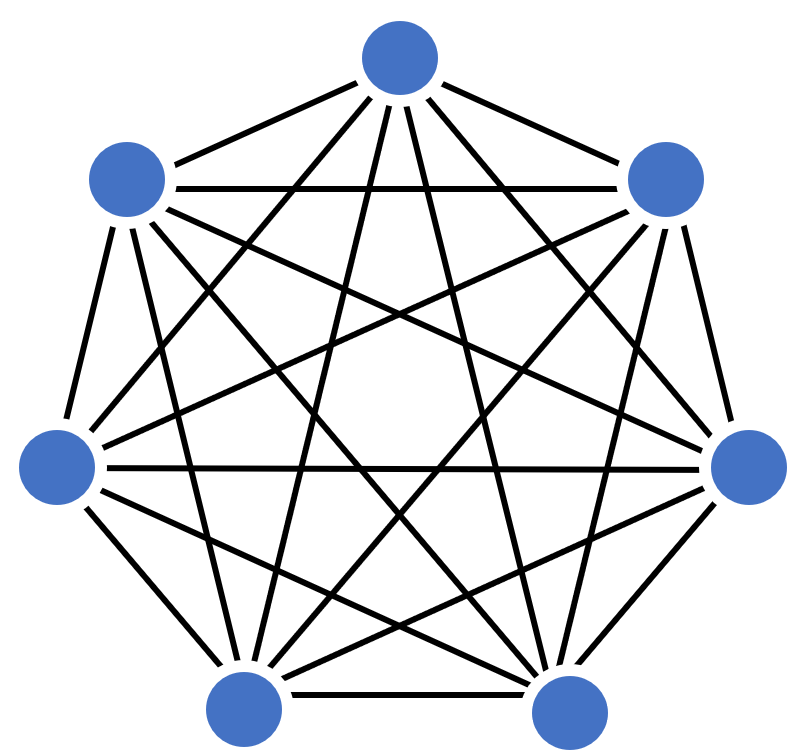}
		& $J_{k\ell} \neq 0$ for all $k,\,\ell$
		& classroom peer effects, public goods settings, oligopoly \\
		\midrule
		\multicolumn{4}{l}{\textit{Panel B: Modes of Social Influence}} \\
		\midrule
		\multicolumn{1}{>{\centering\arraybackslash}m{2.4cm}}{\shortstack{Strategic\\ Complementarity}} & \scalebox{1}{$\frac{\partial [U_i(1|k) - U_i(0|k)]}{\partial \mean_i(\bar\omega^\ell)} \geq 0$} & $J_{k\ell} \geq 0$ & conformity/imitation, social learning, technology adoption \\
		& & & \\
		\multicolumn{1}{>{\centering\arraybackslash}m{2.25cm}}{\shortstack{Strategic\\ Substitutability}} & \scalebox{1}{$\frac{\partial [U_i(1|k) - U_i(0|k)]}{\partial \mean_i(\bar\omega^\ell)} \leq 0$} & $J_{k\ell} \leq 0$ & free-riding, social distinction, segregation, polarization \\
		\bottomrule
	\end{tabularx}
	\begin{minipage}{\textwidth}
		\renewcommand{\baselinestretch}{0}\selectfont\setlength{\spaceskip}{0em}\footnotesize
		\emph{Notes.} This table lists examples of network structures and spillovers nested by the generalized interactions model.
	\end{minipage}
\end{table}
\end{singlespace}
\vspace{-0.5cm}

\subsection{Equilibrium under Noncooperative Decisionmaking}\label{s2.3}

\noindent
My analysis focuses on pure strategy Bayesian Nash equilibria in which agents act noncoo-\\peratively, that is, without coordination. An agent of type $k$ chooses $\omega_i = 1$ with probability:
\begin{equation}\label{eq:4}
	\pr(\omega_i = 1|k) = F_{\varepsilon|k}\Big(h_k + J_{kk} \mean_i(\bar\omega_{-i}^k) + \sum_{\ell \neq k} J_{k\ell} \mean_i(\bar\omega^\ell)\Big).
\end{equation}
Since $\omega_i$ takes values in the set $\{0, 1\}$, the expected action $\mean(\omega_i|k)$ also equals $\pr(\omega_i = 1|k)$.

Agents act with incomplete information. Specifically, an agent $i$ knows her own idiosyncratic taste $\varepsilon_i$ and the deterministic preference parameters ($h_k$, $\{J_{k\ell}\}_\ell$, $F_{\varepsilon|k}$) for each type, but not the idiosyncratic tastes $\varepsilon_j$ of other agents $j \neq i$. This information structure can be interpreted in two ways. First, it may be that agents do not see everyone else's choice, either because actions are taken simultaneously, such as during a test, or because the network is too large to track individual actions, such as in a school or village. Under this interpretation, $\varepsilon_i$ reflects the uncertainty that agents have about others' choices. A second interpretation of the information structure is that agents respond to social norms\textemdash the behavior of a ``typical'' agent of a given type\textemdash while ignoring idiosyncratic deviations from those norms. This interpretation is consistent with economic theories of social influence mechanisms \citep{akkr2000}, where norms are understood to reflect shared values of a social group.

An equilibrium is characterized by the expected average choices that are consistent with individually optimal decisions. I assume that agents hold consistent beliefs, so that they co-\\rrectly anticipate other agents' actions in expectation: $\mean_i(\omega_j|k) = \mean(\omega_j|k)$ for all $i \neq j$ and $k$. Letting $m^{k}$ denote $\mean(\bar \omega^k)$, an equilibrium $m = (m^{1}, \dots, m^{K})$ is defined as a solution to:
\begin{equation}\label{eq:5}
	m^{k} = F_{\varepsilon|k}\bigg(h_k + \sum_{\ell = 1}^K J_{k\ell} m^{\ell}\bigg), \quad\mbox{for } k = 1, \dots, K.
\end{equation}
Since $F_{\varepsilon|k}$ is continuous, Brouwer's fixed point theorem ensures that an equilibrium exists.

\subsection{Equilibrium Uniqueness and Local Stability}\label{s2.4}

\noindent
I now examine two key equilibrium properties of the model: uniqueness and local stability. Uniqueness asserts that, for a fixed set of parameters, there is one solution to the equilibrium system. Local stability concerns whether iteration on agents' best responses converges to an equilibrium. Formally, $m$ is \textit{locally stable} if it is a limiting solution to the dynamical system $m^k_t = F_{\varepsilon|k}(h_k + \sum_{\ell = 1}^K J_{k\ell} m^{\ell}_{t-1})$ for $k = 1, \scalebox{0.95}{$\dots$}, K$, where the starting value $m_0$ lies within some sufficiently small neighborhood of $m$. Meanwhile, $m$ is \textit{unstable} if starting values $m_0$ arbitrarily close to $m$ can produce an iterate $m_t$ outside some fixed neighborhood of $m$.

Both properties are fundamental to comparative statics. If a locally stable equilibrium is perturbed slightly, then behavioral outcomes would return to that equilibrium. Meanwhile, if an equilibrium is unstable, then nearby outcomes would diverge from it. In settings with multiple equilibria, a locally stable equilibrium is one that a researcher observes in practice, while an unstable equilibrium acts as a tipping point between stable equilibrium outcomes.

\subsubsection{Settings with a Unique Locally Stable Equilibrium}\label{s2.4.1}

\noindent
I begin by establishing conditions for the equilibrium to be unique and locally stable. Under the generalized interactions model, these conditions depend on both the network structure and the strength of social interactions relative to private preferences. I find that a unique and locally stable equilibrium exists when network ties are sparse, limiting behavioral feedback, or when interaction effects are weak, so private payoffs tend to dominate individual utility.

To obtain these conditions, I first derive the following property, which shows that, generi-\\cally, multiple equilibria always include at least one unstable equilibrium. Hence, to estab-\\lish uniqueness and local stability, it suffices to show that every equilibrium is locally stable.

\begin{property}
	For almost all distributions $\{F_{\varepsilon|k}\}_k$, the model has a finite, odd number of equilibria. If $d_s$ of these equilibria are locally stable, then at least $d_s-1$ must be unstable.
\end{property}

\noindent
To evaluate equilibrium stability, I examine the \textit{Jacobian matrix} of the equilibrium system:
\begin{equation}\label{eq:6}
	\mathbf{D}(m) = \boldsymbol{\beta}(m)\mathbf{J}.
\end{equation}
Here, $\boldsymbol{\beta}(m)$ is a diagonal matrix whose $k$th entry is $f_{\varepsilon|k}\big(h_k + \sum_{\ell = 1}^K J_{k\ell}m^{\ell}\big)$, the density of \\ type-$k$ agents at the decision threshold. Thus, the Jacobian weights each row of $\mathbf{J}$ by the responsiveness of the corresponding agent type. Fixing $\mathbf{J}$, choice probabilities respond more strongly to changes in others' behavior when a greater share of agents are near indifference.

Equilibrium stability is governed by the \textit{spectral radius} of the Jacobian matrix, defined as the largest of its eigenvalues in absolute value:\ $\rho(\scalebox{0.95}{$\mathbf{D}(m)$}) \!=\! \max\{|\lambda|: \mbox{$\lambda$ eigenvalue of } \scalebox{0.95}{$\mathbf{D}(m)$}\}$. This object measures the overall strength of equilibrium feedback. A value below one imp-\\lies the equilibrium is locally stable; a value above one implies the equilibrium is unstable.\footnote{The setting $\rho(\mathbf{D}(m)) = 1$ represents a knife-edge case that occurs with probability zero and has no particular economic relevance in the model. In such cases, equilibrium stability depends on higher-order derivatives of (\ref{eq:5}).}

\begin{lemma}
	If $\rho(\scalebox{0.95}{$\mathbf{D}(m)$}) < 1$, then $m$ is locally stable; if $\rho(\scalebox{0.95}{$\mathbf{D}(m)$}) > 1$, then $m$ is unstable. 
\end{lemma}

\noindent
Together, these results establish that the equilibrium is unique and locally stable whenever $\rho(\mathbf D(m))<1$ at all equilibria $m$. The following examples illustrate two ways this can occur.

\textit{Example 1 (One-Way Interactions).} Suppose social influence flows only in one direction, so that all feedback effects of the form $J_{k j_1} J_{j_1 j_2} \cdots J_{j_M k}$ equal zero. This structure arises in hierarchies or information networks, where influence flows from upstream to downstream agents but never returns to its source. In this case, behavior cannot become self-reinforcing. The model implies $\rho(\scalebox{0.95}{$\mathbf{D}(m)$}) = 0$ for all $m$, so the equilibrium is unique and locally stable.

\textit{Example 2 (Weak Interactions).} Suppose the interaction effects are weak relative to private payoffs.\ Specifically, for any type $k$, let $\sum_{\ell = 1}^K |J_{k\ell}| < S_k$, where $S_k$ is a threshold that depends on deterministic private payoffs and the distribution of idiosyncratic tastes.\footnote{In a model in which $f_{\varepsilon|k}$ is unimodal and $h_k = 0$ for all types $k$, the threshold $S_k$ simplifies to $S_k \!=\! 1 / f_{\varepsilon|k}(0)$.} This condition places a bound on the overall amount of social influence exerted on each type. It implies that $\rho(\mathbf D(m))<1$ at every $m$, so the equilibrium must be unique and locally stable.

\subsubsection{Settings with Multiple Locally Stable Equilibria}

\noindent
I next examine when the model has multiple locally stable equilibria. This question is well-understood in the classical model with uniformly positive interaction effects ($J_{k\ell} \geq 0$ for all \\[-0.75cm]

\noindent
$k$ and $\ell$). In this case, whenever social interactions are strong enough to generate an unstable equilibrium, the model contains multiple equilibria, at least two of which are locally stable.\footnote{The microeconomic foundations for models of complementarity are largely developed by Donald Topkis (see \cite{to1998} for a summary), \cite{vi1990}, and \cite{miro1990}. This work uses key results from lattice theory, e.g., Tarski's fixed point theorem, to study equilibrium behavior under uniform complementarities.}

This result does not generally extend to settings with negative interactions. Indeed, models with negative interactions may have a unique but unstable equilibrium, such that agents fail to converge on a single set of choices. To see why, suppose that $J_{kk} < 0$, so agents of type $k$ prefer to differentiate from their own type. When $-J_{kk}$ is sufficiently large, every equilibrium is unstable: if most of type $k$ selects the high action, its members would rather \\ switch to the low action, and vice versa. Similar instability can arise from conflicting cross-type interactions. For example, if $J_{k\ell} > 0$ and $J_{\ell k} < 0$, then type $k$ seeks to conform to type $\ell$, while type $\ell$ seeks to differentiate from type $k$. When these effects are sufficiently strong, each type's response can undermine the other's, ruling out any locally stable equilibrium.\footnote{This setting is strategically similar to a matching pennies game, which has no pure strategy Nash equilibrium. As another example, if $\mathbf{J}$ is symmetric and its eigenvalues have nonpositive real parts, then there is only one equ-\\ilibrium, which is locally stable when $\rho(\mathbf{D}(m)) \!<\! 1$ and unstable when $\rho(\mathbf{D}(m)) \!>\! 1$; see Appendix Lemma A.3.}

The contribution of my analysis is to pinpoint the source of global instability in models with negative interactions, and use this insight to classify which models can sustain locally stable equilibria, even with strong spillovers.\ The examples above suggest that global insta-\\bility occurs when agents strongly wish to deviate from other members of their own type.\ In that case, agents never settle on a choice, as they are always dissatisfied with the outcome.

To rule out global instability, it is crucial that social interactions do not lead to self-contr-\\adictory preferences. Specifically, after accounting for all feedback effects $\scalebox{0.95}{$J_{k j_1}J_{j_1 j_2}\cdots J_{j_M k}$}$, agents should not be strongly repelled by their own type. I formalize this condition below:

\begin{assumptionA}
	There is an invertible matrix $\mathbf{B}$ for which $\mathbf{B}\mathbf{D}(\cdot)\mathbf{B}^{-1} \geq 0$ entrywise.
\end{assumptionA}

\noindent
Assumption A states that,\hspace{-0.02cm} under a change of basis,\hspace{-0.02cm} $\mathbf{D}$ can be transformed into a nonnegative matrix. Through this transformation, I can reinterpret the underlying network as one where all interaction effects are positive, such that all agents conform to one another. This reinterpretation preserves the stability of equilibria, as determined by the spectral properties of $\mathbf{D}$. Leveraging this relationship, I can extend the stability properties of models with uniformly positive interactions to models that have negative interactions, yielding the following result:

\begin{property}
	Suppose Assumption A holds. Then, for almost all distributions $\{F_{\varepsilon|k}\}_k$, multiple locally stable equilibria exist if and only if the model has an unstable equilibrium.
\end{property}

\noindent
Property 2 is instrumental in the analysis. It shows that, under Assumption A, a locally sta-\\ble equilibrium generically exists, and multiple locally stable equilibria exist if interaction effects are strong enough to produce an unstable equilibrium.\ Assumption A holds in many economically relevant settings.\ I provide two examples below, with more in Appendix A.1. 

\textit{Example 1 (No Negative Feedback)}. Suppose all feedback effects are nonnegative, such that $J_{kj_1}J_{j_1j_2}\cdots J_{j_M k} \geq 0$ for every feedback path.\footnote{This condition represents a special case of Assumption A in which the change-of-basis matrix $\mathbf{B}$ is diagonal. \\ Specifically, for irreducible $\mathbf{J}$, it holds if and only if $\mathbf{B}\mathbf{D}(\cdot)\mathbf{B}^{-1}$ is nonnegative for some invertible, diagonal $\mathbf{B}$.} This condition holds in many familiar network structures. In a tree, it holds when all reciprocal interactions share the same sign: $J_{k\ell}, J_{\ell k} \geq 0$ or $J_{k\ell}, J_{\ell k} \leq 0$. In a circle, it holds when the network contains an even number of negative links. In a bipartite network, it holds when all interaction effects are either non\-negative or nonpositive. In a complete network, it holds when $\sgn(J_{km} J_{m\ell}) \!=\! \sgn(J_{k\ell})$ for every $k$, $\ell$, and $m$, so that complementarity and substitutability are transitive across types.

\textit{Example 2 (Dominant Own-type Conformity)}. Suppose $\mathbf{D}$ is symmetric and \textit{diagonally dominant}, so $D_{kk} \geq \sum_{\ell\neq k}|D_{k\ell}|$ for any $k$. In this setting, each type's response to its own average behavior is at least as strong as its combined absolute responses to the other types.

These examples show that Assumption A allows for negative interactions across a variety \\[-0.75cm]

\noindent
of network structures. It even permits negative feedback if the pressure to conform to one’s own type outweighs the influence of other types, as in Example 2. The key restriction is that \\[-0.75cm]

\noindent
the accumulation of feedback effects does not lead agents to be repelled by their own type.\footnote{Specifically, as shown in Appendix Lemma A.6, Assumption A requires the responsiveness-weighted sum of feedback effects $J_{kj_1}J_{j_1j_2}\cdots J_{j_Mk}$ to be nonnegative, although the individual feedback effects can be negative.}

\subsection{Comparisons of Social Welfare across Types}\label{s2.5}

\noindent
I conclude this section by characterizing social welfare across different equilibria.\hspace{-0.02cm} Note that, in this model, there is no strict Pareto ranking over equilibria because extreme realizations of the idiosyncratic taste $\varepsilon_i$ can dominate an agent's utility. Therefore, to measure welfare, \\ I analyze the expected utility over $\varepsilon_i$ for each agent type $k$ evaluated at any equilibrium $m$.
\begin{equation}\label{eq:7}
	\mathcal{W}_k(m) = \mean\bigg(\max_{\omega_i}\Big\{h_k\omega_i + \eta_k + \sum_{\ell = 1}^K J_{k\ell}\omega_im^{\ell} + \varepsilon_i\omega_i + \xi_i\Big\}\bigg).
\end{equation}
To facilitate welfare comparisons, I rescale choices so that $\omega_i \in \{-1, 1\}$. This change does not impact behavioral outcomes; yet, it implies that, without private incentives ($h_k = \varepsilon_i = 0$ for all $i, k$), individual utility is the same whether everyone selects the high action or everyone selects the low action. Therefore, this rescaling ensures there is no negative externality \\inherent to an equilibrium when all agents are fully ambivalent between the two choices.\footnote{Welfare analysis in network-based models is especially sensitive to how utility is specified, even if different specifications yield the same choices. The reason is that agents' payoffs depend on underlying drivers of spillover effects, such as conformity versus social learning, not just the effects themselves; see \cite{bhduka2023}.}

The next property characterizes how social welfare outcomes relate to specific patterns of individual behavior. It shows that, for each type, social welfare is highest at the equilibrium in which agents are most unified in their choices. Also, if agents are privately biased toward one specific action, social welfare tends to be highest when most agents choose that action.

\begin{property}
	Let $\mathcal{M}$ denote the set of equilibria, and define $m^* = \argmax_{m \in \mathcal{M}}\mathcal{W}_k(m)$. Then $m^* = \argmax_{m \in \mathcal{M}}|\mean(\bar{\omega}^k)|$. Moreover, there always exists a threshold $T_k \in \mathbb{R}$ where:
	$$\small m^* =
	\begin{cases}
		\argmax_{m \in \mathcal{M}} \mean(\bar{\omega}^k), & \text{if } h_k > T_k \\[4pt]
		\argmin_{m \in \mathcal{M}} \mean(\bar{\omega}^k), & \text{if } h_k < T_k.
	\end{cases}$$
\end{property}

\noindent
In the presence of multiple equilibria, a key consideration for policy is how social welfare rankings differ for different types of agents. Namely, is there a single favorable equilibrium that maximizes welfare for all types, or are there tradeoffs such that improving welfare for one type would always reduce it for another? This question matters for whether policy can resolve coordination failures by shifting an economy toward a more favorable equilibrium.

In the classical model where all the interaction effects are positive, the set of equilibria is well-ordered, forming a complete lattice. From a welfare perspective, this property implies that there is always one equilibrium that maximizes social welfare for all types, provided that all types are privately biased toward the same action.\footnote{Specifically, if $h_k > 0$ for all $k$, then $\argmax_{m \in \mathcal{M}} \mathcal{W}_k(m) = \argmax_{m \in \mathcal{M}} \mathcal{W}_\ell(m)$ for any types $k$ and $\ell$.}\ As the next result demonstrates, this relationship does not generally apply to environments with negative interaction effects.

\begin{property}
	Suppose $\mathbf{J}$ is irreducible with no negative feedback. For any types $k$ and $\ell$:\\[-0.55cm]
	\begin{enumerate}[label=(\roman*)]
		\item If $J_{k \ell} > 0$, then $\scalebox{0.95}{$\mean(\bar{\omega}^k)$}$ and $\scalebox{0.95}{$\mean(\bar{\omega}^\ell)$}$ are maximized (minimized) at the same equilibrium.
		\item If $J_{k \ell} < 0$, then $\scalebox{0.95}{$\mean(\bar{\omega}^k)$}$ is maximized at the same equilibrium where $\scalebox{0.95}{$\mean(\bar{\omega}^\ell)$}$ is minimized.
	\end{enumerate}
\end{property}

\noindent
Property 4 reveals how negative interactions can introduce tradeoffs in improving welfare. For example, consider two types of agents, $k$ and $\ell$, each deciding whether to take a medical treatment, such as a vaccine. Suppose both types, in isolation, would benefit from treatment ($h_k, h_\ell > 0$), but each one prefers to differentiate from the other type ($J_{k\ell}, J_{\ell k} < 0$). If these interaction effects are strong enough, multiple equilibria exist:\ in one, a majority of type $k$ is treated; in another, a majority of type $\ell$ is treated. Each type prefers the equilibrium where a majority of its members are treated, but both cannot achieve this outcome simultaneously.

In this setting, neither equilibrium outcome is broadly favorable. Hence, it is not obvious when a policy intervention can achieve net welfare gain. Any attempt to increase treatment uptake among one type imposes a negative externality by reducing uptake among the other.

\section{Identification and Estimation of Endogenous Interaction Effects}\label{s3}

\noindent
I now show how to recover features of the interaction matrix $\mathbf{J}$ using individual choice data. The analysis proceeds in two steps. First, I prove identification when the expected choices $\mean(\omega_i|k)$ are observed. This allows me to focus on separating endogenous interactions from unobserved network effects. Second, I extend the analysis to the case where $\mean(\omega_i|k)$ is unobserved, but can be approximated by the average choice $\bar\omega^k$ in a finite network. I show how to construct internal instruments to correct for the resulting measurement error bias. I then derive a linear IV estimator for the interaction effects, which performs well in simulations.

\subsection{Econometric Specification}\label{s3.1}

\noindent
Consider a setting with many networks $n = 1, \dots, N$, where agents interact only with others in their network. In practice, these networks could be schools, workplaces, neighborhoods, or villages. Suppose that the researcher has data on individual choices $\omega_i$, either for the entire population or for a random sample of agents drawn from each network. In addition, the researcher knows the network $n$ to which each agent belongs, as well as each agent's type $k$. \\[-0.55cm]

I allow agents' private utilities to vary systematically across networks, replacing the det-\\erministic private payoff $h_k$ with $h_k + \alpha_n$. Hence, for any agent $i$ of type $k$ in network $n$, private utility depends on individual-specific factors $\varepsilon_i$, type-specific factors $h_k$, and conte-\\xtual network effects $\alpha_n$, all of which are unobserved by the researcher. Agent $i$'s action is:
\begin{equation}\label{eq:8}
	\omega_i = \mathbbm{1}\bigg\{h_k + \alpha_n + J_{kk} \mean_i(\bar\omega_{n, -i}^k) + \sum_{\ell \neq k} J_{k\ell} \mean_i(\bar\omega_n^\ell) + \varepsilon_i \geq 0\bigg\},
\end{equation}
where $h_k$, $\alpha_n$, $\varepsilon_i$, and $\{J_{k\ell}\}_{k, \ell}$ are unknown parameters. The model may also be extended to include observed covariates or prices at the agent- or network-level;\hspace{-0.03cm} however, as these do \\ not meaningfully impact the results, I omit them for now and give details in Appendix B.1.

I maintain two assumptions about the error structure. First, I assume the random payoffs $\varepsilon_i$ and $\varepsilon_j$ are independent for any agents $i$ and $j$ within and across networks. Hence, there is no covariation in the error terms. Second, I assume that, given an agent's type $k$, the error term $\varepsilon_i$ is independent of $\alpha_n$. This condition rules out self-selection into networks based on unobserved idiosyncratic payoffs.\footnote{This restriction is stronger than needed for identification. As shown by \cite{ma1988} and also discussed by \cite{ho2009}, full conditional independence can be replaced by a weaker quantile independence assumption.} Both of these assumptions are standard in the literature on social interactions, and they can also be relaxed to allow for selection on observables.\footnote{Let $\pr(\varepsilon_i \leq z|W_n, \alpha_n, k) = F_{\varepsilon|W_n, k}(z)$, where $W_n$ is an observable network characteristic. Even if $F_{\varepsilon|W_n, k}$ does not equal $F_{\varepsilon|k}$, the identification arguments presented still follow by comparing networks with the same $W_n$.}

\begin{assumptionB1}
	(i) $\{\varepsilon_i\}_i$ are mutually independent; (ii) $\pr(\varepsilon_i \leq z|k, \alpha_n) = F_{\varepsilon|k}(z)$.
\end{assumptionB1}

\noindent
As before, all agents act with incomplete information and consistent beliefs. In equilibrium:
\begin{equation}\label{eq:9}
	m^{k}_n = F_{\varepsilon|k}\bigg(h_k + \alpha_n + \sum_{\ell = 1}^K J_{k\ell}m^{\ell}_n\bigg), \quad\mbox{for } k = 1, \dots, K \mbox{ and } n = 1, \dots, N.
\end{equation}
As shown in Section 2, sufficiently strong interaction effects can generate multiple equilibria. In such cases, I assume that agents correctly anticipate which equilibrium is realized in their network.\footnote{This assumption can be rationalized by noting that, whenever the realized equilibrium is locally stable, agents can compute it independently through fixed-point iteration on (\ref{eq:9}), as long as they begin with a common prior $m_0$.} Equilibrium multiplicity does not, by itself, interfere with identification: a \\ given parameter vector may admit multiple equilibria even while the observed choice dist-\\ribution uniquely determines the parameters of interest. My identification analysis uses the restrictions that the equilibrium conditions (\ref{eq:9}) impose, without ever directly solving for eq-\\uilibrium outcomes.\ In doing so, I demonstrate that the parameters of interest are identified and may be consistently estimated regardless of which equilibrium is realized in the data.\footnote{This contrasts with the two-step approaches used in the empirical industrial organization literature, where id-\\entification relies on computing equilibrium conditions. See \cite{bhduka2023}, Section 5.5, for discussion.}

In most cases, it is natural to assume that the equilibrium is locally stable (see Section 2.4 for conditions); otherwise, it is unlikely to be realized in practice. Yet, although stability can matter for interpretation, the identification analysis does not explicitly rely on this property.

\subsection{Identification with Known Expected Average Choices}\label{s3.2}

\noindent
Suppose that researchers observe the expected average choices $m^k_n$. Then, the main barrier to identification is the unobserved network effect $\alpha_n$, which obscures the effect of social in-\\[-0.75cm]

\noindent
teractions on group-level outcomes. This issue is noted by \cite{brdu2007}, Prop.\ 2, who show that network unobservables prevent point identification of interaction effects.

I propose a partial identification strategy to recover key features of $\mathbf{J}$ without placing any \\ added assumptions on $\alpha_n$. This method involves differencing-out $\alpha_n$ among different types of agents in the same network. The intuition behind this strategy is captured in the following lemma. I give the proof alongside the result so the approach can be more clearly interpreted.

\begin{lemma}[Sufficiency Claim]
	$\mean(\omega_i|k, \alpha_n, \{m^{\ell}_n\}_{\ell}) = \mean(\omega_i|k, \{m^{\ell}_n\}_{\ell})$ for all $n$ and $k$.
\end{lemma}

\noindent
\emph{Proof.} Fix any $\tilde k \neq k$. Since $F_{\varepsilon|\tilde k}$ is strictly increasing, equation (\ref{eq:9}) can be inverted so that:\\[-0.55cm]
\begin{equation}\label{eq:10} \small
	\alpha_n = F_{\varepsilon|\tilde k}^{-1}\big(m^{\tilde k}_n\big) - h_{\tilde k} - \sum_{\ell = 1}^K J_{\tilde k\ell}m^{\ell}_n.
\end{equation}
Substituting this expression into the definition of $\mean(\omega_i|k, \alpha_n, \{m^{\ell}_n\}_\ell)$ yields the equation:\\[-0.55cm]
\begin{equation}\label{eq:11} \small
	\mean(\omega_i|k, \alpha_n, \{m^{\ell}_n\}_\ell) = F_{\varepsilon|k}\bigg(h_{k} - h_{\tilde k} + \sum_{\ell = 1}^K (J_{k\ell} - J_{\tilde k\ell})m^{\ell}_n + F_{\varepsilon|\tilde k}^{-1}\big(m^{\tilde k}_n\big)\bigg).
\end{equation}
The network effect cancels. Thus, $\mean(\omega_i|k, \alpha_n, \{m^{\ell}_n\}_\ell)$ is a deterministic function of $\{m^{\ell}_n\}_\ell$. \\\qed

\noindent
This lemma shows how observable differences among agents in the same network may be used to control for contextual effects. Consider two agents $i$ and $j$ of different types, $k_1$ and $k_2$, who both belong to network $n$. Because they reside in the same network, all contextual factors affecting $i$ also apply to $j$. Any difference in $i$ and $j$'s choices is due to differences in type-specific preferences,\hspace{-0.01cm} $(\scalebox{0.95}{$h_{k_1}, \{J_{{k_1}\ell}\}_\ell$})$ and $(\scalebox{0.95}{$h_{k_2}, \{J_{{k_2}\ell}\}_\ell$})$,\hspace{-0.01cm} or idiosyncratic tastes,\hspace{-0.01cm} $\varepsilon_i$ and $\varepsilon_j$. \\[-0.55cm]

This framework offers a natural panel structure, allowing me to control for network-level \\ factors by comparing the outcomes of different agents in the same network. This approach is conceptually similar to the method of differencing-out incidental parameters in panel data models; see, for example, \cite{ch1980}.\footnote{\cite{bo2012} proposes a systematic approach to obtain moment restrictions that are free of fixed effects via functional differencing. Additionally, see \cite{boda2023} and \cite{da2025} for recent advances.} By removing $\alpha_n$, I achieve identification of the differences $h_{k_1} - h_{k_2}$ and $\{J_{k_1\ell} - J_{k_2\ell}\}_{\ell}$, which capture how agents differ in their private incentives and responses to others' behavior. At the same time, this approach leaves the levels of $h_k$ and $J_{k\ell}$, as well as the network effects themselves, unidentified, which, as \\ \cite{ch1984} discusses, may limit the range of probability statements one can make.

I present two versions of the identification result. First, I establish semiparametric identi-\\fication when the error distributions $F_{\varepsilon\mid k}$ are known. Second, I establish nonparametric ide-\\ntification when $F_{\varepsilon\mid k}$ are unknown. While the latter result is more flexible, it requires more variation in the data. In both versions, identification is achieved without $\alpha_n$ being observed.

\vspace{0.3cm}
\begin{flushleft}
	\underline{\textit{Conditions for Semiparametric Identification}}
\end{flushleft}

\noindent
Before stating the semiparametric identification result, I present the following assumption.

\begin{assumptionB2}
   	There is sufficient variation across networks for $\scalebox{0.95}{$\operatorname{rank}(\var(m_n)) \!=\! K$}$.
\end{assumptionB2}

\noindent
This rank condition guarantees that there is enough variation in $\alpha_n$ within the population to uncover the nonlinear relationship between $m_n$ and $\alpha_n$. This nonlinearity is a core feature of discrete choice models, and it marks a key distinction from linear simultaneous equations models, where identification is typically impeded by the \textit{reflection problem} \citep{ma1993}. In linear models, $m_n$ and $\alpha_n$ are linearly dependent. Thus, even if $\alpha_n$ is fully determined by observables, one cannot isolate social interactions without additional exclusion restrictions.

By contrast, the reflection problem does not arise in discrete choice settings. As shown by \cite{brdu2001, brdu2007}, discrete choices generate a nonlinear relationship between \\[-0.75cm]

\noindent
$m_n$ and $\alpha_n$,\hspace{-0.03cm} which acts as a built-in exclusion restriction.\hspace{-0.03cm} This nonlinearity allows me to iso-\\[-0.75cm]

\noindent
late endogenous interactions from contextual network effects.\hspace{-0.03cm} Neither continuous nor unbo-\\unded support of $\alpha_n$ is necessary. It suffices that the covariance matrix of $m_n$ has full rank. 

\begin{theorem}
	Suppose Assumptions B.1 and B.2 hold, and assume $m^{k}_n$ is observed for all networks $n$ and agent types $k$. If the distribution functions $\{F_{\varepsilon|k}\}_k$ are identified, then:
	\begin{itemize}
		\item[(i)] Without further assumptions, $\{h_{k_1} - h_{k_2}\}_{k_1, k_2}$ and $\{J_{k_1\ell} - J_{k_2\ell}\}_{k_1, k_2, \ell}$ are identified.
		\item[(ii)] If $\alpha_n = W_n'd$ for some observed vector $W_n$, then $d$, $\{h_k\}_k$, and $\{J_{k\ell}\}_{k, \ell}$ are identified.
	\end{itemize}
\end{theorem}

\noindent
This theorem has two parts. Part \textit{(i)} draws on Lemma 2 to show that the model is partially identified even if $\alpha_n$ is unobserved. In this case, the differences $h_{k_1}-h_{k_2}$ and $\{J_{k_1\ell}-J_{k_2\ell}\}_\ell$ are identified for all $(k_1, k_2)$.\footnote{If the error distributions are known, then $J_{k\ell}$ is identified up-to-location by normalizing the own-type effects $J_{kk}$ to one. Similarly, private payoffs $h_{k}$ are identified up-to-location by setting $h_{k^*} = 0$ for a reference type $k^*$.}\ This result is new to the literature and is the theorem's main \\ contribution. Part \textit{(ii)} covers a special case, studied by \cite{brdu2001} and others, \\ where $\alpha_n$ is a fixed, linear function of observables.\hspace{-0.03cm} In that case, the model is point identified.

\vspace{0.3cm}
\begin{flushleft}
	\underline{\textit{Conditions for Nonparametric Identification}}
\end{flushleft}

\noindent
For nonparametric identification, I require an extra assumption: there must be an exogenous individual-level covariate with continuous, unbounded support. Variation in this covariate traces out the error distributions $F_{\varepsilon\mid k}$, which I then use to recover the remaining parameters.

To set ideas, I first modify the choice equation to include exogenous covariates, $X_i \in \mathbb{R}^r$. Let $\omega_i = \mathbbm{1}\{X_i'c + h_k + \alpha_n + \sum_{\ell = 1}^K J_{k\ell}m^{\ell}_n + \varepsilon_i \geq 0\}$, with $\pr(\varepsilon_i \leq z|X_i, k, \alpha_n) = F_{\varepsilon|k}(z)$ and $\pr(X_i \leq x|k, \alpha_n) = \pr(X_i \leq x|k)$. Suppose that $X_i$ satisfies the following assumption:

\begin{assumptionB3}
	For every $k$, $\mbox{supp}(X|k)$ is not contained in a proper affine subspace of $\mathbb{R}^r$; also, there is a component $x_j$ of $X$\textemdash with a nonzero coefficient $c_j$\textemdash such that, for almost all $x_{-j|k}$, the distribution of $x_{j|k}$ given $x_{-j|k}$ has positive density everywhere on $\mathbb{R}$.
\end{assumptionB3}

\noindent
This assumption, adapted from \cite{ma1988}, ensures that there is sufficient variation in at least one individual-level covariate to recover the error distributions $F_{\varepsilon|k}$ up-to-scale, even when the functional forms are unknown. With this assumption, I obtain the following result: \\[-0.55cm]

\begin{theorem}
	Suppose Assumptions B.1, B.2, and B.3 hold, and assume $m^{k}_n$ is observed for all networks $n$ and agent types $k$. Then $\{F_{\varepsilon|k}\}_{k}$ and $c$ are identified up-to-scale. Also:
	\begin{itemize}
		\item[(i)] Assuming nothing else, $(\scalebox{0.95}{$\{h_{k_1} - h_{k_2}\}_{k_1, k_2}, \{J_{k_1\ell} - J_{k_2\ell}\}_{k_1, k_2, \ell}$})$ is identified up-to-scale.
		\item[(ii)] If $\alpha_n = W_n'd$ for some observed $W_n$, then $(\scalebox{0.95}{$d, \{h_k\}_k, \{J_{k\ell}\}_{k, \ell}$})$ is identified up-to-scale.
	\end{itemize}
\end{theorem}

\subsection{Identification with Unknown Expected Average Choices}\label{s3.3}

\noindent
In the data, the researcher observes individual choices, but not the equilibrium expectations $m_n^k$ that enter agents' decisions.\ In particular, agents respond to the expected average beh-\\avior of each type, which need not equal the average choice $\bar\omega_n^k$ realized in a finite network.

This distinction between expected and realized behavior rests on substantive assumptions \\[-0.75cm]

\noindent
about the information structure. As discussed in Section 2.3, agents know their own private \\ idiosyncratic tastes, but do not observe those of their peers. When making decisions, they respond to their expectations of each type's average behavior, with beliefs that are assumed to be consistent in equilibrium. This framework is motivated by economic settings in which agents make choices with private information, without being able to perfectly anticipate ev-\\eryone else's contemporaneous actions. It is standard in the literature following \cite{brdu2001}, and it simplifies the identification analysis by yielding an equilibrium characterized in terms of choice probabilities.\footnote{See \cite{pa2013} for a careful review of how information assumptions affect identification of these models.} Nevertheless, this assumption need not have a principled basis in a given setting; I discuss its suitability for each application in Section 4.

For researchers, the assumed information structure creates a measurement problem:\ ag-\\ents respond to expectations that can only be measured imperfectly from observed choices. In particular, with finite networks, the model implies that $\bar\omega_n^k = m_n^k + u_n^k$, where $u_n^k$ denotes measurement error.\ If the researcher substitutes $\bar\omega_n^k$ for $m_n^k$ without accounting for this error, the preceding identification arguments no longer directly apply. To see why, rewrite (\ref{eq:8}) as:
\begin{equation}\label{eq:12}
	\omega_i = \mathbbm{1}\bigg\{h_k + \alpha_n + \sum_{\ell=1}^K J_{k\ell}\bar\omega_n^\ell + \tilde\varepsilon_i \geq 0\bigg\},
	\quad\mbox{where}\quad \tilde\varepsilon_i = \varepsilon_i - \sum_{\ell=1}^K J_{k\ell}u_n^\ell.
\end{equation}
In this equation, the observed average choices $\{\bar\omega^\ell_n\}_\ell$ are endogenous because they are correlated with the error term $\tilde\varepsilon_i$. Even if agent $i$'s choice $\omega_i$ were excluded when constructing $\bar\omega^k_n$, it would still be that $\cov(\bar\omega^\ell_n, \tilde\varepsilon_i) \!=\! -J_{k\ell} \!\times\! \var(u_n^\ell) \!\neq\! 0$ for any type $\ell$.\ Consequently, tre-\\[-0.75cm]

\noindent
ating realized averages as equilibrium expectations can bias estimates of interaction effects.

The remainder of this subsection shows how to correct for measurement error under the incomplete information structure.\ In particular, I show how to recover the parameters of in-\\terest from choice data when networks are finite and agents' expectations are unobserved.

\enlargethispage{0.3\baselineskip}
\begin{flushleft}
	\underline{\textit{Constructing Instruments to Correct for Measurement Error}}
\end{flushleft}

\noindent
I propose a way to correct for measurement error using internal instruments. This approach involves two steps.\hspace{-0.03cm} First, in each network $n$ and type $k$, I randomly partition agents into two \\[-0.75cm]

\noindent
subsamples, $a$ and $b$.\ Second, I compute the average actions, $\smash{\bar\omega^k_{n, a}}$ and $\smash{\bar\omega^k_{n, b}}$, in each subsam-\\ple and estimate the model via IV, treating $\smash{\bar\omega^k_{n, a}}$ as the regressor and $\smash{\bar\omega^k_{n, b}}$ as the instrument.

By construction, $\bar\omega^k_{n, a}$ and $\bar\omega^k_{n, b}$ are noisy proxies for $m^{k}_n$ with sampling errors that are ind-\\[-0.75cm]

\noindent
ependent.\ This independence follows from the assumed information structure, under which agents respond to expected rather than realized choices. Therefore, either proxy can instru-\\ment for the other, yielding consistent estimates as the number of networks grows, for fixed \\ network sizes.\hspace{-0.03cm} Before stating the identification result,\hspace{-0.03cm} I will first motivate it with an example. 

\noindent
\textit{Example (Brock-Durlauf Model).} Consider the model studied by \cite{brdu2001}:
\begin{equation}\label{eq:13}
	m_n = \tanh\big(h + W_n'd + Jm_n\big), \quad\mbox{for } n = 1, \dots, N,
\end{equation}
where $W_n$ is an observed vector and $h$, $d$, and $J$ are unknown parameters.\footnote{The authors allow for individual-level covariates $X_i$, which I exclude in this example for notational simplicity.}\hspace{-0.05cm} When taking the model to data, it is common practice to replace equilibrium expectations $m_n$ with observed averages $\bar\omega_n$, thus treating a noisy finite-network proxy as if it were observed without error. 

Without accounting for measurement error, standard estimators, such as maximum likeli-\\hood or OLS, deliver biased estimates. To see how, consider an OLS estimand based on the observed means $\bar\omega_n$. Defining $u_n = \bar\omega_n - m_n$ and $v_n = \tanh^{-1}(\bar\omega_n) - \tanh^{-1}(m_n)$, I write:
\begin{equation}\label{eq:14}
	\tanh^{-1}(\bar\omega_n) = h + W_n'd + J\bar\omega_n + \tilde\xi_n, \quad\mbox{where } \tilde\xi_n = - Ju_n + v_n.
\end{equation}
To ease notation, let $\tilde m_n$, $\tilde{\bar\omega}_n$, and $\tilde Y_n$ denote the residuals from a least squares regression of $m_n$, $\bar\omega_n$, and $\tanh^{-1}(\bar\omega_n)$, respectively, on the vector $(1, W_n')$. The OLS estimand for $J$ is:
\begin{equation}\label{eq:15}
	J^{\mbox{\scriptsize OLS}} = \frac{\cov(\tilde{\bar\omega}_n, \tilde Y_n)}{\var(\tilde{\bar\omega}_n)} = J \times \frac{\var(\tilde m_n)}{\var(\tilde m_n) + \var(u_n)} + \frac{\cov(\tilde{\bar\omega}_n, v_n)}{\var(\tilde m_n) + \var(u_n)}.
\end{equation}
This estimand suffers from two sources of bias, since both the regressor and the outcome are \\[-0.75cm]

\noindent
measured with error.\ In finite networks, OLS therefore fails to recover the interaction effect. 

As illustrated in Figure 1, the bias can be substantial even for very large networks. When setting $J = 0.5$, OLS overstates the interaction effect by a factor of five in networks with 100 agents and by a factor of two in networks with 500 agents. The bias only becomes small when network sizes reach 5,000 agents\textemdash a scale rarely attainable in empirical applications.

To overcome this issue, I propose IV estimation: splitting each network into two parts, $a$ and $b$, and using the average choice in part $b$ as an instrument for the average choice in part $a$. By construction, the instrument $\bar\omega_{n, b}$ is uncorrelated with measurement error in $\bar\omega_{n, a}$, and instrument relevance holds, as $\cov(\bar\omega_{n, a}, \bar\omega_{n, b}) \!=\! \var(m_n) \!\neq\! 0$.  By standard arguments, $J$ is identified by: (1) regressing $\bar\omega_{n, a}$ on $(1, W_n', \bar\omega_{n, b})$ in the first stage to obtain fitted values $\widehat{\bar\omega}_{n,a}$ and (2) regressing the outcome $Y_n$ on $(1, W_n', \widehat{\bar\omega}_{n,a})$ in the second stage. To guarantee that the estimator jointly identifies $J$, $d$, and $h$, I address a further complication: owing to \\ the nonlinearity of $\tanh^{-1}$, the measurement error in $\tanh^{-1}(\bar\omega_n)$ is nonzero in expectation. This issue can be corrected by defining $Y_n$ as a linear combination of $\tanh^{-1}(\bar\omega_{n, a})$ across subsamples, exploiting the known dependence between $\mean(v_{n, a})$ and the subsample size. As demonstrated below, the procedure jointly recovers $J$, $d$, and $h$ even with finite networks.

\vspace{-1.5cm}
\begin{figure}[!htbp]
	\begin{center}
		\includegraphics[width=0.75\linewidth]{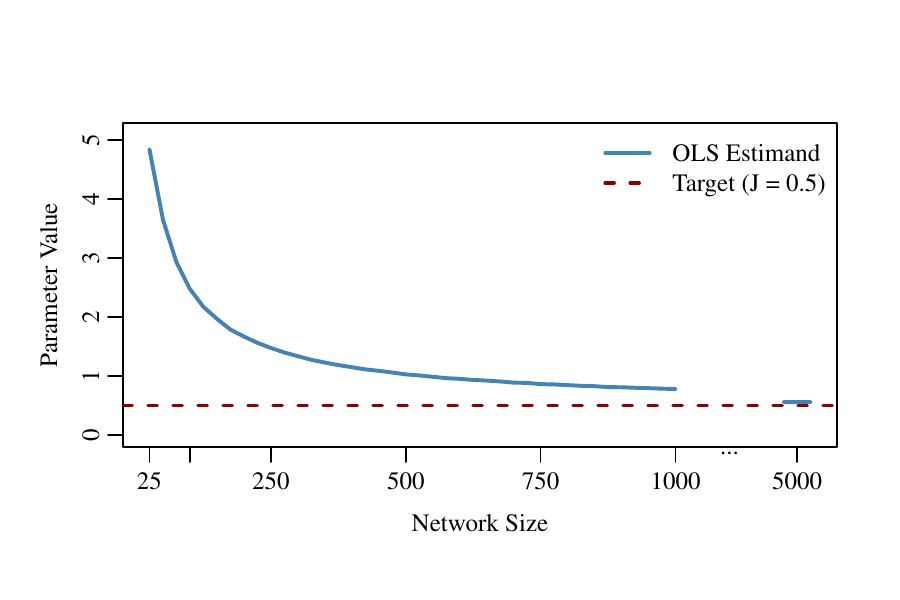}\vspace{-1cm}
		\caption{OLS estimand for $J$ across network sizes ($J = 0.5$, $h = 1$, $d = 1$); same design as in Table 2(A).}
		\label{fig:ols-bias}
	\end{center}
\end{figure}

\begin{flushleft}
	\underline{\textit{An IV Estimand for Endogenous Interaction Effects}}
\end{flushleft}

\noindent
To formalize the identification result, first define $\bar\omega^k_{n, a}$ and $\bar\omega^k_{n, b}$ to be the average choices in \\ each randomly selected subset of network $n$ and type $k$. Throughout the analysis, I define $\bar\omega^k_{n, a}$ as the endogenous regressor and $\bar\omega^k_{n, b}$ as the excluded instrument. Define the outcomes:
\begin{align}\label{eq:16}
	Y_n^k = \Omega_n^kF_{\varepsilon|k}^{-1}(\bar\omega^{k}_{n, a}) + (1-\Omega_n^k)F_{\varepsilon|k}^{-1}(\bar\omega^{k}_{n, a'}),
\end{align}
where $a'$ is a strict subset of $a$, and $\Omega_n^k$ is a known weight that depends on the sizes of $a$ and $a'$, denoted by $\scalebox{0.95}{$I_{n, a}^k$}$ and $\scalebox{0.95}{$I_{n, a'}^k$}$, respectively. To define this weight, I consider a sequence of partitioning rules $p$, where $\kappa(p,n)\in\{a,a'\}$ indicates which subset the rule selects in netw-\\ork $n$. For the first $P$ rules, define the matrix $\mathbf E_{k,P}$ by $\smash{[\mathbf E_{k,P}]_{i,j}=\mean[(I_{n,\kappa(j,n)}^k)^{1-i}]}$, and set \\ $\smash{\Omega_n^k \!=\! \lim_{P \to \infty}\sum_{p=1}^{P}[\mathbf E_{k,P}^{-1}]_{p,1}\mathbbm{1}\{\kappa(p,n)=a\}}$.\hspace{-0.05cm} The Appendix gives details on rules and sam-\\pling conditions under which the matrices are nonsingular and the resulting outcomes yield valid IV moments. While the formal identification argument takes the limit as $P \!\to\! \infty$, I find that in practice, even $P = 1$ or $P = 2$ suffices to deliver estimates with effectively zero bias. \\[-0.55cm]

\begin{theorem}
	Suppose Assumptions B.1 and B.2 hold and $\{F_{\varepsilon|k}\}_k$ are identified. Then:
	\begin{itemize}
		\item[\textit{(i)}] $\beta_{k_1, k_2} = [h_{k_1} - h_{k_2}, J_{k_1 1} - J_{k_2 1}, ..., J_{k_1 K} - J_{k_2 K}]'$ is identified for any $(k_1, k_2)$. Also:
		\begin{equation}\label{eq:17}
			\beta_{k_1, k_2} =  \mean(Z_n X_n')^{-1}\mean(Z_n [Y_n^{k_1} - Y_n^{k_2}]), \quad\mbox{where } \small
			\begin{cases}
				X_n &\hspace{-0.3cm}= (1, \bar\omega^1_{n, a}, \dots, \bar\omega^K_{n, a})' \\
				Z_n &\hspace{-0.3cm}= (1, \bar\omega^1_{n, b}, \dots, \bar\omega^K_{n, b})'.
			\end{cases}
		\end{equation}
		\item[(ii)] If $\alpha_n = W_n'd$ for observed $W_n$, $\beta_k = [h_k, d', J_{k1}, ..., J_{kK}]'$ is identified for any $k$. Also:
		\begin{equation}\label{eq:18}
			\beta_{k} =  \mean(Z_n X_n')^{-1}\mean(Z_n Y_n^k), \quad\mbox{where } \small
			\begin{cases}
				X_n &\hspace{-0.3cm}= (1, W_n', \bar\omega^1_{n, a}, \dots, \bar\omega^K_{n, a})' \\
				Z_n &\hspace{-0.3cm}= (1, W_n', \bar\omega^1_{n, b}, \dots, \bar\omega^K_{n, b})'.
			\end{cases}
		\end{equation}
	\end{itemize}
\end{theorem}

\noindent
Theorem 3 shows that the endogenous interaction effects are identified from data on finite networks. These effects can be estimated by linear IV, where $\bar\omega^k_{n, b}$ is the instrument for $\bar\omega^k_{n, a}$, and outcomes are defined in (\ref{eq:16}). Replacing population moments with sample counterparts yields estimators that are consistent as the number of networks grows, for a fixed network size distribution.\footnote{Consistency requires that the measurement errors are mean zero and uncorrelated, such that $\scalebox{0.95}{$\mean(u_{n, a}^k) = 0$}$ and $\scalebox{0.95}{$\mean(u_{n, a}^k u_{n, b}^\ell) = 0$}$ for all $k, \ell$. These moment conditions continue to hold even in the presence of multiple equilibria.} Appendix B.3 derives their asymptotic properties to support inference.

In Table 2, I assess the performance of the IV estimator using Monte Carlo simulations, varying both the network sizes and the number of networks. The results highlight two key properties of the estimator. First, performance improves for larger network sizes. This is because, with larger networks, $\bar\omega^k_{n,a}$ more closely approximates the expectation $m^k_n$, reducing measurement error. Second, holding network sizes fixed, performance improves as the number of networks increases. This pattern is consistent with the identification arguments underlying Theorem 3, which shows that IV can recover interaction effects from data on finite networks. Crucially, this same pattern does not hold for OLS. In particular, while OLS \\ remains biased as the number of networks tends to infinity, the IV estimator yields increasingly precise estimates. Moreover, the IV estimator's confidence intervals achieve coverage rates tightly centered around the nominal 95\% level, whereas OLS coverage is nearly zero.

While small network sizes do not preclude identification, one practical consideration in finite samples is instrument strength. For small network-type cells, the split-sample means, $\bar\omega^k_{n,a}$ and $\bar\omega^k_{n,b}$, become noisier proxies for the common expectation $m_n^k$, weakening the pop-\\[-0.75cm]

\noindent
ulation first-stage relationship between them. Thus, reliable estimation may require more networks to estimate this weaker first-stage relationship with sufficient precision when the number of agents per network and type is small. Appendix Table \ref*{oa-tab:small-networks} illustrates this tradeoff by extending the simulation design in Table 2 to very small networks with 4 to 10 agents per \\[-0.75cm]

\noindent
type, showing that the estimator can perform well if there are sufficiently many networks.

\begin{singlespace}
	\begin{table}[ht]
	\centering
	\caption{Performance of OLS and IV Estimators by Network Size and Number of Networks}
	\label{tab:monte-carlo}
	 \resizebox{\textwidth}{!}{
	\begin{tabular}{l l c c c c c c c c c c c c}
		\hline
		& & \multicolumn{6}{c}{Specification A ($K = 1$)} & \multicolumn{6}{c}{Specification B ($K = 2$)} \\
		\cmidrule(lr){3-8} \cmidrule(lr){9-14}
		Number of & Agents per
		& \multicolumn{2}{c}{$\widehat{h}$}
		& \multicolumn{2}{c}{$\widehat{d}$}
		& \multicolumn{2}{c}{$\widehat{J}$}
		& \multicolumn{2}{c}{$\widehat{h_1 - h_2}$}
		& \multicolumn{2}{c}{$\widehat{J_{11} - J_{21}}$}
		& \multicolumn{2}{c}{$\widehat{J_{22} - J_{12}}$} \\[-0.5ex]
		\cmidrule(lr){3-4}
		\cmidrule(lr){5-6}
		\cmidrule(lr){7-8}
		\cmidrule(lr){9-10}
		\cmidrule(lr){11-12}
		\cmidrule(lr){13-14}
		\multicolumn{1}{c}{Networks} &  \multicolumn{1}{c}{Network}
		& OLS & IV & OLS & IV & OLS & IV
		& OLS & IV & OLS & IV & OLS & IV \\
		\hline
		\hspace{0.1cm}N = 250 & \hspace{0.5cm}30 & 8.62 & 4.72 & 0.43 & 0.27 & 16.66 & 8.70 & 1.60 & 1.27 & 7.61 & 4.45 & 3.12 & 1.54 \\
		&  & {\footnotesize (0.20)} & {\footnotesize (2.50)} & {\footnotesize (0.05)} & {\footnotesize (0.59)} & {\footnotesize (0.27)} & {\footnotesize (3.40)} & {\footnotesize (0.17)} & {\footnotesize (1.16)} & {\footnotesize (0.28)} & {\footnotesize (2.16)} & {\footnotesize (0.17)} & {\footnotesize (1.23)} \\
		&  & {\footnotesize [0.00]} & {\footnotesize [0.97]} & {\footnotesize [0.00]} & {\footnotesize [0.97]} & {\footnotesize [0.00]} & {\footnotesize [0.97]} & {\footnotesize [0.00]} & {\footnotesize [0.98]} & {\footnotesize [0.00]} & {\footnotesize [0.98]} & {\footnotesize [0.00]} & {\footnotesize [0.97]} \\
		& \hspace{0.5cm}100 & 2.00 & 0.76 & 0.09 & 0.05 & 3.91 & 1.36 & 0.36 & 0.16 & 1.39 & 0.56 & 0.37 & 0.24 \\
		&  & {\footnotesize (0.22)} & {\footnotesize (0.95)} & {\footnotesize (0.06)} & {\footnotesize (0.22)} & {\footnotesize (0.30)} & {\footnotesize (1.29)} & {\footnotesize (0.14)} & {\footnotesize (0.49)} & {\footnotesize (0.25)} & {\footnotesize (0.89)} & {\footnotesize (0.15)} & {\footnotesize (0.51)} \\
		&  & {\footnotesize [0.00]} & {\footnotesize [0.98]} & {\footnotesize [0.00]} & {\footnotesize [0.95]} & {\footnotesize [0.00]} & {\footnotesize [0.98]} & {\footnotesize [0.00]} & {\footnotesize [0.99]} & {\footnotesize [0.00]} & {\footnotesize [0.98]} & {\footnotesize [0.04]} & {\footnotesize [0.96]} \\
		& \hspace{0.5cm}200 & 0.72 & 0.25 & 0.04 & 0.02 & 1.39 & 0.43 & 0.11 & 0.06 & 0.41 & 0.20 & 0.11 & 0.09 \\
		&  & {\footnotesize (0.17)} & {\footnotesize (0.51)} & {\footnotesize (0.04)} & {\footnotesize (0.12)} & {\footnotesize (0.23)} & {\footnotesize (0.70)} & {\footnotesize (0.10)} & {\footnotesize (0.28)} & {\footnotesize (0.18)} & {\footnotesize (0.51)} & {\footnotesize (0.11)} & {\footnotesize (0.30)} \\
		&  & {\footnotesize [0.00]} & {\footnotesize [0.96]} & {\footnotesize [0.03]} & {\footnotesize [0.93]} & {\footnotesize [0.00]} & {\footnotesize [0.97]} & {\footnotesize [0.05]} & {\footnotesize [0.97]} & {\footnotesize [0.03]} & {\footnotesize [0.98]} & {\footnotesize [0.18]} & {\footnotesize [0.95]} \\
		&  &  &  &  &  &  &  &  &  &  &  &  &  \\
		\hspace{0.1cm}N = 1,000 & \hspace{0.5cm}30 & 8.62 & 1.04 & 0.44 & 0.06 & 16.63 & 1.91 & 1.61 & 0.28 & 7.70 & 0.99 & 3.14 & 0.33 \\
		&  & {\footnotesize (0.10)} & {\footnotesize (1.10)} & {\footnotesize (0.03)} & {\footnotesize (0.26)} & {\footnotesize (0.14)} & {\footnotesize (1.49)} & {\footnotesize (0.08)} & {\footnotesize (0.55)} & {\footnotesize (0.14)} & {\footnotesize (1.02)} & {\footnotesize (0.09)} & {\footnotesize (0.58)} \\
		&  & {\footnotesize [0.00]} & {\footnotesize [0.98]} & {\footnotesize [0.00]} & {\footnotesize [0.97]} & {\footnotesize [0.00]} & {\footnotesize [0.98]} & {\footnotesize [0.00]} & {\footnotesize [0.97]} & {\footnotesize [0.00]} & {\footnotesize [0.97]} & {\footnotesize [0.00]} & {\footnotesize [0.96]} \\
		& \hspace{0.5cm}100 & 2.00 & 0.22 & 0.09 & 0.01 & 3.91 & 0.39 & 0.37 & 0.05 & 1.43 & 0.15 & 0.37 & 0.06 \\
		&  & {\footnotesize (0.11)} & {\footnotesize (0.47)} & {\footnotesize (0.03)} & {\footnotesize (0.11)} & {\footnotesize (0.15)} & {\footnotesize (0.64)} & {\footnotesize (0.07)} & {\footnotesize (0.24)} & {\footnotesize (0.12)} & {\footnotesize (0.45)} & {\footnotesize (0.07)} & {\footnotesize (0.26)} \\
		&  & {\footnotesize [0.00]} & {\footnotesize [0.95]} & {\footnotesize [0.00]} & {\footnotesize [0.93]} & {\footnotesize [0.00]} & {\footnotesize [0.95]} & {\footnotesize [0.00]} & {\footnotesize [0.97]} & {\footnotesize [0.00]} & {\footnotesize [0.97]} & {\footnotesize [0.00]} & {\footnotesize [0.96]} \\
		& \hspace{0.5cm}200 & 0.72 & 0.07 & 0.03 & 0.00 & 1.38 & 0.12 & 0.11 & 0.01 & 0.43 & 0.04 & 0.11 & 0.02 \\
		&  & {\footnotesize (0.09)} & {\footnotesize (0.26)} & {\footnotesize (0.02)} & {\footnotesize (0.06)} & {\footnotesize (0.12)} & {\footnotesize (0.35)} & {\footnotesize (0.05)} & {\footnotesize (0.14)} & {\footnotesize (0.09)} & {\footnotesize (0.25)} & {\footnotesize (0.05)} & {\footnotesize (0.15)} \\
		&  & {\footnotesize [0.00]} & {\footnotesize [0.95]} & {\footnotesize [0.00]} & {\footnotesize [0.93]} & {\footnotesize [0.00]} & {\footnotesize [0.95]} & {\footnotesize [0.00]} & {\footnotesize [0.98]} & {\footnotesize [0.00]} & {\footnotesize [0.98]} & {\footnotesize [0.00]} & {\footnotesize [0.96]} \\
		&  &  &  &  &  &  &  &  &  &  &  &  &  \\
		\hspace{0.1cm}N = 5,000 & \hspace{0.5cm}30 & 8.64 & 0.19 & 0.44 & 0.01 & 16.67 & 0.34 & 1.61 & 0.07 & 7.71 & 0.24 & 3.15 & 0.08 \\
		&  & {\footnotesize (0.04)} & {\footnotesize (0.48)} & {\footnotesize (0.01)} & {\footnotesize (0.11)} & {\footnotesize (0.06)} & {\footnotesize (0.65)} & {\footnotesize (0.04)} & {\footnotesize (0.25)} & {\footnotesize (0.06)} & {\footnotesize (0.46)} & {\footnotesize (0.04)} & {\footnotesize (0.26)} \\
		&  & {\footnotesize [0.00]} & {\footnotesize [0.97]} & {\footnotesize [0.00]} & {\footnotesize [0.97]} & {\footnotesize [0.00]} & {\footnotesize [0.97]} & {\footnotesize [0.00]} & {\footnotesize [0.94]} & {\footnotesize [0.00]} & {\footnotesize [0.94]} & {\footnotesize [0.00]} & {\footnotesize [0.94]} \\
		& \hspace{0.5cm}100 & 2.02 & 0.05 & 0.09 & 0.00 & 3.94 & 0.09 & 0.37 & 0.01 & 1.44 & 0.03 & 0.37 & 0.02 \\
		&  & {\footnotesize (0.05)} & {\footnotesize (0.21)} & {\footnotesize (0.01)} & {\footnotesize (0.05)} & {\footnotesize (0.07)} & {\footnotesize (0.29)} & {\footnotesize (0.03)} & {\footnotesize (0.11)} & {\footnotesize (0.06)} & {\footnotesize (0.20)} & {\footnotesize (0.03)} & {\footnotesize (0.11)} \\
		&  & {\footnotesize [0.00]} & {\footnotesize [0.94]} & {\footnotesize [0.00]} & {\footnotesize [0.92]} & {\footnotesize [0.00]} & {\footnotesize [0.94]} & {\footnotesize [0.00]} & {\footnotesize [0.97]} & {\footnotesize [0.00]} & {\footnotesize [0.97]} & {\footnotesize [0.00]} & {\footnotesize [0.94]} \\
		& \hspace{0.5cm}200 & 0.71 & 0.01 & 0.03 & 0.00 & 1.38 & 0.02 & 0.11 & 0.00 & 0.43 & 0.01 & 0.11 & 0.00 \\
		&  & {\footnotesize (0.04)} & {\footnotesize (0.12)} & {\footnotesize (0.01)} & {\footnotesize (0.03)} & {\footnotesize (0.05)} & {\footnotesize (0.16)} & {\footnotesize (0.02)} & {\footnotesize (0.06)} & {\footnotesize (0.04)} & {\footnotesize (0.11)} & {\footnotesize (0.02)} & {\footnotesize (0.07)} \\
		&  & {\footnotesize [0.00]} & {\footnotesize [0.96]} & {\footnotesize [0.00]} & {\footnotesize [0.93]} & {\footnotesize [0.00]} & {\footnotesize [0.96]} & {\footnotesize [0.00]} & {\footnotesize [0.97]} & {\footnotesize [0.00]} & {\footnotesize [0.97]} & {\footnotesize [0.00]} & {\footnotesize [0.95]} \\
		\hline
	\end{tabular}}
	\begin{minipage}{\textwidth}
		\renewcommand{\baselinestretch}{0}\selectfont\setlength{\spaceskip}{0em}\footnotesize
		\emph{Notes.} This table reports mean squared errors, average estimated standard errors (in parentheses), and 95\% confi- \\
		
		dence interval coverage rates (in brackets), based on $1{,}000$ simulation draws.\hspace{-0.03cm} Specification A estimates $\beta_{k}$,\hspace{-0.02cm} defined \\
		
		in equation \eqref{eq:18}, with $K \!=\! 1$ and $\smash{\alpha_n \!=\! W_n'd}$, using the same design as Figure 1. Specification B estimates $\beta_{k_1\!, k_2}$, \\
		
		defined in equation \eqref{eq:17}, with $K \!=\! 2$ and $\alpha_n$ treated as unobserved. $X_n$ and $Z_n$ are constructed from disjoint 2/3 \\
		 
		and 1/3 subsamples, respectively, with $P \!=\! 1$. Appendix C.1 describes the simulation design and parameter values.
	\end{minipage}
\end{table}

\end{singlespace}

\section{Empirical Applications}\label{s4}

\noindent
I now present two empirical applications of the model:\ peer effects in Tennessee primary sc-\\[-0.75cm]

\noindent
hools and spillovers in deworming treatment uptake in Kenya.\ In both settings, allowing for nonuniform\hspace{-0.0025cm} interactions\hspace{-0.0025cm} materially\hspace{-0.0025cm} changes\hspace{-0.0025cm} the\hspace{-0.0025cm} conclusions\hspace{-0.0025cm} about\hspace{-0.0025cm} the\hspace{-0.0025cm} role\hspace{-0.0025cm} of\hspace{-0.0025cm} peer\hspace{-0.0025cm} influence.

\subsection{Project STAR: Gender Differences in Peer Effects}\label{s4.1}\enlargethispage{0.2\baselineskip}

\enlargethispage{\baselineskip}\noindent
The first application studies the Project STAR class size reduction experiment, launched in \\ 1985 among kindergarteners at 79 public schools in Tennessee. In each school, students and teachers were randomly assigned to one of three classroom types: small (13–17 students), large (22–25), and large with a teacher's aide.\ At the end of the year, students took exams in math and reading.\hspace{-0.04cm} My analysis uses data from the first-year cohort,\hspace{-0.04cm} which covered 6,325 stu-\\[-0.75cm]

\noindent
dents across 325 classrooms.\ See \cite{wo1990} for a detailed summary of the program.

To learn about gender differences in peer effects, I study a model with two student types: male ($m$) and female ($f$). The Project STAR sample is well-suited to this decomposition, as boys and girls are both well-represented in all classrooms.\footnote{As the data lack classroom identifiers, I follow \cite{boca2001} and \cite{gr2008} by uniquely matching students to classes using class observables. This yields a sample of 6,248 students across 321 classrooms.} As outcomes, I use indicators \\ for whether a student scored in the top 50\% statewide in math and in reading.\footnote{Appendix Table A.2 shows that the estimates remain similar for alternative thresholds, such as 25\% and 75\%.} Because scores are realized only at the end-of-year exam, students can respond only to their classma-\\tes' expected performance,\hspace{-0.04cm} consistent with the model's incomplete information assumption.

Outcomes are given by $\omega_i = \mathbbm{1}\{e_i \geq \varepsilon_i\}$, where $e_i$ is a student's effort and $\varepsilon_i$ is a threshold \\ that depends on individual ability,\hspace{-0.03cm} with distribution $\scalebox{0.9}{$F_{\hspace{-0.05cm}\varepsilon|k}$}$.\hspace{-0.04cm} Students choose effort to maximize:
\begin{equation}\label{eq:19}
U_i(e_i | k, n) = \left(h_k + \alpha_n  + J_{kk} \, \mean_i(\bar{\omega}_{n, -i}^k)  + J_{k\ell} \, \mean_i(\bar{\omega}_n^\ell)\right) e_i - \frac{1}{2} e_i^2.
\end{equation}
The gender fixed effect $h_k$ captures systematic differences in cognitive development and pr-\\[-0.75cm]

\noindent
ior socialization. The classroom fixed effect $\alpha_n$ accounts for contextual factors such as tea-\\cher quality and class size. The gender-specific peer effects $J_{k\ell}$ measure how effort depends on the expected performance of male and female peers. Equation \eqref{eq:9} defines an equilibrium.

Under the experimental protocols, students within a school could not self-select into cla-\\ssrooms.\hspace{-00.05cm} Thus, Assumption B.1 holds conditional on school.\hspace{-00.05cm} To account for nonrandom as-\\[-0.75cm]

\noindent
signment into schools, I include school fixed effects in the estimation.\footnote{Following \cite{gr2008}, I also report estimates for a subsample of schools that had only one classroom of each type. This addresses the unlikely possibility of nonrandom assignment to classrooms within classroom types.} When implementing the estimators, I assume $\varepsilon_i$ has a logistic distribution, $F_{\varepsilon|k}(z) = [1 + \exp(-z)]^{-1}$, thus avoiding the practical challenge of nonparametric estimation with a relatively small sample.

Table 3 shows IV estimates of $h_m - h_f$, $J_{mm} - J_{fm}$, and $J_{ff} - J_{mf}$ using Project STAR data. These estimates measure differences by gender in private effort costs and peer effects, while controlling for all factors at the classroom-level.\hspace{-0.08cm} The estimated difference between $h_m$ \\[-0.75cm]

\noindent
and $h_f$  is tightly centered around zero, implying that boys and girls perform similarly in the absence of peer effects.\ Peer effects, however, vary significantly by gender, with $J_{mm} \!>\! J_{fm}$ and $J_{ff} \!>\! J_{mf}$, indicating that students conform more strongly to peers of the same gender. These differences are slightly larger for reading than for math, though the differences across outcomes are not statistically significant. I also cannot reject that $J_{mm} \!-\! J_{fm} = J_{ff} \!-\! J_{mf}$, which suggests that peer effect differences may be approximately symmetric across gender.

To assess how allowing for generalized interactions affects the results, I investigate, in Table \ref*{oa-tab:star-uniform-effects} of the Appendix, how the estimates change under the assumption of uniform peer effects ($J_{k\ell} = J$). This restriction delivers substantively different conclusions, such as estimating that girls have lower private effort costs than boys ($h_f > h_m$), though this difference becomes statistically insignificant after including school fixed effects.\ Moreover, restricting peer effects to be uniform has implications for the predicted impacts of classroom-level pol-\\[-0.75cm]

\noindent
icies. Since $J_{mm} \!>\! J_{fm}$ and $J_{ff} \!>\! J_{mf}$, a policy that boosts the achievement of boys would have outsized effects on other boys but limited effects on girls. The overall impact of a clas-\\sroom intervention can therefore depend substantially on which students it directly targets.

\enlargethispage{0.3\baselineskip}

To further validate the approach, I conduct three robustness analyses. First, in Appendix Table \ref*{oa-tab:star-misspecification}, I test whether the parameters $(h_f,\! h_m,\! J_{ff},\! J_{fm},\! J_{mf},\! J_{mm})$ depend on observed classroom characteristics, such as the poverty rate, minority share, location, or teacher education and experience. Such dependence would suggest the model is misspecified. I find no evidence that these parameters depend on classroom factors. Second, in Appendix Figure \ref*{oa-fig:star-partitions}, I assess the sensitivity of the IV estimates to alternative classroom partitioning rules, plotting histograms of estimates across different partitions. I find that the estimates are rela-\\tively stable across partitioning rules. Third, in Appendix Table \ref*{oa-tab:star-simulations}, I validate the IV estim-\\ator's finite-sample performance using simulations calibrated to Project STAR's number of classrooms, classroom sizes and gender composition, and the estimated interaction effects.

\begin{singlespace}
	\begin{table}[ht]
	\centering
	\caption{Project STAR: IV Estimates of Peer Effect Differences by Gender}
	\label{tab:star-estimates}
	\footnotesize
		\begin{tabular}{lccc@{\hskip 20pt}cc}
			\toprule
			& \multicolumn{2}{c}{Top 50\% Math Score} & & \multicolumn{2}{c}{Top 50\% Reading Score}  \\
			\cmidrule(lr){2-3} \cmidrule(lr){5-6}
			& (1) & (2) & & (1)  & (2) \\
			\midrule
			$h_m - h_f$ & \phantom{1}-0.112\hspace{0.2cm}(0.16) & \phantom{-}\phantom{1}0.082\hspace{0.2cm}(0.20) & & \phantom{1}-0.029\hspace{0.2cm}(0.18) & \phantom{-}\phantom{1}0.019\hspace{0.2cm}(0.29) \\ 
			$J_{mm} - J_{fm}$ & \phantom{-}\phantom{1}4.229\hspace{0.2cm}(0.86) & \phantom{-}\phantom{1}4.755\hspace{0.2cm}(0.76) & & \phantom{-}\phantom{1}5.339\hspace{0.2cm}(0.82) & \phantom{-}\phantom{1}4.975\hspace{0.2cm}(0.55) \\ 
			$J_{ff} - J_{mf}$ & \phantom{-}\phantom{1}4.127\hspace{0.2cm}(1.01) & \phantom{-}\phantom{1}4.789\hspace{0.2cm}(1.07) & & \phantom{-}\phantom{1}5.223\hspace{0.2cm}(1.09) & \phantom{-}\phantom{1}4.686\hspace{0.2cm}(0.81) \\ 
			\midrule
			Observations & 5,798 & 5,798 & & 5,718 & 5,718 \\
			School FE       & No  & Yes & & No  & Yes \\
			1st Stage $F$ ($\bar\omega_n^m$) & 29.40 & 30.06 & & 25.52 & 25.73 \\
			1st Stage $F$ ($\bar\omega_n^f$) & 19.77 & 19.72 & & 13.37 & 12.98 \\
			\bottomrule
	\end{tabular}
	\begin{minipage}{\textwidth}
	\renewcommand{\baselinestretch}{0}\selectfont\setlength{\spaceskip}{0em}\footnotesize
	\emph{Notes.} This table presents estimates of gender differences in peer effects, applying the estimand in \eqref{eq:17}. Standard errors use the asymptotic variance \hspace{-0.01cm}$Q\Omega Q'$\hspace{-0.01cm} derived in Appendix B.3,\hspace{-0.03cm} where $\scalebox{0.925}{$\Omega \!=\! \var\!(Z_n[Y_n^{k_1} \!-\! Y_n^{k_2} \!-\! X_n'\beta_{k_1,k_2}])$}$ and $\scalebox{0.925}{$Q \!=\! \mean(Z_nX_n')^{-1}$}$. $X_n$ and $Z_n$ are constructed using disjoint 2/3 and 1/3 sample splits, respectively, setting $P = 1$.
	\end{minipage}
\end{table}
\end{singlespace}

\subsection{Primary School Deworming: Treatment Uptake with Nonuniform Spillovers}\label{s4.2}

\noindent
The second application revisits \citeauthor{mikr2004}\textcolor{blue}{'s} (\citeyear{mikr2004}) seminal study of the Primary School Deworming Project, conducted in Busia, Kenya, from 1998 to 2001. In this study, 75 \\[-0.75cm]

\noindent
schools were randomly assigned to three equal-sized groups: Group 1 began receiving free deworming treatment in 1998, Group 2 in 1999, and Group 3 not until 2001. A key feature of the intervention was that individuals could opt out of treatment. Indeed, among the nearly 15,000 students eligible in 1998-1999, only about 75\% received treatment in a given year. This partial compliance provides a natural setting to analyze spillovers in treatment uptake.

Following the original paper, I focus on the years 1998 and 1999, and I restrict the sample to school-aged children (ages 6-18) eligible for treatment. Eligibility covered all students at treatment schools, except girls older than 13 and, in 1999, anyone without parental consent.

I model each school as a social network in which interactions are global, and I allow the interactions to vary across age and gender. Given that girls older than 13 were ineligible for treatment, I specify a model with three types of students: girls aged $\leq 13$, boys aged $\leq 13$, and boys aged $>13$. Since most schools were large (averaging over 400 students), I assume students cannot track everyone else's choice and therefore act with incomplete information.\\[-0.55cm]

In each year $t$, a student $i$ of type $k$ in school $n$ decides whether to accept treatment based on private benefits and social interaction effects. A student's treatment decision is given by:\\[-0.55cm]
\begin{equation}\label{eq:20}
	\omega_{i, t} = \mathbbm{1}\bigg\{h_{k, t} + \alpha_{n, t} + X_{i, t}'c + J_{kk} \mean_i(\bar\omega_{n, -i, t}^k) + \sum_{\ell \neq k} J_{k\ell} \mean_i(\bar\omega_{n, t}^\ell) + \varepsilon_{i, t} \geq 0\bigg\}.
\end{equation}
The time-varying type fixed effect $h_{k, t}$ captures systematic differences by age and gender in \\[-0.75cm]

\noindent
health risk and other determinants of treatment uptake. The time-varying school fixed effect \\ $\alpha_{n, t}$ accounts for any cross-school externalities, as specified by \cite{mikr2004}, as well as school-level differences in treatment provision, baseline infection rates, and other factors.\footnote{In the paper, cross-school spillovers operate through the number of students living near the school who attend other treatment schools. As this channel is school$\times$time-specific, it is fully accounted for by the fixed effect $\alpha_{n, t}$.} The model also includes student-level controls $X_{i, t}$ for prior treatment uptake and prior school attendance, as well as an idiosyncratic taste term $\varepsilon_{i, t}$, with distribution $F_{\varepsilon|k, t}$.\footnote{I control for prior uptake to isolate contemporaneous interactions from preexisting factors. In Appendix Table \\ A.7,\hspace{-0.04cm} I also report estimates restricted to first-time eligible students.\hspace{-0.05cm} I control for attendance to ensure spillovers op-\\erate through uptake, not attendance, given that absence on the treatment day was a key source of noncompliance. Note that only prior attendance is used, as current attendance might be endogenously affected by treatment uptake.}

The interaction effects $J_{k\ell}$ measure how uptake depends on the expected treatment rates of different types of peers. Since deworming medication was both a new technology and a public good (by reducing infection rates among untreated individuals), social interactions in this setting reflect a combination of peer effects, social learning, and free-riding motives.\footnote{\cite{mikr2007} also study peer effects in this setting; yet, their analysis differs in three respects. First, they study parents' consent decisions, while I study student uptake after parental consent is granted. Second, they examine local interactions among close social contacts in a small sample of parents, where spillovers operate via the \textit{exogenous} number of links to treatment schools. Instead, I study \textit{endogenous} interactions within the school, controlling for the number of nearby students attending treatment schools. Third, they focus on later years of the program (post-2000), after schools started charging deworming fees, while I study uptake before such fees existed.}

Under the experimental design, students could not self-select into schools based on treatment, and most schools had similar observable characteristics \citep{mikr2004}. Nevertheless, to address any potential self-selection, I condition on a variety of school-level factors such as geographic zone, treatment group, local infection rate, population density, and number of nearby students attending other treatment schools. As in the previous application, I assume that the idiosyncratic tastes, $\varepsilon_{i, t}$, are logistically distributed in each period.

\enlargethispage{\baselineskip}
Table 4 presents IV estimates of $\scalebox{0.95}{$h_{k_1, t} - h_{k_2, t}$}$ and  $\scalebox{0.95}{$J_{k_1\ell} - J_{k_2\ell}$}$ using data from the deworming experiment. These estimates measure age and gender differences in private payoffs and interaction effects, while controlling for all school-level determinants of treatment uptake. I report four specifications with different combinations of year, student, and school controls.

The results show no significant difference in private payoffs. Yet, I find strong differences in interaction effects, with $J_{kk} \!>\! J_{\ell k}$ for $k \neq \ell$,  indicating that students prefer to conform to \\ peers of the same type. The degree of own-type conformity, as measured by the parameter $\delta_{k_1k_2} = J_{k_1k_1} + J_{k_2k_2} - J_{k_1k_2} - J_{k_2k_1}$, is highest for girls and boys under 14 ($\scalebox{0.95}{$k_1, k_2 \in \{1, 2\}$}$), implying that gender-based homophily is particularly pronounced among younger students.

The finding that students prefer to conform to similar peers helps to clarify the underlying mechanisms driving spillovers in this setting. In particular, this pattern of peer influence is difficult to reconcile with a pure free-riding motive, which would predict uniform spillovers or spillovers proportional to the epidemiological externality imposed by each type. Rather, it is more consistent with an imitation motive, where peer effects depend on social identity.

This finding also sheds light on the distributional impacts of school-based policies. Wit-\\hin schools,\hspace{-0.03cm} treatment uptake often differs greatly by age and gender,\hspace{-0.03cm} while the magnitude \\[-0.75cm]

\noindent
and direction of these gaps vary substantially from one school to another (Appendix Figure \ref*{oa-fig:deworming-uptake-gaps}). This variation suggests that school-level factors contribute to disparities in treatment uptake. By isolating the contribution of nonuniform peer effects in schools, the generalized interactions framework uncovers how social interactions affect these disparities. My results show that peer effect differences strongly influence which students decide to take treatment. Accounting for these effects can help policymakers target treatment access more effectively.

As before, I conduct three robustness checks. First, in Appendix Table \ref*{oa-tab:deworming-misspecification}, I test whether the parameters $h_{k, t}$ and $J_{k\ell}$ depend on observed school characteristics, such as the treatment group or local infection rate. I find no evidence of such dependence. Second, in Appendix Figure \ref*{oa-fig:deworming-partitions}, I plot the IV estimates across alternative network partitioning rules, and I find that they are relatively stable. Third, in Appendix Table \ref*{oa-tab:deworming-simulations}, I validate the finite-sample per-\\[-0.75cm]

\noindent
formance of the estimator using simulations calibrated to the number of schools, school size \\[-0.75cm]

\noindent
and demographic composition, and the  estimated interaction effects in the deworming data.

\begin{singlespace}
	\begin{table}[ht]
	\centering
	\caption{Deworming in Schools: IV Estimates of Spillover Differences by Age and Gender}
	\label{tab:deworming-estimates}
	\scriptsize
	\begin{tabular}{l@{\hskip 25pt}cccc}
		\toprule
		& (1) & (2) & (3) & (4) \\
		\midrule
		\multicolumn{2}{l}{\underline{\textit{Girls $\leq$ 13 and Boys $\leq$ 13}}} & & & \\
		$h_1 - h_2$ & \phantom{-}\phantom{1}0.928\hspace{0.2cm}(0.69) & \phantom{-}\phantom{1}0.857\hspace{0.2cm}(0.79) & \phantom{-}\phantom{1}1.431\hspace{0.2cm}(1.67) & \phantom{-}\phantom{1}1.256\hspace{0.2cm}(1.99) \\ 
		$J_{11} - J_{21}$ & \phantom{-}\phantom{1}8.293\hspace{0.2cm}(2.53) & \phantom{-}\phantom{1}8.679\hspace{0.2cm}(2.74) & \phantom{-}11.428\hspace{0.2cm}(4.38) & \phantom{-}12.273\hspace{0.2cm}(4.97) \\ 
		$J_{12} - J_{22}$ & \phantom{1}-9.382\hspace{0.2cm}(2.93) & -10.029\hspace{0.2cm}(3.26) & -13.827\hspace{0.2cm}(4.25) & -15.598\hspace{0.2cm}(5.17) \\ 
		$J_{13} - J_{23}$ & \phantom{-}\phantom{1}0.071\hspace{0.2cm}(0.71) & \phantom{-}\phantom{1}0.397\hspace{0.2cm}(0.86) & \phantom{-}\phantom{1}0.378\hspace{0.2cm}(1.16) & \phantom{-}\phantom{1}0.600\hspace{0.2cm}(1.13) \\
		& & & & \\
		\multicolumn{2}{l}{\underline{\textit{Girls $\leq$ 13 and Boys $>$ 13}}} & & & \\
		$h_1 - h_3$ & \phantom{1}-0.309\hspace{0.2cm}(0.89) & \phantom{1}-0.555\hspace{0.2cm}(0.97) & \phantom{-}\phantom{1}0.873\hspace{0.2cm}(1.41) & \phantom{-}\phantom{1}1.266\hspace{0.2cm}(1.17) \\ 
		$J_{11} - J_{31}$ & \phantom{-}\phantom{1}5.551\hspace{0.2cm}(2.78) & \phantom{-}\phantom{1}6.892\hspace{0.2cm}(3.02) & \phantom{-}\phantom{1}7.401\hspace{0.2cm}(4.13) & \phantom{-}\phantom{1}7.824\hspace{0.2cm}(3.88) \\ 
		$J_{12} - J_{32}$ & \phantom{1}-0.132\hspace{0.2cm}(2.38) & \phantom{1}-2.382\hspace{0.2cm}(2.95) & \phantom{1}-5.055\hspace{0.2cm}(4.45) & \phantom{1}-6.311\hspace{0.2cm}(3.91) \\ 
		$J_{13} - J_{33}$ & \phantom{1}-4.844\hspace{0.2cm}(0.57) & \phantom{1}-3.707\hspace{0.2cm}(0.93) & \phantom{1}-2.633\hspace{0.2cm}(1.23) & \phantom{1}-2.480\hspace{0.2cm}(1.43) \\ 
		& & & & \\
		\multicolumn{2}{l}{\underline{\textit{Boys $\leq$ 13 and Boys $>$ 13}}} & & & \\
		$h_2 - h_3$ & \phantom{1}-1.237\hspace{0.2cm}(1.09) & \phantom{1}-1.413\hspace{0.2cm}(0.86) & \phantom{1}-0.559\hspace{0.2cm}(1.25) & \phantom{-}\phantom{1}0.010\hspace{0.2cm}(1.46) \\ 
		$J_{21} - J_{31}$ & \phantom{1}-2.743\hspace{0.2cm}(3.22) & \phantom{1}-1.787\hspace{0.2cm}(3.02) & \phantom{1}-4.027\hspace{0.2cm}(4.49) & \phantom{1}-4.449\hspace{0.2cm}(5.38) \\ 
		$J_{22} - J_{32}$ & \phantom{-}\phantom{1}9.250\hspace{0.2cm}(3.06) & \phantom{-}\phantom{1}7.647\hspace{0.2cm}(3.20) & \phantom{-}\phantom{1}8.772\hspace{0.2cm}(5.43) & \phantom{-}\phantom{1}9.287\hspace{0.2cm}(6.51) \\ 
		$J_{23} - J_{33}$ & \phantom{1}-4.914\hspace{0.2cm}(0.47) & \phantom{1}-4.104\hspace{0.2cm}(0.61) & \phantom{1}-3.012\hspace{0.2cm}(1.00) & \phantom{1}-3.080\hspace{0.2cm}(1.02) \\ 
		\midrule
		Year Fixed Effects & No & Yes & Yes & Yes \\
		Student Controls & No & No & Yes & Yes \\
		School Controls  & No & No & No & Yes \\
		1st Stage $F$ ($\bar\omega_n^1$) & 14.10 & 11.20 & 14.35 & 13.09 \\
		1st Stage $F$ ($\bar\omega_n^2$) & 12.05 & 9.97 & 11.72 & 10.40 \\
		1st Stage $F$ ($\bar\omega_n^3$) & 35.60 & 16.69 & 18.56 & 17.88 \\
		\bottomrule
	\end{tabular}
	\begin{minipage}{\textwidth}
		\renewcommand{\baselinestretch}{0}\selectfont\setlength{\spaceskip}{0em}\footnotesize
		\emph{Notes.}\ This table reports estimated differences in spillover effects in treatment uptake, using the estimand in \eqref{eq:17}. Student controls are prior eligibility, uptake, and attendance. School controls are local infection rate, pop.\ density, and  nearby pupils at other treatment schools. Zone fixed effects are included. Standard errors use the asymptotic variance $Q\Omega Q'$ derived in Appendix B.3, where $\smash{\scalebox{0.925}{$\Omega \!=\! \var\!(Z_n[Y_n^{k_1} \!-\! Y_n^{k_2} \!-\! X_n'\beta_{k_1,k_2}])$}}$ and $\smash{\scalebox{0.925}{$Q \!=\! \mean(Z_nX_n')^{-1}$}}$. For all $(k, \ell)$, $h_k \!-\! h_\ell$ denotes a time average of $h_{k,t} \!-\! h_{\ell,t}$. $X_n$ and $Z_n$ use disjoint 2/3 and 1/3 sample splits with $P = 1$. \\[-0.55cm]
	\end{minipage}
\end{table}
\end{singlespace}

\section{Conclusion}\label{s6}

\noindent
In this paper, I developed new empirical tools to study social interaction effects in discrete choice models in which individuals respond differently to different network members, conf-\\orming to some while distinguishing from others. I first analyzed how the signs and magni-\\tudes of endogenous interaction effects influence people's behavior in equilibrium, and the implications for welfare. I then showed how to estimate endogenous interactions from data on individual choices while addressing two identification challenges:\hspace{-0.05cm} unobserved contextual \\[-0.75cm]

\noindent
effects and unobserved agent expectations in finite networks. By resolving these challenges, the paper's framework can be readily applied to standard datasets with finite network sizes. \\[-0.75cm]

\noindent
I applied the results to data from two large-scale education experiments:\ Project STAR and the Primary School Deworming Project. In both cases, I found that nonuniform interactions play an important role in shaping people's choices.\hspace{-0.05cm} While this paper provides an initial an-\\[-0.75cm]

\noindent
alysis of nonuniform and mixed-sign interactions in discrete choice models, there is much scope for future work. One valuable extension would be to combine the model with endog-\\enous network formation, allowing individuals to choose their exposure to different peers.

\newpage


\setcounter{footnote}{0}

\section*{Appendix}

\begin{proof}[Proof of Property 1.]
	Define $\mathcal{Q}: [0, 1]^K \to [0, 1]^K$ and $\mathcal{H}: [0, 1]^K \to \mathbb{R}^K$ such that:
	\begingroup\small\begin{equation}
		\mathcal{Q}_k(m) = F_{\varepsilon|k}\big(h_k + \sum_{\ell = 1}^K J_{k\ell} m^\ell\big) \quad\mbox{and}\quad \mathcal{H}_k(m) = m^k - \mathcal{Q}_k(m), \quad\mbox{for } k = 1, \dots, K. \notag
	\end{equation}\endgroup
	Since no root of $\mathcal{H}$ lies on the boundary of $[0, 1]^K$, I restrict attention to the interior $(0, 1)^K$.\footnote{Since each $F_{\varepsilon|k}$ has positive density everywhere, $\mathcal{H}_k(m) < 0$ when $m^k = 0$ and $\mathcal{H}_k(m) > 0$ when $m^k = 1$.} Letting $\mathbf{D}_\mathcal{H}$ denote the Jacobian of $\mathcal{H}$, define $\mathcal{C} = \{m \in (0, 1)^K : \detr(\mathbf{D}_\mathcal{H}(m)) = 0\}$. By Sard's Theorem, $\mathcal{H}(\mathcal{C})$ has Lebesgue measure zero. Also, for any given $y = \mathcal{H}(m)$, almost all $\{F_{\varepsilon|k}\}_{k}$ ensure that $y \notin \mathcal{H}(\mathcal{C})$. In particular, for $y = \mathbf{0}_K$, it follows that $\mathbf{D}_\mathcal{H}(m)$ is almost always nonsingular at every root of $\mathcal{H}$. By the inverse function theorem, each root of $\mathcal{H}$ is isolated. Since the equilibrium set is compact and has only isolated points, it must be finite.
	
	To prove that the number of equilibria is odd, I apply the Poincaré-Hopf Index Theorem (\citeauthor{mi1965}, \citeyear{mi1965}, Ch. 6). To do so, first note that the set of equilibria $\mathcal{M}$ lies in $(0, 1)^K$, which is diffeomorphic to an open disk. Also, the vector field $\mathcal{H}$ is smooth and points outward on the boundary, since $\lim_{m^k \to 0} \mathcal{H}_k(m) < 0$ and $\lim_{m^k \to 1} \mathcal{H}_k(m) > 0$ for all $k$. Moreover, for almost all $\{F_{\varepsilon|k}\}_k$, $\mathcal{H}$ has finitely-many isolated roots $m$, each with $\detr(\mathbf{D}_\mathcal{H}(m)) \neq 0$. By the index theorem: $\sum_{m \in \mathcal{M}} \mbox{\footnotesize index}_{m}(\scalebox{0.95}{$\mathcal{H}$}) = 1$ where $\mbox{\footnotesize index}_{m}(\scalebox{0.95}{$\mathcal{H}$}) = 2 \mathbbm{1}\{\detr(\mathbf{D}_\mathcal{H}(m)) > 0\} - 1$.
	
	Finally, suppose there exist $d_s$ equilibria where $\rho(\scalebox{0.95}{$\mathbf{D}(m)$}) < 1$. At each of these equilibria:
	\begingroup\small\begin{equation}
		\detr\big(\mathbf{D}_\mathcal{H}(m)\big) = \detr\big(I - \mathbf{D}(m)\big) = \prod_{k=1}^K\big(1 - \lambda_k(m)\big) > 0, \notag
	\end{equation}\endgroup
	where $\{\lambda_k(m)\}_{k = 1}^K$ are the eigenvalues of $\mathbf{D}(m)$. By the index theorem, there must also be at least $d_s - 1$ equilibria $m$ at which $\det(\mathbf{D}_\mathcal{H}(m)) < 0$, each one satisfying $\rho(\scalebox{0.95}{$\mathbf{D}(m)$}) > 1$.

\end{proof}
\vspace{-0.75cm}
\begin{proof}[Proof of Lemma 1.]
	\cite{hojo1985}, Thm. 5.6.9-5.6.10, show that $\rho(\scalebox{0.95}{$\mathbf{D}(m)$})$ is the greatest lower bound for the set of all matrix norms of $\mathbf{D}(m)$. In particular, they show:\\[-0.9cm]
	\begin{itemize}
		\item[\textit{(a)}] $\rho(\mathbf{D}(m)) \leq ||\mathbf{D}(m)||$ for any matrix norm $||\cdot||$.
		\item[\textit{(b)}] For any $\epsilon > 0$, there is a matrix norm $||\cdot||$ where $||\mathbf{D}(m)|| \leq \rho(\mathbf{D}(m)) + \epsilon$. \\[-0.9cm]
	\end{itemize}
	Let $\rho(\scalebox{0.95}{$\mathbf{D}(m)$}) < 1$. Setting $\epsilon < 1 - \rho(\scalebox{0.95}{$\mathbf{D}(m)$})$, it follows from \textit{(b)} that $||\mathbf{D}(m)|| < 1$ for some matrix norm $||\cdot||$. The system (\ref{eq:5}) is a contraction map at $m$ under this norm, which means:
	\begingroup\small\begin{equation}
		||m_{t} - m|| = ||\mathcal{Q}(m_{t-1}) - \mathcal{Q}(m)|| \leq \kappa||m_{t-1} - m||, \notag
	\end{equation}\endgroup
	where $\kappa \in [0, 1)$ for any vector $m_{t-1}$ that lies in a sufficiently small neighborhood of $m$. Iterating on this inequality implies that $||m_{t} - m|| \leq \kappa^t||m_{0} - m||$, where $\lim_{t\to\infty}\kappa^t = 0$.
	
	Next, let $\rho(\scalebox{0.95}{$\mathbf{D}(m)$}) > 1$. By \textit{(a)}, $||\mathbf{D}(m)|| > 1$ for any matrix norm $||\cdot||$. By \cite{he1981}, Thm. 5.1.5, there exists a scalar $u > 0$ such that, for any $\delta > 0$, there is an initial iterate $m_0$ where $||m_{0} - m|| < \delta$ and where some future iterate $m_t$, $t \geq 1$, satisfies $||m_{t} - m|| \geq u$.
	
\end{proof}
\vspace{-0.75cm}
\begin{proof}[Proof of Property 2.]
	First, suppose $\mathbf{D}(\cdot)$ is a nonnegative matrix. Assume there is \\ an equilibrium $m \in (0, 1)^K$ with $\rho(\scalebox{0.95}{$\mathbf{D}(m)$}) > 1$. The Perron-Frobenius theorem gives a non-\\zero vector $x \geq 0$ where $\mathbf{D}(m)x = \rho(\mathbf{D}(m))x$. For sufficiently small $\delta > 0$, let $a = m - \delta x$ and $b = m + \delta x$, both in $(0, 1)^K$. Taking the first-order Taylor approximation about $m$ gives:
	\begingroup\small\begin{align}
		\mathcal{Q}(m \pm \delta x) = m \pm \delta\rho(\mathbf{D}(m))x + o(\delta). \notag
	\end{align}\endgroup
	Thus, $\mathcal{Q}(a) \leq a$ and $\mathcal{Q}(b) \geq b$, with strict inequalities in coordinates where $x_k > 0$ and equ-\\ality where $x_k = 0$. These inequalities imply that $\mathcal Q$ maps each of the intervals $\left[\mathbf0_K,a\right]$ and $\left[b,\mathbf1_K\right]$ into itself, as $\mathbf0_K \!<\! \mathcal Q(\mathbf0_K) \!\leq\! \mathcal Q(a) \!\leq\! a$ and $b \!\leq\! \mathcal Q(b) \!\leq\! \mathcal Q(\mathbf1_K) \!<\! \mathbf1_K$.\ Brouwer's fixed point theorem yields two equilibria $\underline m\leq a$ and $\overline m\geq b$, both distinct from $m$. If either is unstable, repeating the argument guarantees another equilibrium farther downward or upward, \\[-0.75cm]
	
	\noindent
	respectively. Property 1 implies the equilibrium set is generically finite and $\det(I-\mathbf D)\neq0$ at every equilibrium. Because $\mathbf D(m)$ is nonnegative, its spectral radius is an eigenvalue by \\ Perron-Frobenius. Thus, $\det(I-\mathbf D(m))\neq0$ rules out $\rho(\mathbf D(m))=1$. Both chains must th-\\erefore terminate at distinct equilibria with $\rho(\mathbf D)<1$, which are locally stable by Lemma 1. \\[-0.55cm]
	
	I next extend the result to settings satisfying Assumption A. Let $\mathbf{B}$ be an invertible matrix for which $\mathbf{B}\mathbf{D}(\cdot)\mathbf{B}^{-1}$ is nonnegative. For $\mathcal{I} = [0, 1]^K$, define $\smash{{\hat{\mathcal{Q}}}: \scalebox{0.95}{$\mathbf{B}\mathcal{I} \to \mathbf{B}\mathcal{I}$}}$ such that $\smash{{\hat{\mathcal{Q}}}(m) = \scalebox{0.95}{$\mathbf{B}\mathcal{Q}(\mathbf{B}^{-1}m)$}}$. Its Jacobian is $\smash{\mathbf{D}_{\hat{\mathcal{Q}}}(m) = \scalebox{0.95}{$\mathbf{B}\mathbf{D}(\mathbf{B}^{-1}m)\mathbf{B}^{-1}$}}$. Since an invertible linear transformation preserves interior points, $\operatorname{int}(\mathbf B\mathcal I)=\mathbf B\operatorname{int}(\mathcal I)$. Thus, Assumption A ensures that $\mathbf D_{\hat{\mathcal Q}}(\cdot)\geq0$ throughout $\operatorname{int}(\mathbf B\mathcal I)$. It follows that the mapping ${\hat{\mathcal{Q}}}$ is monotone on $\mbox{int}(\mathbf{B}\mathcal{I})$.
	
	Since ${\hat{\mathcal{Q}}}(\mathbf{B}m) \!=\! \mathbf{B}\mathcal{Q}(m) \!=\! \mathbf{B}m$, the vector $m$ is a fixed point of $\mathcal{Q}$ if and only if $\mathbf{B}m$ is a \\ fixed point of ${\hat{\mathcal{Q}}}$.\hspace{-0.05cm} The Jacobians $\mathbf D(m)$ and $\mathbf D_{\hat{\mathcal Q}}(\mathbf Bm)$ are similar, so the corresponding fixed \\ points share the same eigenvalues. Thus, while the ordering of $\mathcal{Q}$ and $\hat{\mathcal{Q}}$'s fixed points could \\ differ, their number and local stability properties are the same. Applying the preceding arg-\\ument to the monotone map $\hat{\mathcal Q}$ establishes that, under Assumption A, an unstable equilibri-\\um implies multiple locally stable equilibria.\hspace{-0.05cm} The converse follows directly from Property 1. \\[-0.55cm]
	
\end{proof}
\vspace{-0.75cm}
\begin{proof}[Proof of Property 3.]
	If actions take values in the set $\{-1, 1\}$, then $\mathcal{W}_k(m)$ equals:
	\begingroup\small\begin{equation}
		\mathcal{W}_k(m) = \mean\bigg(\Big|h_k + \sum_{\ell = 1}^K J_{k\ell}m^{\ell} + \varepsilon_i\Big|\bigg) + \eta_k + \mean(\xi_i|k). \notag
	\end{equation}\endgroup
	Since each $F_{\varepsilon|k}$ is symmetric about zero, $\mathcal{W}_k(m)$ is strictly increasing in $|h_k + \sum_{\ell=1}^K J_{k\ell} m^\ell|$.\\[-0.075cm] Also, in equilibrium $|\mean(\bar{\omega}^k)|$ is strictly increasing in $|h_k + \sum_{\ell = 1}^K J_{k\ell}m^{\ell}|$. Therefore, it must be that $\argmax_{m \in \mathcal{M}}\mathcal{W}_k(m) = \argmax_{m \in \mathcal{M}} |\mean(\bar{\omega}^k)|$. The second part of Property 3 follows directly from the observations that $\lim_{h_k\to\infty}\mean(\bar\omega^k) = 1$ and $\lim_{h_k\to-\infty}\mean(\bar\omega^k) = -1$.
	
\end{proof}
\vspace{-0.5cm}
\begin{proof}[Proof of Property 4.]
	Suppose $J_{kj_1}J_{j_1j_2}\cdots J_{j_M k} \geq 0$ for any types $k, j_1, j_2, \dots, j_M$. By Lemma A.4 of the Online Appendix, there exists a diagonal matrix $\mathbf{B} = \diag[\gamma_1^{k}, \dots, \gamma_K^{k}]$, where $\gamma^{k}_\ell = 2\mathbbm{1}\{J_{k j_1}J_{j_1 j_2}\dots J_{j_M \ell} \geq 0 \mbox{ for all types } j_1, j_2, \dots, j_M\} - 1$, for which $\mathbf{B}\mathbf{J}\mathbf{B}^{-1}$ is nonnegative. As shown in the proof of Property 2, ${\hat{\mathcal{Q}}}(m) = \mathbf{B}\mathcal{Q}(\mathbf{B}^{-1}m)$ is monotone on $\operatorname{int}(\mathbf{B}\mathcal{I})$. Since $(\mathbf{B}\mathcal{I}, \leq)$ is a complete lattice, Tarski’s fixed point theorem ensures that the set of fixed points of $\hat{\mathcal{Q}}$ is a complete lattice and, consequently, contains both a greatest and least fixed point. If $\mathbf{B}\overline{m}$ is the greatest fixed point of ${\hat{\mathcal{Q}}}(m)$, then $\overline{m} = \argmax_{m \in \mathcal{M}} \mean(\bar{\omega}^\ell)$ if $\gamma^{k}_\ell = 1$ and $\overline{m} = \argmin_{m \in \mathcal{M}} \mean(\bar{\omega}^\ell)$ if $\gamma^{k}_\ell = -1$. Alternatively, if $\mathbf{B}\underline{m}$ is the least fixed point \\[-0.75cm]
	
	\noindent
	of ${\hat{\mathcal{Q}}}(m)$, then $\underline{m} = \argmin_{m \in \mathcal{M}} \mean(\bar{\omega}^\ell)$ if $\gamma^{k}_\ell = 1$ and $\underline{m} = \argmax_{m \in \mathcal{M}} \mean(\bar{\omega}^\ell)$ if $\gamma^{k}_\ell = -1$.
	
\end{proof}
\vspace{-0.75cm}
\begin{proof}[Proof of Theorem 1.]
	Fix $k_1$ and $k_2$. Since $F_{\varepsilon|k_1}$ is strictly increasing, it follows that:
	\begingroup\footnotesize\begin{align}
		m^{k_1}_n &= F_{\varepsilon|k_1}\bigg(h_{k_1} - h_{k_2} + \sum_{\ell = 1}^K (J_{k_1\ell} - J_{k_2\ell})m_n^{\ell} + F_{\varepsilon|k_2}^{-1}(m_n^{k_2})\bigg) \notag \\
		&= F_{\varepsilon|k_1}\bigg(\overline{h_{k_1} - h_{k_2}} + \sum_{\ell = 1}^K (\overline{J_{k_1\ell} - J_{k_2\ell}})m_n^{\ell} + F_{\varepsilon|k_2}^{-1}(m_n^{k_2})\bigg) \notag
	\end{align}\endgroup
	if and only if $\scalebox{0.88}{$(h_{k_1} - h_{k_2}) - (\overline{h_{k_1} - h_{k_2}}) = \sum_{\ell = 1}^K \big[(\overline{J_{k_1\ell} - J_{k_2\ell}}) - (J_{k_1\ell} - J_{k_2\ell})\big]m_n^{\ell}$}$. This property holds for all $n$. The rank condition in Assumption B.2 implies that $\scalebox{0.88}{$h_{k_1} - h_{k_2} = \overline{h_{k_1} - h_{k_2}}$}$ and $\scalebox{0.88}{$J_{k_1\ell} - J_{k_2\ell} = \overline{J_{k_1\ell} - J_{k_2\ell}}$}$ for any type $\ell \in \{1, \dots, K\}$. Since $k_1$ and $k_2$ are chosen arbitrarily, this result extends to all $k_1, k_2 \in \{1, \dots, K\}$. The proof in the $\alpha_n = W_n'd$ case is analogous.
	
\end{proof}
\vspace{-0.5cm}
\begin{proof}[Proof of Theorem 2.]
	I demonstrate that $\{F_{\varepsilon|k}\}_{k}$ and $c$ are identified up-to-scale, and the remainder follows as in the proof of Theorem 1. Fix any type $k$. By Assumption B.3, there is some element $x_j$ of $X$ with continuous variation on $\mathbb{R}$. Without loss of generality, let $x_j = x_1$, and normalize the coefficient $c_1$ to one. Also, fix some network $n$, and define $\zeta_{n}^k = h_k + \alpha_n + \sum_{\ell = 1}^K J_{k\ell}m_{n}^{\ell}$. For any agent $i$ in type $k$ and network $n$ with covariates $X_i$, the expected choice $\mean(\omega_i|X_i, k, \alpha_n, \{m_n^{\ell}\}_\ell)$ equals $F_{\varepsilon|k}(\zeta_n^k + X_i'c)$. So, to recover $(c, F_{\varepsilon|k})$,\\[-0.1cm] I must show that $F_{\varepsilon|k}(\zeta_n^k + X_i'c) = \overline{F}_{\varepsilon|k}(\overline{\zeta_n^k} + X_i'\overline{c})$ implies $c = \overline{c}$ and $F_{\varepsilon|k} = \overline{F}_{\varepsilon|k}$ for any $X_i \in \mbox{supp}(X|k)$. This property is shown by Corollary 5 of Proposition 2 in \cite{ma1988}. Also, since this holds for all $k$, I thus conclude that $c$ and $\{F_{\varepsilon|k}\}_k$ are identified up-to-scale.
	
\end{proof}
\vspace{-0.75cm}
\begin{proof}[Proof of Theorem 3.]
	I prove part $(ii)$ only, as the proof of part $(i)$ is analogous. As a first step, note that for any pair of types $k$ and $\ell$: (1) $\scalebox{0.95}{$\mean(\bar\omega_{n,j}^k) = \mean[\mean(m_n^k + u_{n,j}^k | n)] = \mean(m_n^k)$}$ for each $j \in \{a, b\}$; (2) $\scalebox{0.95}{$\mean(\bar\omega^k_{n, a}\bar\omega^\ell_{n, b}) = \mean[ \mean((m_n^k + u_{n,a}^k)(m_n^\ell + u_{n,b}^\ell) | n) ] = \mean(m^{k}_n m^{\ell}_n)$}$; and (3) $\scalebox{0.95}{$\mean(\bar\omega_{n, b}^\ell F_{\varepsilon| k}^{-1}(\bar\omega_{n, a}^k)) = \mean[\mean((m_n^\ell + u^\ell_{n, b}) F_{\varepsilon| k}^{-1}(\bar\omega_{n, a}^k)|n)] = \mean(m_n^\ell F_{\varepsilon| k}^{-1}(\bar\omega_{n, a}^k))$}$. Hence, I can write:
	\begingroup\footnotesize\begin{align}
		\mean\big(Z_n X_n'\big)^{-1}\mean\big(Z_n F_{\varepsilon| k}^{-1}(\bar\omega_{n, a}^k)\big) &= \mean\big(X_n^* (X_n^*)'\big)^{-1}\mean\big(Z_n F_{\varepsilon| k}^{-1}(\bar\omega_{n, a}^k)\big) \quad\quad\quad\mbox{by (1) and (2)} \notag 
		\\&= \mean\big(X_n^* (X_n^*)'\big)^{-1}\mean\big(X_n^* F_{\varepsilon| k}^{-1}(\bar\omega_{n, a}^k)\big) \quad\quad\quad\hspace{-0.05cm}\mbox{by (3)} \notag
		\\&= \underbrace{\mean\big(X_n^* (X_n^*)'\big)^{-1}\mean\big(X_n^* F_{\varepsilon|k}^{-1}(m_n^k)\big)}_{=\beta_k} + \mean\big(X_n^* (X_n^*)'\big)^{-1}\mean\big(X_n^* v_{n, a}^k\big), \notag
	\end{align}\endgroup
	where $\scalebox{0.95}{$X_n^* = (1, W_n', m_n^1, \dots, m_n^K)'$}$. As a notational shorthand, define $\smash{q_k = F_{\varepsilon | k}^{-1}}$. Expanding $q_k(\bar\omega_{n,a}^k)$ about $m_n^k$ through order $2P-1$, for $P \geq 2$, while retaining the remainder, gives:\footnote{For notational simplicity, this proof suppresses the treatment of subsample proportions $\bar\omega_{n,a}^k$ equal to zero or one, at which the inverse choice function could be infinite. Online Appendix B.2 specifies an endpoint smoothing convention and verifies the conditions under which the bias expansion and limiting correction below remain valid.}
	\begingroup\footnotesize\begin{align}
		\mean(v_{n, a}^k | n) &= \mean\left(q_k(\bar\omega^k_{n, a}) - q_k(m^{k}_n) \middle| n\right) \notag
		\\&= \mean\bigg(\sum_{r = 0}^{2P - 1} \frac{q_k^{(r)}(m_n^k)}{r!} (\bar\omega^k_{n, a} - m^k_n)^r + \rho_{n,a,P}^k - q_k(m^{k}_n) \bigg| n\bigg) \notag
		\\[-0.2cm]&= q_k'(m_n^k)\mean\left(\bar\omega^k_{n, a} - m^k_n \middle| n\right) + \sum_{r=2}^{2P-1}\frac{q_k^{(r)}(m_n^k)}{r!}\mean\left((\bar\omega_{n,a}^k-m_n^k)^r\middle| n\right) + \mean(\rho_{n,a,P}^k \mid n)
		\notag
		\\[-0.15cm]&= \sum_{r=2}^{2P-1} \frac{q_k^{(r)}(m_n^k)}{r!} B_r\left(0, \frac{\kappa_2(\omega | k, n)}{I_{n, a}^k}, \frac{\kappa_3(\omega | k, n)}{(I_{n, a}^k)^2}, \dots, \frac{\kappa_r(\omega | k, n)}{(I_{n, a}^k)^{r - 1}}\right) + \mean(\rho_{n,a,P}^k \mid n), \notag
	\end{align}\endgroup
	where $B_r$ is the $r$th complete Bell polynomial, and $\{\kappa_i(\omega | k, n)\}_i$ are cumulants of $\omega_i$ given $k$ and $n$. Each $B_r(\scalebox{0.8}{$0, \frac{\kappa_2(\omega | k, n)}{I_{n, a}^k}, \frac{\kappa_3(\omega | k, n)}{(I_{n, a}^k)^2}, \dots, \frac{\kappa_r(\omega | k, n)}{(I_{n, a}^k)^{r - 1}}$})$ is therefore a polynomial in $\scalebox{0.85}{$(I_{n, a}^k)^{-1}$}$ with no\\[-0.05cm] constant term and coefficients depending only on cumulants of $\omega_i$ given $k$ and $n$. It follows \\ that $\smash{\scalebox{0.95}{$\mean(v_{n,a}^k \mid n) = \sum_{s=1}^{P-1}a_{n,s}^k / (I_{n,a}^k)^s + r_{n,a,P}^k$}}$, where each coefficient $a_{n, s}^k$ does not depend on the subgroup size $I_{n, a}^k$, and the remainder $\smash{r_{n,a,P}^k}$ collects all the other terms. Noting that \\ $\scalebox{0.95}{$\mean(X_n^* v_{n, a}^k) \!=\! \mean(X_n^* \mean(v_{n, a}^k | n))$}$ and using independence of the subgroup sizes from $X_n^*$, write:
	\begingroup\footnotesize\begin{align}
		\mean\big(Z_n X_n'\big)^{-1}\mean\big(Z_n q_{k}(\bar\omega_{n, a}^k)\big) 
		&= \beta_k + \mean\left(X_n^* (X_n^*)'\right)^{-1}\mean\bigg(X_n^* \bigg[\sum_{s = 1}^{P - 1} \frac{a_{n, s}^k}{(I_{n, a}^k)^s} + r_{n,a,P}^k\bigg]\bigg) \notag
		\\[-0.3cm]&= \beta_k + \sum_{s = 1}^{P - 1} C_{k, s} \mean\big((I_{n, a}^k)^{-s}\big) + \mean\big(X_n^*(X_n^*)'\big)^{-1}\mean(X_n^*r_{n,a,P}^k), \notag
	\end{align}\endgroup
	where $\scalebox{0.95}{$C_{k, s} = \mean\left(X_n^* (X_n^*)'\right)^{-1}\mean\left(X_n^* a_{n, s}^k\right)$}$ are constant vectors that do not depend on $\scalebox{0.95}{$\{I_{n, a}^k\}_n$}$. By the definition of $\mathbf E_{k,P}$, both $\sum_{p=1}^{P}[\mathbf E_{k,P}^{-1}]_{p,1} \!=\! 1$ and $\sum_{p=1}^{P}[\mathbf E_{k,P}^{-1}]_{p,1}\mean[(I_{n,\kappa(p,n)}^k)^{-s}] \!=\! 0$ for $s = 1, \dots, P-1$. Therefore, letting $R_{k,P,p} = \mean(X_n^*(X_n^*)')^{-1}\mean\![X_n^*r_{n,\kappa(p,n),P}^k]$, I obtain:
	\begingroup\footnotesize\begin{align}
		\mean(Z_nX_n')^{-1}\mean\!\left[Z_n\sum_{p=1}^{P}[\mathbf E_{k,P}^{-1}]_{p,1}q_k(\bar\omega_{n,\kappa(p,n)}^k)\right]
		&= \sum_{p=1}^{P}[\mathbf E_{k,P}^{-1}]_{p,1}\left[\beta_k+\sum_{s=1}^{P-1}C_{k,s}\mean\!\left[(I_{n,\kappa(p,n)}^k)^{-s}\right]+R_{k, P, p}\right] \notag\\
		&= \beta_k+\sum_{p=1}^{P}[\mathbf E_{k,P}^{-1}]_{p,1}R_{k, P, p}. \notag
	\end{align}\endgroup
	Hence, the $P$-rule correction eliminates the first $P-1$ inverse-size terms in the expansion. To establish exact identification, it remains to justify the limiting correction as $P \!\to\! \infty$. Online Appendix B.2 gives examples of sequences of partitioning rules for which $\mathbf E_{k,P}$ is nonsingular at all $P$ and the resulting $\smash{Y_n^k}$ is integrable with $\smash{\mean(Y_n^k|m_n,W_n) = q_k(m_n^k)}$, so:
	\begingroup\small\begin{align}
		\beta^{IV}_{k} = \mean\big(X_n^*(X_n^*)'\big)^{-1}\mean\big(X_n^* Y_n^k\big) \notag
		= \mean\big(X_n^*(X_n^*)'\big)^{-1}\mean\big(X_n^* q_k(m_n^k)\big) = \beta_k. \notag
	\end{align}\endgroup
\end{proof}

\onehalfspacing
\bibliographystyle{ecta-fullname}
\bibliography{References}

\begin{thebibliography}{37}
\newcommand{\enquote}[1]{``#1''}
\expandafter\ifx\csname natexlab\endcsname\relax\def\natexlab#1{#1}\fi

\bibitem[\protect\citeauthoryear{Akerlof and Kranton}{Akerlof and
  Kranton}{2000}]{akkr2000}
\textsc{Akerlof, George~A. and Rachel~E. Kranton} (2000): \enquote{Economics
  and Identity,} \emph{The Quarterly Journal of Economics}, 115 (3), 715--753.

\bibitem[\protect\citeauthoryear{Bernheim}{Bernheim}{1994}]{be1994}
\textsc{Bernheim, B.~Douglas} (1994): \enquote{A Theory of Conformity,}
  \emph{Journal of Political Economy}, 102, 841--877.

\bibitem[\protect\citeauthoryear{Bhattacharya, Dupas, and Kanaya}{Bhattacharya
  et~al.}{2024}]{bhduka2023}
\textsc{Bhattacharya, Debopam, Pascaline Dupas, and Shin Kanaya} (2024):
  \enquote{Demand and Welfare Analysis in Discrete Choice Models with Social
  Interactions,} \emph{The Review of Economic Studies}, 91 (2), 748--784.

\bibitem[\protect\citeauthoryear{Bonhomme}{Bonhomme}{2012}]{bo2012}
\textsc{Bonhomme, St\'{e}phane} (2012): \enquote{Functional Differencing,}
  \emph{Econometrica}, 80 (4), 1337--1385.

\bibitem[\protect\citeauthoryear{Bonhomme and Dano}{Bonhomme and
  Dano}{2023}]{boda2023}
\textsc{Bonhomme, St\'{e}phane and Kevin Dano} (2023): \enquote{Functional
  Differencing in Networks,} Working Paper.

\bibitem[\protect\citeauthoryear{Boozer and Cacciola}{Boozer and
  Cacciola}{2001}]{boca2001}
\textsc{Boozer, Michael~A. and Stephen~E. Cacciola} (2001): \enquote{Inside the
  `Black Box' of Project STAR: Estimation of Peer Effects Using Experimental
  Data,} Center Discussion Paper 832, Yale University, Economic Growth Center.

\bibitem[\protect\citeauthoryear{Boxell, Gentzkow, and Shapiro}{Boxell
  et~al.}{2024}]{bogesh2022}
\textsc{Boxell, Levi, Matthew Gentzkow, and Jesse~M. Shapiro} (2024):
  \enquote{Cross-Country Trends in Affective Polarization,} \emph{The Review of
  Economics and Statistics}, 106 (2), 557--565.

\bibitem[\protect\citeauthoryear{Bramoullé and Kranton}{Bramoullé and
  Kranton}{2007}]{brkr2007}
\textsc{Bramoullé, Yann and Rachel Kranton} (2007): \enquote{Public Goods in
  Networks,} \emph{Journal of Economic Theory}, 135 (1), 478--494.

\bibitem[\protect\citeauthoryear{Brock and Durlauf}{Brock and
  Durlauf}{2001}]{brdu2001}
\textsc{Brock, William and Steven Durlauf} (2001): \enquote{Discrete Choice
  with Social Interactions,} \emph{The Review of Economic Studies}, 68 (2),
  235--260.

\bibitem[\protect\citeauthoryear{Brock and Durlauf}{Brock and
  Durlauf}{2007}]{brdu2007}
---\hspace{-.1pt}---\hspace{-.1pt}--- (2007): \enquote{Identification of Binary
  Choice Models with Social Interactions,} \emph{Journal of Econometrics}, 140
  (1), 52--75.

\bibitem[\protect\citeauthoryear{Chamberlain}{Chamberlain}{1980}]{ch1980}
\textsc{Chamberlain, Gary} (1980): \enquote{Analysis of Covariance with
  Qualitative Data,} \emph{The Review of Economic Studies}, 47 (1), 225--238.

\bibitem[\protect\citeauthoryear{Chamberlain}{Chamberlain}{1984}]{ch1984}
---\hspace{-.1pt}---\hspace{-.1pt}--- (1984): \enquote{Panel Data,} in
  \emph{Handbook of Econometrics}, Elsevier, vol.~2, 1247--1318.

\bibitem[\protect\citeauthoryear{Dano}{Dano}{2025}]{da2025}
\textsc{Dano, Kevin} (2025): \enquote{Transition Probabilities and Moment
  Restrictions in Dynamic Fixed Effects Logit Models,} Working Paper.

\bibitem[\protect\citeauthoryear{de~Paula}{de~Paula}{2013}]{pa2013}
\textsc{de~Paula, \'{A}ureo} (2013): \enquote{Econometric Analysis of Games
  with Multiple Equilibria,} \emph{Annual Review of Economics}, 5, 107--131.

\bibitem[\protect\citeauthoryear{de~Paula}{de~Paula}{2017}]{pa2017}
---\hspace{-.1pt}---\hspace{-.1pt}--- (2017): \enquote{Econometrics of Network
  Models,} in \emph{Advances in Economics and Econometrics: Eleventh World
  Congress}, ed. by Bo~Honor{\'e}, Ariel Pakes, Monika Piazzesi, and Larry
  Samuelson, Cambridge: Cambridge University Press, vol.~1, 268--323.

\bibitem[\protect\citeauthoryear{Ellison}{Ellison}{1993}]{el1993}
\textsc{Ellison, Glenn} (1993): \enquote{Learning, Local Interaction, and
  Coordination,} \emph{Econometrica}, 61, 1047--1071.

\bibitem[\protect\citeauthoryear{Graham}{Graham}{2008}]{gr2008}
\textsc{Graham, Bryan~S.} (2008): \enquote{Identifying Social Interactions
  through Conditional Variance Restrictions,} \emph{Econometrica}, 76 (3),
  643--660.

\bibitem[\protect\citeauthoryear{Graham}{Graham}{2015}]{gr2015}
---\hspace{-.1pt}---\hspace{-.1pt}--- (2015): \enquote{Methods of
  Identification in Social Networks,} \emph{Annual Review of Economics}, 7,
  465--485.

\bibitem[\protect\citeauthoryear{Guerra and Mohnen}{Guerra and
  Mohnen}{2022}]{gumo2022}
\textsc{Guerra, José-Alberto and Myra Mohnen} (2022): \enquote{Multinomial
  Choice with Social Interactions: Occupations in Victorian London,} \emph{The
  Review of Economics and Statistics}, 104 (4), 736--747.

\bibitem[\protect\citeauthoryear{Henry}{Henry}{1981}]{he1981}
\textsc{Henry, Daniel} (1981): \emph{Geometric Theory of Semilinear Parabolic
  Equations}, vol. 840 of \emph{Lecture Notes in Mathematics}, Berlin:
  Springer-Verlag.

\bibitem[\protect\citeauthoryear{Horn and Johnson}{Horn and
  Johnson}{1985}]{hojo1985}
\textsc{Horn, Roger~A. and Charles~R. Johnson} (1985): \emph{Matrix Analysis},
  Cambridge University Press.

\bibitem[\protect\citeauthoryear{Horowitz}{Horowitz}{2009}]{ho2009}
\textsc{Horowitz, Joel~L.} (2009): \emph{Semiparametric and Nonparametric
  Methods in Econometrics}, Springer Series in Statistics, New York: Springer.

\bibitem[\protect\citeauthoryear{Hoxby}{Hoxby}{2000}]{ho2000}
\textsc{Hoxby, Caroline~M.} (2000): \enquote{The Effects of Class Size on
  Student Achievement: New Evidence from Population Variation,} \emph{The
  Quarterly Journal of Economics}, 115 (4), 1239--1285.

\bibitem[\protect\citeauthoryear{Jackson and Zenou}{Jackson and
  Zenou}{2015}]{jaze2015}
\textsc{Jackson, Matthew~O. and Yves Zenou} (2015): \enquote{Games on
  Networks,} in \emph{Handbook of Game Theory with Economic Applications}, ed.
  by H.~Peyton Young and Shmuel Zamir, Elsevier, vol.~4, 95--163.

\bibitem[\protect\citeauthoryear{Kline and Tamer}{Kline and
  Tamer}{2020}]{klta2020}
\textsc{Kline, Brendan and Elie Tamer} (2020): \enquote{Econometric Analysis of
  Models with Social Interactions,} in \emph{The Econometric Analysis of
  Network Data}, ed. by Bryan Graham and \'{A}ureo~de Paula, Academic Press,
  149--181.

\bibitem[\protect\citeauthoryear{Kremer and Miguel}{Kremer and
  Miguel}{2007}]{mikr2007}
\textsc{Kremer, Michael and Edward Miguel} (2007): \enquote{The Illusion of
  Sustainability,} \emph{The Quarterly Journal of Economics}, 122 (3),
  1007--1065.

\bibitem[\protect\citeauthoryear{Lee, Li, and Lin}{Lee
  et~al.}{2014}]{lelili2014}
\textsc{Lee, Lungfei, Ji~Li, and Xu~Lin} (2014): \enquote{Binary Choice Models
  with Social Network Under Heterogeneous Rational Expectations,} \emph{The
  Review of Economics and Statistics}, 96 (3), 402--417.

\bibitem[\protect\citeauthoryear{Manski}{Manski}{1988}]{ma1988}
\textsc{Manski, Charles~F.} (1988): \enquote{Identification of Binary Response
  Models,} \emph{Journal of the American Statistical Association}, 83 (403),
  729--738.

\bibitem[\protect\citeauthoryear{Manski}{Manski}{1993}]{ma1993}
---\hspace{-.1pt}---\hspace{-.1pt}--- (1993): \enquote{Identification of
  Endogenous Social Effects: The Reflection Problem,} \emph{The Review of
  Economic Studies}, 60 (3), 531--542.

\bibitem[\protect\citeauthoryear{Miguel and Kremer}{Miguel and
  Kremer}{2004}]{mikr2004}
\textsc{Miguel, Edward and Michael Kremer} (2004): \enquote{Worms: Identifying
  Impacts on Education and Health in the Presence of Treatment Externalities,}
  \emph{Econometrica}, 72 (1), 159--217.

\bibitem[\protect\citeauthoryear{Milgrom and Roberts}{Milgrom and
  Roberts}{1990}]{miro1990}
\textsc{Milgrom, Paul and John Roberts} (1990): \enquote{Rationalizability,
  Learning, and Equilibrium in Games with Strategic Complementarities,}
  \emph{Econometrica}, 58 (6), 1255--1277.

\bibitem[\protect\citeauthoryear{Milnor}{Milnor}{1965}]{mi1965}
\textsc{Milnor, John} (1965): \emph{Topology from the Differentiable
  Viewpoint}, The University Press of Virginia.

\bibitem[\protect\citeauthoryear{Salop}{Salop}{1979}]{sa1979}
\textsc{Salop, Steven} (1979): \enquote{Monopolistic Competition with Outside
  Goods,} \emph{Bell Journal of Economics}, 10 (1), 141--156.

\bibitem[\protect\citeauthoryear{Schelling}{Schelling}{1971}]{sc1971}
\textsc{Schelling, Thomas~C.} (1971): \enquote{Dynamic Models of Segregation,}
  \emph{The Journal of Mathematical Sociology}, 1 (2), 143--186.

\bibitem[\protect\citeauthoryear{Topkis}{Topkis}{1998}]{to1998}
\textsc{Topkis, Donald~M.} (1998): \emph{Supermodularity and Complementarity},
  Frontiers of Economic Research, Princeton University Press.

\bibitem[\protect\citeauthoryear{Vives}{Vives}{1990}]{vi1990}
\textsc{Vives, Xavier} (1990): \enquote{Nash Equilibrium with Strategic
  Complementarities,} \emph{Journal of Mathematical Economics}, 19 (3),
  305--321.

\bibitem[\protect\citeauthoryear{Word, Johnston, Bain, Fulton, Zaharias,
  Achilles, Lintz, Folger, and Breda}{Word et~al.}{1990}]{wo1990}
\textsc{Word, Elizabeth~R., John Johnston, Helen~Pate Bain, B.~DeWayne Fulton,
  Jayne~B. Zaharias, Charles~M. Achilles, Martha~Nannette Lintz, John Folger,
  and Carolyn Breda} (1990): \enquote{The State of Tennessee's Student/Teacher
  Achievement Ratio (STAR) Project,} Technical report, 1985--1990, Tennessee
  State Department of Education, Nashville, TN.

\end{thebibliography}


\begin{thebibliography}{4}
\newcommand{\enquote}[1]{``#1''}
\expandafter\ifx\csname natexlab\endcsname\relax\def\natexlab#1{#1}\fi

\bibitem[\protect\citeauthoryear{Bramoull\'{e}, Kranton, and
  D'Amours}{Bramoull\'{e} et~al.}{2014}]{brkrda2014}
\textsc{Bramoull\'{e}, Yann, Rachel Kranton, and Martin D'Amours} (2014):
  \enquote{Strategic Interaction and Networks,} \emph{American Economic
  Review}, 104 (3), 898--930.

\bibitem[\protect\citeauthoryear{Graham}{Graham}{2008}]{gr2008}
\textsc{Graham, Bryan~S.} (2008): \enquote{Identifying Social Interactions
  through Conditional Variance Restrictions,} \emph{Econometrica}, 76 (3),
  643--660.

\bibitem[\protect\citeauthoryear{Horn and Johnson}{Horn and
  Johnson}{1985}]{hojo1985}
\textsc{Horn, Roger~A. and Charles~R. Johnson} (1985): \emph{Matrix Analysis},
  Cambridge University Press.

\bibitem[\protect\citeauthoryear{Manski}{Manski}{1988}]{ma1988}
\textsc{Manski, Charles~F.} (1988): \enquote{Identification of Binary Response
  Models,} \emph{Journal of the American Statistical Association}, 83 (403),
  729--738.

\end{thebibliography}


\end{document}


\begin{frontmatter}
	
\title{Online Appendix: \\ Discrete Choice with Generalized Social Interactions}
\runtitle{Online Appendix: Discrete Choice with Generalized Social Interactions}
\begin{aug}
	\author[add1]{\fnms{Oscar}~\snm{Volpe}\ead[label=e1]{ovolpe@fas.harvard.edu}}
	\address[add1]{\orgdiv{Department of Economics}, \orgname{Harvard University}}
\end{aug}
\vspace{-0.5cm}
\begin{center}
	September 30, 2026
\end{center}
\vspace{0.05cm}
\begin{abstract}
This online appendix accompanies the paper ``Discrete Choice with Generaliz-\\ed Social Interactions''. It contains supplementary proofs, model extensions, additional details about identification and estimation, and robustness analyses for the empirical applications. All notation is consistent with the main text of the paper.
\end{abstract}
\end{frontmatter}

\begin{appendix}
\section{Equilibrium Analysis and Model Extensions}

\subsubsection*{A.1.\hspace{0.1cm} Proofs of Supplementary Results}
\enlargethispage{0.5\baselineskip}
\begin{lemma}
	If $\scalebox{0.95}{$J_{kj_1}J_{j_1j_2}\cdots J_{j_M k} = 0$}$ for every integer $M \geq 1$ and all types $\scalebox{0.95}{$k, j_1, \dots, j_M$}$, then $\scalebox{0.95}{$\rho(\mathbf{D}(m)) = 0$}$ for all $\scalebox{0.95}{$m \in [0, 1]^K$}$. Hence, there is a unique, locally stable equilibrium.
\end{lemma}
\vspace{-0.6cm}
\begin{proof}
	Taking $M = 1$ and $j_1 = k$ gives $J_{kk}^2 = 0$, so $J_{kk} = 0$. It follows that $\scalebox{0.975}{$(\mathbf{J}^M)_{kk} = 0$}$ for all types $k$ and integers $M \geq 1$. Since this property also applies to $\mathbf{D}(m)$, it must be that: \\[-0.5cm]
	\begin{equation}
		\sum_{k = 1}^K [\lambda_k(m)]^M = \mbox{tr}([\mathbf{D}(m)]^M) = \sum_{k = 1}^K ([\mathbf{D}(m)]^M)_{kk} = 0, \notag
	\end{equation}
	for every $M \geq 1$, where $\{\lambda_k(m)\}_k$ denote the eigenvalues of $\mathbf{D}(m)$, counted with algebraic multiplicity.\hspace{-0.02cm} By Newton's identities, the vanishing of the first $K$ power sums guarantees that $\lambda_k(m) = 0$ for all $k$. Given that $\rho(\mathbf{D}(m)) = \max_k \{|\lambda_k(m)|\}$, it follows that $\rho(\mathbf{D}(m)) = 0$. The index theorem implies the equilibrium is unique. Local stability holds from Lemma 1. \\
\end{proof}

\begin{lemma}
	There exist positive thresholds $S_k$, which depend only on $f_{\varepsilon | k}$ and $h_k$, so \\ that, whenever $\scalebox{0.95}{$\sum_{\ell = 1}^K |J_{k\ell}| \!<\! S_k$}$ for every type $k$, the equilibrium is unique and locally stable.
\end{lemma}
\vspace{-0.6cm}
\begin{proof}
	By \cite{hojo1985}, Thm.\ 5.6.9, it follows that $\rho(\mathbf{D}(m)) \leq ||\mathbf{D}(m)||_\infty$, where $||\cdot||_\infty$ represents the $p$-norm defined by the maximum absolute row sum of $\mathbf{D}(m)$:
	\begin{equation} \small
		\|\mathbf{D}(m)\|_\infty = \max_k\left\{\sum_{\ell = 1}^K \big|[\mathbf{D}(m)]_{k\ell}\big|\right\} = \max_k\left\{f_{\varepsilon|k}\bigg(h_k + \sum_{\ell = 1}^K J_{k\ell}m^\ell\bigg)\sum_{\ell = 1}^K |J_{k\ell}|\right\}. \notag
	\end{equation}
	For any type $k$, choose $S_k \!>\! 0$ such that $\smash{S_k \max_{|z-h_k| \leq S_k} f_{\varepsilon|k}(z) \!=\! 1}$. Such a threshold exists because the left-hand side is continuous in $S_k$, tends to zero as $S_k \to 0$, and tends to infinity as $S_k \!\to\! \infty$.\ Its construction depends only on $h_k$ and $f_{\varepsilon|k}$.\ If $\sum_{\ell = 1}^K |J_{k\ell}| \!<\! S_k$, then for all $m$: \\[-0.5cm]
	\begin{equation} \small
		\left|\sum_{\ell = 1}^K J_{k\ell}m^\ell\right|
		\leq \sum_{\ell = 1}^K |J_{k\ell}| < S_k. \notag
	\end{equation}
	Therefore, $f_{\varepsilon|k}(h_k + \sum_{\ell = 1}^K J_{k\ell}m^\ell) \leq 1/S_k$. If the condition holds for every type $k$, then:
	\begin{equation} \small
		\sup_{m \in [0,1]^K}\|\mathbf{D}(m)\|_\infty
		\leq \max_k\left\{\frac{1}{S_k}\sum_{\ell = 1}^K |J_{k\ell}|\right\} < 1. \notag
	\end{equation}
	Thus, the equilibrium mapping is a contraction, with a unique and locally stable fixed point. \\
\end{proof}

\begin{lemma}
	If $\mathbf{J}$ is symmetric and its eigenvalues all have nonpositive real parts, then there is a unique equilibrium, locally stable if $\rho(\mathbf{D}(m)) < 1$ and unstable if $\rho(\mathbf{D}(m)) > 1$.
\end{lemma}
\vspace{-0.6cm}
\begin{proof}
	Let $\mathbf{J}$ be symmetric with eigenvalues that all have nonpositive real parts. Since $\boldsymbol{\beta}(m)$ has nonnegative diagonal entries, $\boldsymbol{\beta}^{1/2}(m)\mathbf{J}\boldsymbol{\beta}^{1/2}(m)$ is symmetric negative semidefinite. Additionally, this matrix and $\mathbf{D}(m) = \boldsymbol{\beta}(m)\mathbf{J}$ both share the same characteristic polynomial, because $AB$ and $BA$ have the same characteristic polynomial, taking $A = \boldsymbol{\beta}^{1/2}(m)$ and $B = \boldsymbol{\beta}^{1/2}(m)\mathbf{J}$. Therefore, all eigenvalues of $\mathbf{D}(m)$ are real and nonpositive. It follows that: $\scalebox{0.925}{$0 < \prod_{k=1}^K\big(1 - \lambda_k(m)\big) = \detr\big(I - \mathbf{D}(m)\big) = \detr\big(\mathbf{D}_\mathcal{H}(m)\big)$}$ at any equilibrium $m$. By the index theorem, the model always has a unique equilibrium. By Lemma 1, this equilibrium will be locally stable whenever $\scalebox{0.925}{$\rho(\mathbf{D}(m)) < 1$}$, and it will be unstable whenever $\scalebox{0.925}{$\rho(\mathbf{D}(m)) > 1$}$. \\
\end{proof}

\begin{lemma}
	Suppose that the interaction matrix $\mathbf{J}$ is irreducible. Then the No Negative Feedback condition holds if and only if $\scalebox{0.95}{$\mathbf{B}\mathbf{D}(\cdot)\mathbf{B}^{-1}$}$ is nonnegative for a diagonal matrix $\scalebox{0.95}{$\mathbf{B}$}$. \\[-0.5cm]
\end{lemma}
\vspace{-0.6cm}
\begin{proof}\enlargethispage{\baselineskip}
	Because $\mathbf{B}$ and $\boldsymbol{\beta}(m)$ are diagonal, it follows that $\mathbf{B}\mathbf{D}(m)\mathbf{B}^{-1} = \boldsymbol{\beta}(m)\mathbf{B}\mathbf{J}\mathbf{B}^{-1}$. Hence, nonnegativity of $\mathbf{B}\mathbf{J}\mathbf{B}^{-1}$ is equivalent to nonnegativity of $\mathbf{B}\mathbf{D}(m)\mathbf{B}^{-1}$ at every $m$.
	
	``$\Leftarrow$'' To begin, suppose that $\mathbf{J}$ is similar to a nonnegative matrix $\mathbf{A}$ by way of a diagonal change-of-basis matrix $\mathbf{B}$. Specifically, there exists a diagonal $\mathbf{B}$ where $\mathbf{A} = \mathbf{B}\mathbf{J}\mathbf{B}^{-1} \geq \mathbf{0}$. Since $\mathbf{B}$ is diagonal, the elements of $\mathbf{A}$ can be expressed as $A_{k\ell} = B_{kk} J_{k\ell} / B_{\ell\ell}$, for all $k, \ell$. Therefore, for any choice of indices $k$ and $\ell_1, \dots, \ell_M$ in the set $\{1, \dots, K\}$, it must be that:
	\begin{align} \small
		0 \leq A_{k\ell_1}A_{\ell_1\ell_2} \cdots A_{\ell_Mk} \notag
		&= \Big(\frac{B_{kk}}{B_{\ell_1\ell_1}}J_{k\ell_1}\Big)\Big(\frac{B_{\ell_1\ell_1}}{B_{\ell_2\ell_2}}J_{\ell_1\ell_2}\Big) \cdots \Big(\frac{B_{\ell_M\ell_M}}{B_{kk}}J_{\ell_Mk}\Big) \notag
		\\&= J_{k\ell_1}J_{\ell_1\ell_2} \cdots J_{\ell_Mk}. \notag
	\end{align}
	
	``$\Rightarrow$'' Suppose that NNF holds. When $\mathbf{J}$ is irreducible, all nonzero products of interaction effects along paths from $k$ to $\ell$ share the same sign; otherwise, appending the same return path with a nonzero product would make one feedback effect negative. For every $k$, define $\{\gamma^k_\ell\}_{\ell=1}^K$ so that $\gamma^k_\ell = 1$ if $J_{k j_1}J_{j_1 j_2}\dots J_{j_M \ell} \geq 0$ for all integers $M \geq 1$ and types $j_1, \dots, j_M$, and $\gamma^k_\ell = -1$ otherwise. Next, fixing some reference type $k_0$, define the matrix $\mathbf{B}$ such that:
	\begin{equation} \small
		\mathbf{B} = \diag
		\begin{bmatrix}
			\gamma_1^{k_0} \\ \vdots \\ \gamma_K^{k_0}
		\end{bmatrix}. \notag
	\end{equation}
	Notice that $\mathbf{B}$ is involutory, i.e. $\mathbf{B}^{-1} = \mathbf{B}$. Therefore, for all $(k, \ell)$, NNF would ensure that:
	\begin{equation} \small
		[\mathbf{B}\mathbf{J}\mathbf{B}^{-1}]_{k, \ell} = [\mathbf{B}\mathbf{J}\mathbf{B}]_{k, \ell} = \gamma_k^{k_0}\gamma_\ell^{k_0} J_{k\ell} = \gamma_{k_0}^{k}\gamma_\ell^{k_0} J_{k\ell} = \gamma^k_\ell J_{k \ell} = |J_{k \ell}| \notag
	\end{equation}
	It follows that $\mathbf{B}\mathbf{J}\mathbf{B}^{-1}$ equals the absolute value of $\mathbf{J}$. Therefore, $\mathbf{B}\mathbf{J}\mathbf{B}^{-1}$ is nonnegative.
\end{proof}

\begin{lemma}
	No Negative Feedback holds if $\scalebox{0.95}{$\sgn(J_{k\ell}) \!=\! \sgn(J_{km}J_{m\ell})$}$ for all $k, \ell,$ and $m$.
\end{lemma}
\vspace{-0.6cm}
\begin{proof}
	To begin, suppose that $\sgn(J_{k\ell}) = \sgn(J_{km}J_{m\ell})$ for any types $k, \ell, m$. I prove by induction that $\sgn(J_{j_0j_1}J_{j_1j_2}\cdots J_{j_{M}j_0}) \geq 0$ for any arbitrary types $j_0, j_1, \dots, j_M$. First, note that $\sgn(J_{j_0j_2}) = \sgn(J_{j_0j_1}J_{j_1j_2})$. Next, let $\sgn(J_{j_0 j_{M}}) = \sgn(J_{j_0j_1}J_{j_1j_2}\cdots J_{j_{M-1}j_{M}})$ for some positive integer $M \geq 1$. By construction, the following equation must be satisfied:
	\begin{equation}
		\sgn(J_{j_0j_{M + 1}}) = \sgn(J_{j_0 j_M}J_{j_Mj_{M + 1}}) = \sgn(J_{j_0j_1}J_{j_1j_2}\cdots J_{j_{M - 1}j_{M}}J_{j_{M}j_{M + 1}}). \notag
	\end{equation}
	By induction, $\sgn(J_{j_0 \ell}) = \sgn(J_{j_0j_1}J_{j_1j_2}\cdots J_{j_{M}\ell})$ for any agent type $\ell$. By setting $\ell = j_0$, I obtain that $\sgn(J_{j_0j_1}J_{j_1j_2}\cdots J_{j_{M}j_0}) = \sgn(J_{j_0 j_0})$ where $\sgn(J_{j_0 j_0}) = \sgn(J_{j_0 j_0}J_{j_0 j_0}) \geq 0$. \\
\end{proof}

\begin{lemma}\label{lem:weighted-feedback}
	Assumption A implies that, for each integer $r \geq 1$ and vector $m \in [0,1]^K$:
	\begin{equation}
		\sum_{(k_0,\ldots,k_{r-1})\in\{1,\ldots,K\}^r}
		\Bigg[\prod_{q=0}^{r-1}\beta_{k_q}(m)\Bigg]
		J_{k_0k_1}J_{k_1k_2}\cdots J_{k_{r-1}k_0} \geq 0, \notag
	\end{equation}
	where $\beta_k(m) \!=\! [\boldsymbol{\beta}(m)]_{kk}$ and, in the case where $r = 1$, the interaction product equals $J_{k_0k_0}$.
\end{lemma}
\vspace{-0.6cm}
\begin{proof}
	Assumption A implies $\scalebox{0.95}{$\mathbf{A}(m) \!=\! \mathbf{B}\mathbf{D}(m)\mathbf{B}^{-1}$}$ is entrywise nonnegative, so $\scalebox{0.95}{$\mathbf{A}(m)^r$}$ \\ is entrywise nonnegative for every $r \geq 1$. As trace is invariant under similarity, $\scalebox{0.95}{$\mbox{tr}(\mathbf{D}(m)^r)$} = \scalebox{0.95}{$\mbox{tr}(\mathbf{A}(m)^r)$} \!\geq\! 0$. Expanding $\scalebox{0.95}{$\mbox{tr}(\mathbf{D}(m)^r)$}$ and substituting $\scalebox{0.95}{$[\mathbf{D}(m)]_{k\ell} \!=\! \beta_k(m)J_{k\ell}$}$ gives the result.
\end{proof}

\subsubsection*{A.2.\hspace{0.1cm} Extensions to Alternative Network-Based Models}

\noindent
The properties laid out in Section 2 are not specific to binary choice models. They also have implications for a much broader class of models where agents interact in a network. To see \\ how, consider a game with $K$ players. Each player $k$ chooses $a_k$ from a compact action sp-\\ace $A_k \subset \mathbb{R}$. Given a profile of actions $a \in A_1\times A_2\times\cdots\times A_K$, each player's best response is $a_k^* = q_k\big(\sum_{\ell = 1}^K J_{k\ell}a_\ell\big)$, where $q_k(\cdot)$ is a non-decreasing mapping from $\mathbb{R}$ to $A_k$. The game permits, but does not require, each player's response to directly depend on her own action.

This game encompasses a variety of economic models. For example, the action $a_k$ could represent a person's investment into a public good, where everyone benefits from how much their neighbors contribute.\footnote{\cite{brkrda2014} study a public goods game that is nested within this framework. In their paper, players choose from an interval $[0, 1]$ and the best responses are $a_k = \max\{0, 1 - \delta\sum_{\ell \neq k} g_{k\ell}a_\ell\}$, where $\delta \in [0, 1]$ and $g_{k\ell} \in \{0, 1\}$ indicates whether players $k$ and $\ell$ are linked. Note that this game involves pure strategic substitutes.} Alternatively, $a_k$ could represent the output produced by a firm that competes in an oligopoly, where every firm's action affects the market price. Both these models are well-studied, and they both involve negative interaction effects between agents.

Another interpretation of the game is that each player is a community of individuals. In \\ the binary choice model, a player is a social group whose members make one of two choices \\ subject to social influences. Agents act noncooperatively, and $a_k$ is the average choice in group $k$. The same logic applies when residents of a state or local institution $k$ make a collective choice $a_k$. For example, one could model policy spillovers across U.S. states, where voters support more liberal or conservative positions in response to laws adopted elsewhere. \\[-0.5cm]

I focus on pure strategy Nash equilibria, defined as the action profiles $a$ from which no player $k$ wishes to deviate. An equilibrium will exist if there are continuous best responses and compact, convex action spaces. Moreover, if the best responses are continuously differentiable, then the spectral radius of the Jacobian matrix provides a sufficient statistic for determining uniqueness. To understand how, let $a^*$ be the equilibrium at which $\rho(\mathbf{D}(a))$ is greatest. If $\rho(\mathbf{D}(a^*)) \!<\! 1$, then $a^*$ is a unique, locally stable equilibrium. Conversely, under Assumption A, if $\rho(\mathbf{D}(a^*)) \!>\! 1$, then there are generically multiple, locally stable equilibria. \\[-0.55cm]

The set of equilibria is generally not well-ordered. In particular, it need not form a lattice structure, which would guarantee greatest and least equilibria at which every player takes her highest and lowest equilibrium actions, respectively. Property 4 illustrates how negative interactions can place different players’ highest equilibrium actions at different equilibria, and it considers the implications for social welfare that depend on the underlying payoffs.

\subsubsection*{A.3.\hspace{0.1cm} Endogenous Network Formation with Generalized Interactions}

\noindent
In the paper, it is assumed that agents respond to the expected average action for each type. I now consider an alternative model, where utility depends on the expected composition of agents who choose an action; that is, agents care about $\pr(k = a|\omega_i)$ instead of $\mean(\omega_i|k = a)$.

This alternative model specification can be used to analyze how social identity influences network formation. For example, suppose agents must decide whether to enter a new environment, such as a school, and they care about the types of people they expect to encounter. This situation invariably leads to negative interaction effects, because preferring one group to be in the majority is equivalent to preferring another group to be in the minority. These preferences can be analyzed by using the same techniques outlined in Section 2 of the paper. \\[-0.5cm]

Consider a model with two social groups: $a$ and $b$.\footnote{I focus on the two-group case for simplicity, but this framework may be generalized to cases with many groups.} Let $\lambda_a$ denote the share of people in group $a$, and let $s_a(\omega_i)$ denote the probability of being in group $a$ given that one chooses $\omega_i$. The payoff from the choice $\omega_i$ depends on the composition of agents who make that choice. \\[-0.5cm]
\begin{equation}\label{eq:11}
	U_i(\omega_i|k) = v_k(\omega_i) + J_{k} s_a(\omega_i) + \epsilon_i(\omega_i), \quad\mbox{for } k \in \{a, b\}. \notag
\end{equation}
Under this framework, the parameter $J_k$ indicates how much people in group $k$ benefit from associating with members of group $a$. Both $v_k(\cdot)$ and $\epsilon_i(\cdot)$ are specified exactly as before.

In equilibrium, the expected composition of agents making a choice must be consistent with individually optimal decisionmaking. By Bayes' rule, every equilibrium must satisfy:
\begin{equation}\label{eq:12}
	s_a(\omega_i) = \frac{\lambda_a \pr(\omega_i|a)}{\lambda_a\pr(\omega_i|a) + (1-\lambda_a)\pr(\omega_i|b)}, \quad\mbox{for } \omega_i \in \{0, 1\}, \notag
\end{equation}
where $\pr(\omega_i = 1|k) = F_{\varepsilon|k}\big(h_k + J_k (s_a(1) - s_a(0))\big)$. By Brouwer's FPT, a solution exists.

Uniqueness and dynamic stability of equilibria depend on the Jacobian matrix $\mathbf{D} \in \mathbb{R}^{2\times 2}$ of the system above. Letting $\beta_k = f_{\varepsilon|k}\big(h_k + J_k (s_a(1) - s_a(0))\big)$, I can write this matrix as:
\begin{equation}\label{eq:13}
	\mathbf{D} = 
	\begin{bmatrix}
		D_{11} & -D_{11} \\
		-D_{22} & D_{22}
	\end{bmatrix}, \quad\mbox{where: }
	\begin{cases}
		D_{11} &\hspace{-0.15cm}=\hspace{0.15cm} s_a(0)[1 - s_a(0)]\left(\frac{J_a\beta_a}{\pr(0|a)} - \frac{J_b\beta_b}{\pr(0|b)}\right) \\
		D_{22} &\hspace{-0.15cm}=\hspace{0.15cm} s_a(1)[1 - s_a(1)]\left(\frac{J_a\beta_a}{\pr(1|a)} - \frac{J_b\beta_b}{\pr(1|b)}\right)
	\end{cases} \notag
\end{equation}
In this model, the amount of diversity in the network directly shapes how strongly network composition responds to changes in group behavior. Specifically, each row of the Jacobian matrix $\mathbf{D}$ is scaled by the product $s_a(\omega_i)[1-s_a(\omega_i)]$, which is largest when the two groups are equally represented, and approaches zero as the network becomes homogeneous. Thus, holding all other factors fixed, a change in one group’s probability of choosing the action $\omega_i=1$ will have a larger effect on network composition when the network is more diverse.

Suppose that all agents prefer to associate with their own group, so $J_a \geq 0$ and $J_b \leq 0$. In this case, the matrix $\mathbf{D}$ satisfies Assumption A. Hence, there is almost always a locally stable equilibrium\textemdash even if the spillover effects are strong enough to generate multiplicity.

If agents prefer to associate with members of the other group, so $J_a < 0$ or $J_b > 0$, then \\ Assumption A may fail,\hspace{-0.03cm} and the model may not have any locally stable equilibria.\hspace{-0.05cm} Consider, for example, a model of school choice in which all households prefer schools with higher-achieving peers. If these preferences are sufficiently weak, then a unique, stable equilibrium exists. However, if households of students with lower prior achievement respond strongly to peer composition, a locally stable equilibrium may not exist: small shifts toward schools with higher-achieving peers can alter their composition enough to trigger larger shifts back.

\section{Details on Identification and Estimation}

\subsubsection*{B.1.\hspace{0.1cm} Identification with Individual-Level Covariates}
\enlargethispage{1.6\baselineskip}
\noindent
Suppose that the outcome equation is modified to allow for exogenous covariates $X_i \in \mathbb{R}^r$. Let $\omega_i = \mathbbm{1}\{X_i'c + h_k + \alpha_n + \sum_{\ell = 1}^K J_{k\ell}m^{\ell}_n + \varepsilon_i \geq 0\}$ where $\pr(\varepsilon_i \leq z|X_i, k, \alpha_n) = F_{\varepsilon|k}(z)$ and $\pr(X_i \leq x|k, \alpha_n) = F_{X|k}(x)$. Then $m_n^{k} = \int \mean(\omega_i|X_i, k, \alpha_n, \{m^{\ell}_n\}_{\ell = 1}^K) dF_{X|k}$, where:
\begin{equation} \small
	\mean(\omega_i|X_i, k, \alpha_n, \{m^{\ell}_n\}_{\ell = 1}^K) = F_{\varepsilon|k}\bigg(h_k + \alpha_n + X_i'c + \sum_{\ell = 1}^K J_{k\ell}m_n^{\ell}\bigg), \quad\mbox{for } k = 1, \dots, K. \notag
\end{equation}
In the presence of these covariates, Lemma 2 must be adapted. I do so in the following way.
\begin{lemma2}[Sufficiency Claim]
	For any agents $i$ and $j$ in groups $k_1$ and $k_2$, respectively, \\ and network $n$:
	$\scalebox{0.95}{$\mean(\omega_i|X_i, k_1, \alpha_n, \{m^{\ell}_n\}_{\ell}) = \mean(\omega_i|X_i, X_j, k_1, \{m^{\ell}_n\}_{\ell}, \mean(\omega_j|X_j, k_2, \alpha_n, \{m^{\ell}_n\}_{\ell}))$}$.
\end{lemma2}

\begin{proof}
	As $F_{\varepsilon|k_2}$ is strictly increasing, its inverse exists. By this property, I can write:
	\[
	\small\begin{aligned}
		\alpha_n = F_{\varepsilon|k_2}^{-1}\big(\mean(\omega_j|X_j, k_2, \alpha_n, \{m^{\ell}_n\}_{\ell})\big) - h_{k_2} -  X_j'c - \sum_{\ell = 1}^K J_{k_2\ell}m_n^{\ell}.
	\end{aligned}
	\]
	By plugging this expression for $\alpha_n$ into the definition of $\mean(\omega_i|X_i, k_1, \alpha_n, \{m^{\ell}_n\}_{\ell})$, I obtain:
	\[
	\footnotesize\begin{aligned}
		\mean(\omega_i|X_i, k_1, \alpha_n, \{m^{\ell}_n\}_{\ell}) = F_{\varepsilon|k_1}&\bigg(h_{k_1} - h_{k_2} + (X_i - X_j)'c \\[-0.2cm]&+ \sum_{\ell = 1}^K (J_{k_1\ell} - J_{k_2\ell})m^{\ell}_n + F_{\varepsilon|k_2}^{-1}\big(\mean(\omega_j|X_j, k_2, \alpha_n, \{m^{\ell}_n\}_{\ell})\big)\bigg).
	\end{aligned}
	\]\vspace{-0.65cm}
	
\end{proof}

\noindent
To prove Theorem 1 (semiparametric identification) with covariates $X_i$, I require that there is some $k$ such that $\mbox{supp}(X|k)$ is not contained in a proper affine subspace of $\mathbb{R}^r$. Note that this condition is weaker than Assumption B.3. Yet, it is sufficient to show that $c$, $h_{k_1} - h_{k_2}$ and $\{J_{k_1\ell} - J_{k_2\ell}\}_{\ell}$ are identified for all $k_1$ and $k_2$. Consider the following result and proof.\footnote{In these results and the main text, whenever network-level covariates $W_n$ are used, Assumption B.2 is unders-\\tood to require $\operatorname{rank}\!\left(\operatorname{Var}[(W_n',m_n')']\right) = \dim(W_n)+K$, so the covariance matrix of $(W_n',m_n')'$ has full rank.}

\begin{theorem}
Suppose Assumptions B.1 and B.2 hold, and assume $m^{k}_n$ is observed for all networks $n$ and agent types $k$. Also, suppose that there is some type $k$ where $\mbox{supp}(X|k)$ is not contained in a proper affine subspace of $\mathbb{R}^r$. If the functions $\{F_{\varepsilon|k}\}_k$ are known, then: \\[-0.5cm]
	\begin{itemize}
		\item[(i)] With no added assumptions, $c$, $\{h_{k_1} - h_{k_2}\}_{k_1, k_2}$ and $\{J_{k_1\ell} - J_{k_2\ell}\}_{k_1, k_2, \ell}$ are identified.
		\item[(ii)] If $\alpha_n = W_n'd$ for an observed vector $W_n$, then $c$, $d$, $\{h_k\}_k$, and $\{J_{k\ell}\}_{k, \ell}$ are identified.
	\end{itemize}
\end{theorem}

\begin{proof}
	Fix some network $n$ and define the term $\zeta_{n}^k = h_{k} + \alpha_{n} + \sum_{\ell = 1}^K J_{k\ell}m_{n}^{\ell}$. I write:
	\[
	\small\begin{aligned}
		\mean(\omega_i|X_i, k, \alpha_n, \{m^{\ell}_n\}_{\ell = 1}^K) = F_{\varepsilon|k}(\zeta_{n}^k + X_i'c)
	\end{aligned}
	\]
	To recover $c$, I must show that $F_{\varepsilon|k}(\zeta_{n}^k + X_i'c) = F_{\varepsilon|k}(\hat\zeta_{n}^k + X_i'\hat c)$ implies $c = \hat{c}$. Since $F_{\varepsilon|k}$ is known, this property holds by Proposition 5 of \cite{ma1988}. It follows that $c$ is identified.
	
	Having demonstrated that $c$ can be recovered, I now focus on the other parameters. Consider any two social groups $k_1$ and $k_2$, and then define the function $\phi: \mathbb{R} \!\to\! [0,\! 1]$ such that:
	\[
	\small\begin{aligned}
		\phi(\nu) = \int F_{\varepsilon|k_1}\big((X_i - X_j)'c + \nu\big)dF_{X|k_1},
	\end{aligned}
	\]
	where $X_j$ is chosen from $\mbox{supp}(X|k_2)$ and the integral is evaluated over the conditional support of $X_i$ given $k_1$. By definition, $\phi(\cdot)$ is nonlinear and monotonically increasing in $\nu$. 
	
	For any $n$, let $\scalebox{0.95}{$\nu_n = h_{k_1} - h_{k_2} + \sum_{\ell = 1}^K (J_{k_1\ell} - J_{k_2\ell})m_n^{\ell} + F_{\varepsilon|k_2}^{-1}(\mean(\omega_j|X_j, k_2, \alpha_n, \{m^{\ell}_n\}_{\ell = 1}^K))$}$. By Lemma 2, the expected average choice $m_n^{k_1}$ equals $\phi(\nu_n)$. Because $\phi(\cdot)$ is monotonic:
	\[
	\footnotesize\begin{aligned}
		m^{k_1}_n &= \int F_{\varepsilon|k_1}\bigg(h_{k_1} - h_{k_2} + (X_i - X_j)'c + \sum_{\ell = 1}^K (J_{k_1\ell} - J_{k_2\ell})m_n^{\ell} + F_{\varepsilon|k_2}^{-1}(\mean(\omega_j|X_j, k_2, \alpha_n, \{m^{\ell}_n\}_{\ell = 1}^K))\bigg)dF_{X|k_1} \\
		&= \int F_{\varepsilon|k_1}\bigg(\overline{h_{k_1} - h_{k_2}} + (X_i - X_j)'c + \sum_{\ell = 1}^K \overline{(J_{k_1\ell} - J_{k_2\ell})}m_n^{\ell} + F_{\varepsilon|k_2}^{-1}(\mean(\omega_j|X_j, k_2, \alpha_n, \{m^{\ell}_n\}_{\ell = 1}^K))\bigg)dF_{X|k_1}
	\end{aligned}
	\]
	is satisfied if and only if $\scalebox{0.95}{$\sum_{\ell = 1}^K \big[(\overline{J_{k_1\ell} - J_{k_2\ell}}) - (J_{k_1\ell} - J_{k_2\ell})\big]m_n^\ell = (h_{k_1} - h_{k_2}) - (\overline{h_{k_1} - h_{k_2}})$}$ for all networks in the population. By the rank condition in Assumption B.2, it follows that:
	\[
	\small\begin{aligned}
		h_{k_1} - h_{k_2} = \overline{h_{k_1} - h_{k_2}} \quad\mbox{and}\quad J_{k_1\ell} - J_{k_2\ell} = \overline{J_{k_1\ell} - J_{k_2\ell}},
	\end{aligned}
	\]
	for all $\ell \in \{1, \dots, K\}$. As $k_1$ and $k_2$ are chosen arbitrarily, this holds for all type pairs. The proof in the case where $\alpha_n = W_n'd$, corresponding to part \textit{(ii)} of the theorem, is analogous. \\
\end{proof}

\noindent
For Theorem 2, no modifications to the proof are needed, since the result already incorpor-\\ates individual-level covariates.\hspace{-0.02cm} Theorem 3 also extends to finitely supported covariates.\hspace{-0.02cm} Let $Y_n^{k,x}$ be the corrected outcome among agents with type $k$ and covariates $X_i=x$, satisfying \\ $\mean(Y_n^{k,x}\mid\alpha_n,m_n) \!=\! h_k \!+\! \alpha_n \!+\! x'c \!+\! \sum_\ell J_{k\ell}m_n^\ell$. Comparing two covariate values in the same type and network gives $\mean(Y_n^{k,x} \!-\! Y_n^{k,x_0}) \!=\! (x-x_0)'c$. The affine-support condition above identifies $c$, since the differences $x \!-\! x_0$ span $\mathbb{R}^r$ for at least one type. Once $c$ is identified, define $G_k(t;c)=\int F_{\varepsilon|k}(t+x'c)\,dF_{X|k}(x)$. The IV construction uses $G_k^{-1}(\cdot;c)$ in place of $F_{\varepsilon|k}^{-1}$, subject to the corresponding regularity and bias-correction conditions for this inverse.

\subsubsection*{B.2.\hspace{0.1cm} Additional Details on the Measurement Error Correction}
\enlargethispage{0.6\baselineskip}
\noindent
This appendix subsection provides an example of a sequence of partitioning rules for which $\mathbf E_{k,P}$ is nonsingular at each $P$ and $\mean(Y_n^k| m_n,W_n) \!=\! q_k(m_n^k)$. Consider a sampling rule that, \\ for any network and type,\hspace{-0.03cm} selects a subset $a$ independently of $(m_n,\! W_n)$,\hspace{-0.03cm} leaving $b$ nonempty. \\ Suppose that $a$ has sizes in $\{2, 4, 8, ...\}$, with positive probability on each size.\footnote{The network size can depend on $(m_n, W_n)$. The independence condition applies only to the selected size of a subgroup $a$. Additionally, in this example, the probabilities of larger subgroup sizes must not decline too quickly:\hspace{-0.03cm} setting $\liminf_{j\to\infty\hspace{-0.03cm}}\pr(I_{n,a}^k \!=\! 2^{j+1})/\pr(I_{n,a}^k \!=\! 2^{j}) \!>\! 1/2$ is sufficient for the outcome $Y_n^k$ to have a finite variance.} Let $a'$ be a \\ random half of $a$. Set $\kappa(1,n) \!=\! a$ and, for $j \!\geq\! 1$, $\kappa(j \!+\! 1, n) \!=\! a'$ if $I_{n,a}^k \!=\! 2^{j+1}$ and $a$ otherwise.

To verify that this sequence of partitioning rules delivers identification as $P \to \infty$, first note that the matrix $\mathbf E_{k,P}$ is nonsingular at any $P$, since the partitioning rules create linearly independent changes in the inverse-size moments by halving subgroups of distinct sizes.\ As \\[-0.75cm]
 
\noindent
$P$ grows, the limiting weight whenever $I_{n,a}^k=2^{j+1}$ is equal to $\Omega_n^k = H_j / \pi_j$, where I define $H_j \!=\! \pr(I_{n,a}^k \!\leq\! 2^{j+1})$ and $\pi_j \!=\! \pr(I_{n,a}^k \!=\! 2^{j+1})$. To see why, observe that if $j=0$, every rule selects $a$, so $\Omega_n^k=1=H_0/\pi_0$. Next, if $j\geq1$, then letting $e_1$ denote $(1,0,\ldots,0)'$, I obtain: \\[-0.6cm]
\[
\small\begin{aligned}
	\Omega_n^k
	&=\lim_{P\to\infty}\sum_{p=1}^{P}
	[\mathbf E_{k,P}^{-1}e_1]_p
	\mathbbm{1}\{\kappa(p,n)=a\}\\[-0.1cm]
	&=\lim_{P\to\infty}\Bigg[\sum_{p=1}^{P}[\mathbf E_{k,P}^{-1}e_1]_p-[\mathbf E_{k,P}^{-1}e_1]_{j+1}\Bigg]\\[-0.1cm]
	&=\lim_{P\to\infty}\left(1-[\mathbf E_{k,P}^{-1}e_1]_{j+1}\right)\\[-0.1cm]
	&=\lim_{P\to\infty}\Bigg[
	1+\frac{H_{j-1}}{\pi_j}
	+O\!\Bigg(\frac{2^{-(P-j)(P-j+1)/2}}{\pi_j}\Bigg)
	\Bigg]
	= 1+\frac{H_{j-1}}{\pi_j} = \frac{H_j}{\pi_j},
\end{aligned}
\]
where the fourth equality follows from Lagrange interpolation.\hspace{-0.02cm} By construction, the weights \\ depend only on the specified subgroup sizes $I_{n,a}^k$, not on the underlying model parameters.

I next verify that $\mean(Y_n^k| m_n,W_n) \!=\! q_k(m_n^k)$. For this step, take $q_k$ and its derivative to be \\ bounded;\hspace{-0.05cm} treatment of endpoints is discussed below.\hspace{-0.05cm} Denote $\scalebox{0.95}{$\mean[q_k(\bar\omega_{n,a}^k) | m_n^k \!=\! m, I_{n,a}^k \!=\! 2^{j+1}]$}$ by $\mu_j(m)$. Then $\mu_j(m) \!\to\! q_k(m)$ as $j\to\infty$. Since $a'$ is a random half of $a$, when $I_{n,a}^k \!=\! 2^{j+1}$ with $j \geq 1$ its corresponding expectation is $\mu_{j-1}(m)$. Using the identity $\Omega_n^k=H_j/\pi_j$ gives:  \\[-0.55cm]
\[
\small\begin{aligned}
	\mean(Y_n^k\mid m_n,W_n)
	&=\sum_{j=0}^{\infty}\pi_j\,
	\mean\!\left[
	\Omega_n^k q_k(\bar\omega_{n,a}^k)
	+(1-\Omega_n^k)q_k(\bar\omega_{n,a'}^k)
	\,\middle|\,m_n,W_n,\ I_{n,a}^k=2^{j+1}
	\right]\\
	&=H_0\mu_0(m_n^k)
	+\sum_{j=1}^{\infty}\pi_j
	\left[
	\frac{H_j}{\pi_j}\mu_j(m_n^k)
	+\left(1-\frac{H_j}{\pi_j}\right)\mu_{j-1}(m_n^k)
	\right]\\
	&=\lim_{R\to\infty}\left[
	H_0\mu_0(m_n^k)
	+\sum_{j=1}^{R}
	\left(H_j\mu_j(m_n^k)-H_{j-1}\mu_{j-1}(m_n^k)\right)
	\right]\\
	&=\lim_{R\to\infty}H_R\mu_R(m_n^k)
	=q_k(m_n^k).
\end{aligned}
\]
The sum in the third line telescopes. As $H_R \to 1$, it follows that $\mean(Y_n^k| m_n,W_n) \!=\! q_k(m_n^k)$.

\textit{Endpoint convention.}\ A valid convention, when the equilibrium choice probabilities are uniformly bounded away from zero and one, is to choose $\delta\in(0,1/2)$ so that $m_n^k\in[\delta,1-\delta]$ almost surely, and evaluate empirical inverse responses using a bounded, smooth function $\widetilde q_k$ on $[0,1]$ that agrees with $q_k$ on $[\delta,1-\delta]$. This function would have a bounded derivative and preserve all derivatives at the true probabilities. Hence, the finite-order bias expansion \\ remains valid, and the telescoping argument still delivers $\mean(Y_n^k| m_n,W_n) \!=\! q_k(m_n^k)$ exactly.

\textit{Network sizes and practical implementation.} Exact identification under this construction uses $P\to\infty$ and therefore requires subgroup sizes, and hence the available network sizes, to have a large support. Importantly, however, every network in the data can remain finite: consistency only requires that the number of networks grow under a fixed sampling scheme, for a fixed distribution of network sizes. Moreover, in practice, very few partitioning rules can suffice for precise estimation. For example, Table \ref*{main-tab:monte-carlo} in the main text shows that $P = 1$ already yields accurate estimates over a variety of specifications with fixed finite networks.

\subsubsection*{B.3.\hspace{0.1cm} Asymptotic Properties of the IV Estimator}

\noindent
Consider an i.i.d. sample of networks, with outcomes $Y_n^k$ as defined in (\ref*{main-eq:16}). Maintain Assumptions B.1 and B.2. Additionally, assume that $Z_nY_n^k$ and $Z_nX_n'$ have finite second moments for all $k$.\hspace{-0.05cm} First, consider the estimand $\beta_{k}$ defined in Theorem 3\textit{(ii)}.\hspace{-0.05cm} Its estimator is: 
\[
\small\begin{aligned}
	\hat\beta_{k} =  \left(\frac{1}{N}\sum_{n=1}^N Z_n X_n'\right)^{-1}\left(\frac{1}{N}\sum_{n=1}^N Z_n Y_n^{k}\right).
\end{aligned}
\]
Note that $\hat\beta_{k} \overset{p}{\to} \beta_{k}$ as the number of networks $N$ grows large.\footnote{Specifically, consistency is satisfied because $\smash{\scalebox{0.925}{$\frac{1}{N}\sum_{n=1}^N Z_n X_n' \overset{p}{\to} \mean(Z_nX_n')$}}$ and $\smash{\scalebox{0.925}{$\frac{1}{N}\sum_{n=1}^N Z_n Y_n^{k} \overset{p}{\to} \mean(Z_n Y_n^{k})$}}$ by the (weak) law of large numbers. Consequently, it follows from the continuous mapping theorem that $\hat\beta_k \overset{p}{\to} \beta_k$.} Additionally, it follows that:
\[
\small\begin{aligned}
	\sqrt{N}(\hat{\beta}_k - \beta_k) = \left(\frac{1}{N}\sum_{n=1}^N Z_n X_n'\right)^{-1}\left(\frac{1}{\sqrt{N}}\sum_{n=1}^N Z_n [Y_n^k - X_n'\beta_k]\right) \overset{d}{\to} \mathcal{N}(0, Q\Omega Q'),
\end{aligned}
\]
where $\scalebox{0.925}{$Q = \mean(Z_nX_n')^{-1}$}$ and $\scalebox{0.925}{$\Omega = \var(Z_n [Y_n^k - X_n'\beta_k])$}$. By the sample analogue principle, I can construct a consistent estimator for the asymptotic variance of $\hat{\beta}_k$ from $\hat Q\hat \Omega \hat Q'$, where:
\[
\small\begin{aligned}
	\hat Q = \left(\frac{1}{N}\sum_{n=1}^N Z_n X_n'\right)^{-1} \quad\mbox{and}\quad \hat \Omega = \frac{1}{N}\sum_{n=1}^N Z_n[Y_n^k - X_n'\hat\beta_k]^2 Z_n'.
\end{aligned}
\]
Next, consider the estimator $\hat\beta_{k_1, k_2}$, corresponding to the estimand $\beta_{k_1, k_2}$ in Theorem 3\textit{(i)}:
\[
\small\begin{aligned}
	\hat\beta_{k_1, k_2} =  \left(\frac{1}{N}\sum_{n=1}^N Z_n X_n'\right)^{-1}\left(\frac{1}{N}\sum_{n=1}^N Z_n [Y_n^{k_1} - Y_n^{k_2}]\right).
\end{aligned}
\]
As in the previous case, $\hat\beta_{k_1, k_2} \overset{p}{\to} \beta_{k_1, k_2}$ as $N$ grows large. Moreover, the estimator satisfies:
\[
\small\begin{aligned}
	\sqrt{N}(\hat\beta_{k_1, k_2} - \beta_{k_1, k_2}) \overset{d}{\to} \mathcal{N}(0, Q\Omega Q'),
\end{aligned}
\]
where $\scalebox{0.95}{$Q = \mean(Z_nX_n')^{-1}$}$ and $\scalebox{0.95}{$\Omega = \var(Z_n [Y_n^{k_1} - Y_n^{k_2} - X_n'\beta_{k_1, k_2}])$}$. As before, by the sample analogue principle, I can construct a consistent estimator $\hat Q\hat \Omega \hat Q'$ for $Q\Omega Q'$, where I define:
\[
\small\begin{aligned}
	\hat Q = \left(\frac{1}{N}\sum_{n=1}^N Z_n X_n'\right)^{-1} \quad\mbox{and}\quad \hat \Omega = \frac{1}{N}\sum_{n=1}^N Z_n[Y_n^{k_1} - Y_n^{k_2} - X_n'\hat\beta_{k_1, k_2}]^2 Z_n'.
\end{aligned}
\]

\section{Supplement to the Monte Carlo Simulations}

\subsubsection*{C.1.\hspace{0.1cm} Details on the Simulation Design}

\noindent
For the first set of Monte Carlo simulations (Table \ref*{main-tab:monte-carlo}, Specification A), choices are simulated from a data generating process (DGP) with $K = 1$ and $\alpha_n = W_n'd$. Outcomes are given by:
\[
\small\begin{aligned}
	\omega_i = \mathbbm{1}\bigg\{h + W_n'd + Jm_n + \varepsilon_i \geq 0\bigg\}, \quad\mbox{for } n = 1, \dots, N,
\end{aligned}
\]
where $m_n \!=\! \scalebox{0.925}{$F_{\varepsilon|k}(h \!+\! W_n'd \!+\! Jm_n)$}$ and $W_n$ takes equally spaced values on the interval $[-2,\! 2]$.\\ I assign the parameter values: $J = 0.5$, $h = 1$, and $d = 1$. I also assume $\smash{\varepsilon_i \overset{\mbox{\tiny i.i.d.}}{\sim} \mbox{Logistic}(0, 1)}$. This same design is used to generate Figure \ref*{main-fig:ols-bias}, which plots the bias of the OLS estimand.\footnote{While the model in the example presented in Section \ref*{main-s3.3} of the main text, given in (\ref*{main-eq:13}), technically has choices scaled to take values in $\{-1, 1\}$, the simulations throughout adopt a logit specification where choices lie in $\{0, 1\}$.}

For the second set of Monte Carlo simulations (Table \ref*{main-tab:monte-carlo}, Specification B), choices are si-\\mulated from a DGP with $K = 2$ and $\alpha_n$ treated as unobservable. Outcomes are given by:
\[
\small\begin{aligned}
	\omega_i = \mathbbm{1}\bigg\{h_k + \alpha_n + J_{k1}m^{1}_n + J_{k2}m^{2}_n + \varepsilon_i \geq 0\bigg\}, \quad\mbox{for } n = 1, \dots, N,
\end{aligned}
\]
where $m^{k}_n \!=\! F_{\varepsilon|k}(h_k \!+\! \alpha_n \!+\! J_{k1}m^{1}_n \!+\! J_{k2}m^{2}_n)$ for $k \in \{1, 2\}$, and $\alpha_n$ takes equally spaced va-\\lues on the interval $[-2, 2]$. I set $h_1 = 1$ and $h_2 = 0$, and I define $\mathbf{J}$ such that: $J_{11} = J_{22} = 1.5$ and $J_{12} = J_{21} = -1$. As in the first DGP, I assume that $\varepsilon_i|k \overset{\mbox{\tiny i.i.d.}}{\sim} \mbox{Logistic}(0, 1)$ for $k \in \{1, 2\}$. For all cases with $K = 2$, each network contains equal numbers of agents of the two types.

For each network, I compute the equilibrium choice probabilities by fixed-point iteration on (\ref*{main-eq:9}), then simulate individual choices conditional on these probabilities. Under both specifications, the equilibrium mapping is a contraction: the standard logistic density is bounded by \(1/4\), and the maximum absolute row sum of \(\mathbf J\) is less than four. Thus, each network has a unique equilibrium, and fixed-point iteration converges from any initial probability vector. Note that the contraction property is not required for the validity of the estimation; yet, it is useful here for guaranteeing convergence of the numerical algorithm in these simulations.

\subsubsection*{C.2.\hspace{0.1cm} Performance with Very Small Networks}

\noindent
Table \ref{tab:small-networks} extends the simulation designs in Table \ref*{main-tab:monte-carlo} of the main text to very small networks, containing 4 to 10 agents per network-type cell. Under this specification,  I vary the number of networks from $2{,}000$ to $5{,}000$ and compare the estimators’ mean squared errors, average estimated standard errors, and 95\% confidence interval coverage rates. All other aspects of the simulation design are the same as in Table 2 and as described in Appendix C.1 above.

Holding network sizes fixed, increasing the number of networks substantially reduces the IV estimator's mean squared errors and standard errors, while those of OLS remain largely unchanged. The IV estimator also achieves coverage rates close to the nominal 95\% level across all designs, whereas OLS coverage is essentially zero throughout. These properties hold even in the smallest designs, where there are only 4 or 6 agents per network-type cell. The results in Table \ref{tab:small-networks} therefore illustrate that the paper's IV estimator can perform relia-\\bly with few agents per network, provided that the data contain sufficiently many networks.

\begin{singlespace}
	\begin{table}[ht!]
	\centering
	\caption{OLS and IV Performance by Network-Type Cell Size and Number of Networks}
	\label{tab:small-networks}
	\resizebox{\textwidth}{!}{
		\begin{tabular}{l l c c c c c c c c c c c c}
			\hline
			& & \multicolumn{6}{c}{Specification A ($K = 1$)} & \multicolumn{6}{c}{Specification B ($K = 2$)} \\
			\cmidrule(lr){3-8} \cmidrule(lr){9-14}
			Number of & Agents per
			& \multicolumn{2}{c}{$\widehat{h}$}
			& \multicolumn{2}{c}{$\widehat{d}$}
			& \multicolumn{2}{c}{$\widehat{J}$}
			& \multicolumn{2}{c}{$\widehat{h_1 - h_2}$}
			& \multicolumn{2}{c}{$\widehat{J_{11} - J_{21}}$}
			& \multicolumn{2}{c}{$\widehat{J_{22} - J_{12}}$} \\[-0.5ex]
			\cmidrule(lr){3-4}
			\cmidrule(lr){5-6}
			\cmidrule(lr){7-8}
			\cmidrule(lr){9-10}
			\cmidrule(lr){11-12}
			\cmidrule(lr){13-14}
			\multicolumn{1}{c}{Networks} &  \multicolumn{1}{c}{$(n, k)$ Cell}
			& OLS & IV & OLS & IV & OLS & IV
			& OLS & IV & OLS & IV & OLS & IV \\
			\hline
			\hspace{0.1cm}N = 2,000 & \hspace{0.5cm}4 & 16.53 & 6.39 & 0.81 & 0.35 & 33.19 & 11.97 & 2.68 & 2.38 & 19.57 & 8.12 & 11.68 & 2.57 \\
			&  & {\footnotesize (0.06)} & {\footnotesize (2.68)} & {\footnotesize (0.02)} & {\footnotesize (0.63)} & {\footnotesize (0.07)} & {\footnotesize (3.64)} & {\footnotesize (0.06)} & {\footnotesize (1.51)} & {\footnotesize (0.09)} & {\footnotesize (2.84)} & {\footnotesize (0.06)} & {\footnotesize (1.61)} \\
			&  & {\footnotesize [0.00]} & {\footnotesize [0.92]} & {\footnotesize [0.00]} & {\footnotesize [0.93]} & {\footnotesize [0.00]} & {\footnotesize [0.91]} & {\footnotesize [0.00]} & {\footnotesize [0.98]} & {\footnotesize [0.00]} & {\footnotesize [0.98]} & {\footnotesize [0.00]} & {\footnotesize [0.96]} \\
			& \hspace{0.5cm}6 & 14.18 & 2.94 & 0.72 & 0.16 & 27.77 & 5.42 & 2.57 & 0.66 & 15.47 & 2.19 & 8.16 & 0.69 \\
			&  & {\footnotesize (0.06)} & {\footnotesize (1.75)} & {\footnotesize (0.02)} & {\footnotesize (0.41)} & {\footnotesize (0.08)} & {\footnotesize (2.38)} & {\footnotesize (0.06)} & {\footnotesize (0.77)} & {\footnotesize (0.10)} & {\footnotesize (1.45)} & {\footnotesize (0.06)} & {\footnotesize (0.82)} \\
			&  & {\footnotesize [0.00]} & {\footnotesize [0.94]} & {\footnotesize [0.00]} & {\footnotesize [0.95]} & {\footnotesize [0.00]} & {\footnotesize [0.94]} & {\footnotesize [0.00]} & {\footnotesize [0.96]} & {\footnotesize [0.00]} & {\footnotesize [0.96]} & {\footnotesize [0.00]} & {\footnotesize [0.95]} \\
			& \hspace{0.5cm}8 & 12.52 & 2.10 & 0.65 & 0.12 & 24.23 & 3.89 & 2.29 & 0.61 & 12.70 & 2.10 & 6.32 & 0.65 \\
			&  & {\footnotesize (0.06)} & {\footnotesize (1.50)} & {\footnotesize (0.02)} & {\footnotesize (0.35)} & {\footnotesize (0.08)} & {\footnotesize (2.03)} & {\footnotesize (0.06)} & {\footnotesize (0.72)} & {\footnotesize (0.09)} & {\footnotesize (1.35)} & {\footnotesize (0.06)} & {\footnotesize (0.77)} \\
			&  & {\footnotesize [0.00]} & {\footnotesize [0.96]} & {\footnotesize [0.00]} & {\footnotesize [0.96]} & {\footnotesize [0.00]} & {\footnotesize [0.96]} & {\footnotesize [0.00]} & {\footnotesize [0.95]} & {\footnotesize [0.00]} & {\footnotesize [0.95]} & {\footnotesize [0.00]} & {\footnotesize [0.95]} \\
			& \hspace{0.5cm}10 & 11.20 & 1.23 & 0.59 & 0.07 & 21.54 & 2.28 & 2.03 & 0.27 & 10.74 & 0.92 & 5.13 & 0.30 \\
			&  & {\footnotesize (0.06)} & {\footnotesize (1.12)} & {\footnotesize (0.02)} & {\footnotesize (0.26)} & {\footnotesize (0.08)} & {\footnotesize (1.53)} & {\footnotesize (0.06)} & {\footnotesize (0.54)} & {\footnotesize (0.09)} & {\footnotesize (1.01)} & {\footnotesize (0.06)} & {\footnotesize (0.57)} \\
			&  & {\footnotesize [0.00]} & {\footnotesize [0.96]} & {\footnotesize [0.00]} & {\footnotesize [0.96]} & {\footnotesize [0.00]} & {\footnotesize [0.96]} & {\footnotesize [0.00]} & {\footnotesize [0.96]} & {\footnotesize [0.00]} & {\footnotesize [0.95]} & {\footnotesize [0.00]} & {\footnotesize [0.95]} \\
			&  &  &  &  &  &  &  &  &  &  &  &  &  \\
			\hspace{0.1cm}N = 5,000 & \hspace{0.5cm}4 & 16.51 & 4.12 & 0.81 & 0.23 & 33.16 & 7.58 & 2.67 & 0.88 & 19.53 & 2.82 & 11.65 & 0.89 \\
			&  & {\footnotesize (0.04)} & {\footnotesize (1.99)} & {\footnotesize (0.01)} & {\footnotesize (0.47)} & {\footnotesize (0.05)} & {\footnotesize (2.70)} & {\footnotesize (0.04)} & {\footnotesize (0.90)} & {\footnotesize (0.06)} & {\footnotesize (1.69)} & {\footnotesize (0.04)} & {\footnotesize (0.96)} \\
			&  & {\footnotesize [0.00]} & {\footnotesize [0.94]} & {\footnotesize [0.00]} & {\footnotesize [0.95]} & {\footnotesize [0.00]} & {\footnotesize [0.93]} & {\footnotesize [0.00]} & {\footnotesize [0.97]} & {\footnotesize [0.00]} & {\footnotesize [0.98]} & {\footnotesize [0.00]} & {\footnotesize [0.95]} \\
			& \hspace{0.5cm}6 & 14.17 & 1.00 & 0.72 & 0.05 & 27.75 & 1.85 & 2.55 & 0.30 & 15.47 & 0.88 & 8.19 & 0.26 \\
			&  & {\footnotesize (0.04)} & {\footnotesize (1.01)} & {\footnotesize (0.01)} & {\footnotesize (0.24)} & {\footnotesize (0.05)} & {\footnotesize (1.37)} & {\footnotesize (0.04)} & {\footnotesize (0.48)} & {\footnotesize (0.06)} & {\footnotesize (0.90)} & {\footnotesize (0.04)} & {\footnotesize (0.51)} \\
			&  & {\footnotesize [0.00]} & {\footnotesize [0.95]} & {\footnotesize [0.00]} & {\footnotesize [0.95]} & {\footnotesize [0.00]} & {\footnotesize [0.95]} & {\footnotesize [0.00]} & {\footnotesize [0.93]} & {\footnotesize [0.00]} & {\footnotesize [0.95]} & {\footnotesize [0.00]} & {\footnotesize [0.94]} \\
			& \hspace{0.5cm}8 & 12.54 & 0.67 & 0.65 & 0.04 & 24.25 & 1.23 & 2.30 & 0.20 & 12.74 & 0.72 & 6.33 & 0.23 \\
			&  & {\footnotesize (0.04)} & {\footnotesize (0.91)} & {\footnotesize (0.01)} & {\footnotesize (0.21)} & {\footnotesize (0.05)} & {\footnotesize (1.23)} & {\footnotesize (0.04)} & {\footnotesize (0.45)} & {\footnotesize (0.06)} & {\footnotesize (0.84)} & {\footnotesize (0.04)} & {\footnotesize (0.48)} \\
			&  & {\footnotesize [0.00]} & {\footnotesize [0.96]} & {\footnotesize [0.00]} & {\footnotesize [0.96]} & {\footnotesize [0.00]} & {\footnotesize [0.96]} & {\footnotesize [0.00]} & {\footnotesize [0.96]} & {\footnotesize [0.00]} & {\footnotesize [0.96]} & {\footnotesize [0.00]} & {\footnotesize [0.95]} \\
			& \hspace{0.5cm}10 & 11.23 & 0.40 & 0.59 & 0.02 & 21.59 & 0.73 & 2.03 & 0.11 & 10.74 & 0.39 & 5.11 & 0.12 \\
			&  & {\footnotesize (0.04)} & {\footnotesize (0.68)} & {\footnotesize (0.01)} & {\footnotesize (0.16)} & {\footnotesize (0.05)} & {\footnotesize (0.92)} & {\footnotesize (0.04)} & {\footnotesize (0.34)} & {\footnotesize (0.06)} & {\footnotesize (0.64)} & {\footnotesize (0.04)} & {\footnotesize (0.36)} \\
			&  & {\footnotesize [0.00]} & {\footnotesize [0.95]} & {\footnotesize [0.00]} & {\footnotesize [0.95]} & {\footnotesize [0.00]} & {\footnotesize [0.95]} & {\footnotesize [0.00]} & {\footnotesize [0.96]} & {\footnotesize [0.00]} & {\footnotesize [0.96]} & {\footnotesize [0.00]} & {\footnotesize [0.96]} \\
			\hline
	\end{tabular}}
	\begin{minipage}{\textwidth}
		\renewcommand{\baselinestretch}{0}\selectfont\setlength{\spaceskip}{0em}\footnotesize
		\emph{Notes.} This table extends the simulation design in Table 2 to networks with 4 to 10 agents per network-type cell. It reports the mean squared errors, average estimated standard errors (in parentheses), and coverage rates for 95\% confidence intervals (in brackets), based on $1{,}000$ simulation draws. All network-type cells within a given design have equal size.\hspace{-0.04cm} All other features of the simulation design are the same as in Table 2,\hspace{-0.03cm} as described in Appendix C.1. \\
	\end{minipage}
\end{table}

\end{singlespace}
\vspace{-0.1cm}

\newpage

\section{Robustness Analyses for Project STAR}

\subsubsection*{D.1.\hspace{0.1cm} Examining Alternative Outcome Variables}

\noindent
Table \ref{tab:star-alternative-outcomes} reports IV estimates using alternative binary outcomes that indicate whether students’ math and reading scores fall in the top 25\% or top 75\% of the statewide distribution. The estimates in this table are statistically indistinguishable from those reported in Table \ref*{main-tab:star-estimates} of the main text, which use 50\% as the cutoff for defining student outcomes. This similarity suggests that the findings are robust to the choice of cutoff used to define high achievement.

\vspace{-0.15cm}
\begin{singlespace}
	\begin{table}[ht!]
	\centering
	\caption{Project STAR: IV Estimates for Alternative Outcome Variables}
	\label{tab:star-alternative-outcomes}
	\footnotesize
		\begin{tabular}{lccc@{\hskip 20pt}cc}
			\toprule
			& \multicolumn{2}{c}{Math Score} & & \multicolumn{2}{c}{Reading Score}  \\
			\cmidrule(lr){2-3} \cmidrule(lr){5-6}
			& Top 25\% & Top 75\% & & Top 25\%  & Top 75\% \\
			\midrule
			$h_m - h_f$ & \phantom{1}-0.088\hspace{0.2cm}(0.56) & \phantom{-}\phantom{1}0.625\hspace{0.2cm}\phantom{1}(2.69) & & \phantom{-}\phantom{1}0.060\hspace{0.2cm}(0.15) & \phantom{-}\phantom{1}0.211\hspace{0.2cm}(0.47) \\ 
			$J_{mm} - J_{fm}$ & \phantom{-}\phantom{1}5.373\hspace{0.2cm}(2.73) & \phantom{-}\phantom{1}2.021\hspace{0.2cm}(11.26) & & \phantom{-}\phantom{1}5.157\hspace{0.2cm}(0.98) & \phantom{-}\phantom{1}4.738\hspace{0.2cm}(0.76) \\ 
			$J_{ff} - J_{mf}$ & \phantom{-}\phantom{1}5.446\hspace{0.2cm}(2.45) & \phantom{-}\phantom{1}2.351\hspace{0.2cm}\phantom{1}(9.51) & & \phantom{-}\phantom{1}5.293\hspace{0.2cm}(1.20) & \phantom{-}\phantom{1}4.988\hspace{0.2cm}(0.79) \\ 
			\midrule
			Observations & 5,798 & 5,798 & & 5,718 & 5,718 \\
			School FE       & Yes  & Yes & & Yes  & Yes \\
			1st Stage $F$ ($\bar\omega_n^m$) & 8.74 & 15.53 & & 24.53 & 15.50 \\
			1st Stage $F$ ($\bar\omega_n^f$) & 13.72 & 15.24 & & 17.32 & 6.69 \\
			\bottomrule
	\end{tabular}
	\begin{minipage}{\textwidth}
		\renewcommand{\baselinestretch}{0}\selectfont\setlength{\spaceskip}{0em}\footnotesize
		\emph{Notes.} This table shows estimates of gender differences in peer effects for two alternative binary outcomes:  scoring in the top 25\% and scoring in the top 75\%. All the other specifications are the same as in Table \ref*{main-tab:star-estimates} of the main text.
	\end{minipage}
\end{table}
\end{singlespace}
\vspace{-0.1cm}

\subsubsection*{D.2.\hspace{0.1cm} Estimates for a Model with Uniform Peer Effects}

\noindent
Table \ref{tab:star-uniform-effects} examines how the estimates change if peer effects are restricted to be uniform, $J_{k\ell}=J$. Under this restriction, the estimates imply that girls face lower private effort costs than boys ($h_f>h_m$), although this difference is no longer statistically significant after including school fixed effects. These results differ substantially from those of the generalized model, suggesting that allowing peer effects to differ across students can materially change the conclusions about the mechanisms behind gender differences in student achievement.

\begin{singlespace}
	\begin{table}[ht!]
	\centering
	\caption{Project STAR: Estimates for a Model with Uniform Peer Effects}
	\label{tab:star-uniform-effects}
	\footnotesize
		\begin{tabular}{lccc@{\hskip 20pt}cc}
			\toprule
			& \multicolumn{2}{c}{Top 50\% Math Score} & & \multicolumn{2}{c}{Top 50\% Reading Score}  \\
			\cmidrule(lr){2-3} \cmidrule(lr){5-6}
			& (1) & (2) & & (1)  & (2) \\
			\midrule
			$h_m - h_f$ & \phantom{1}-0.363\hspace{0.2cm}(0.09) & \phantom{1}-1.199\hspace{0.2cm}(1.73) & & \phantom{1}-0.314\hspace{0.2cm}(0.10) & \phantom{1}-1.951\hspace{0.2cm}(1.69) \\ 
			$J_{mm} - J_{fm}$ & \textemdash & \textemdash & & \textemdash & \textemdash \\
			$J_{ff} - J_{mf}$ & \textemdash & \textemdash & & \textemdash & \textemdash \\
			\midrule
			Observations & 5,798 & 5,760 & & 5,718 & 5,681 \\
			School FE       & No  & Yes & & No  & Yes \\
			\bottomrule
	\end{tabular}
	\begin{minipage}{\textwidth}
		\renewcommand{\baselinestretch}{0}\selectfont\setlength{\spaceskip}{0em}\footnotesize
		\emph{Notes.} This table presents estimates of the gender fixed-effect difference, $h_m-h_f$, for a model imposing uniform peer effects, $J_{k\ell} = J$ for all $k, \ell \in \{m, f\}$. All the other specifications are the same as in Table \ref*{main-tab:star-estimates} of the main text.
	\end{minipage}
\end{table}
\end{singlespace}

\subsubsection*{D.3.\hspace{0.1cm} Verifying Random Assignment to Classrooms}

\enlargethispage{0.8\baselineskip}
\noindent
Under the experimental protocols, students within a school were randomly assigned to one of three classroom types. Yet, some schools had multiple classrooms of each type, raising the possibility of nonrandom sorting across classrooms within a given type. Although \cite{gr2008} argues that such sorting is unlikely, I follow his robustness check by restricting the sample to schools with only one classroom of each type, eliminating the possibility of within-type sorting. The first stage is weak for this restricted sample, particularly for math scores, but the estimates remain statistically indistinguishable from the full-sample results.

\vspace{-0.15cm}
\begin{singlespace}
	\begin{table}[ht!]
	\centering
	\caption{Project STAR: IV Estimates for Schools with One Classroom of Each Type}
	\label{tab:star-restricted-schools}
	\footnotesize
		\begin{tabular}{lccc@{\hskip 20pt}cc}
			\toprule
			& \multicolumn{2}{c}{Top 50\% Math Score} & & \multicolumn{2}{c}{Top 50\% Reading Score}  \\
			\cmidrule(lr){2-3} \cmidrule(lr){5-6}
			& (1) & (2) & & (1)  & (2) \\
			\midrule
			$h_m - h_f$ & \phantom{-}\phantom{1}0.819\hspace{0.2cm}(3.90) & \phantom{1}-1.211\hspace{0.2cm}(9.40) & & \phantom{1}-0.089\hspace{0.2cm}(0.13) & \phantom{1}-0.109\hspace{0.2cm}(0.49) \\ 
			$J_{mm} - J_{fm}$ & \phantom{-}\phantom{1}7.851\hspace{0.2cm}(16.93) & \phantom{-}\phantom{1}5.235\hspace{0.2cm}(4.18) & & \phantom{-}\phantom{1}4.696\hspace{0.2cm}(0.53) & \phantom{-}\phantom{1}4.172\hspace{0.2cm}(0.76) \\ 
			$J_{ff} - J_{mf}$ & \phantom{-}\phantom{1}9.676\hspace{0.2cm}(25.19) & \phantom{-}\phantom{1}2.781\hspace{0.2cm}(13.27) & & \phantom{-}\phantom{1}4.539\hspace{0.2cm}(0.76) & \phantom{-}\phantom{1}3.933\hspace{0.2cm}(1.32) \\ 
			\midrule
			Observations & 1,603 & 1,603 & & 1,587 & 1,587 \\
			School FE       & No  & Yes & & No  & Yes \\
			1st Stage $F$ ($\bar\omega_n^m$) & 6.18 & 6.41 & & 18.99 & 19.63 \\
			1st Stage $F$ ($\bar\omega_n^f$) & 3.24 & 3.36 & & 5.94 & 5.91 \\
			\bottomrule
	\end{tabular}
	\begin{minipage}{\textwidth}
		\renewcommand{\baselinestretch}{0}\selectfont\setlength{\spaceskip}{0em}\footnotesize
		\emph{Notes.} This table presents estimates of gender differences in peer effects for an estimation sample that is restricted to schools with one classroom of each type. All the other specifications are the same as in Table \ref*{main-tab:star-estimates} of the main text.
	\end{minipage}
\end{table}
\end{singlespace}

\vspace{-0.5cm}
\subsubsection*{D.4.\hspace{0.1cm} Misspecification Tests}

\enlargethispage{1.1\baselineskip}
\noindent
I conduct hypothesis tests to assess whether the type-specific parameter differences depend on observed classroom characteristics. Specifically, for each of $J_{ff} \!-\! J_{mf}$, $J_{mm} \!-\! J_{fm}$, and $h_m \!-\! h_f$, I evaluate the null hypothesis of no difference along each of the following criteria:
\begin{enumerate}
	\item high poverty classrooms ($\geq50$\% FRPL) and low poverty classrooms ($<50$\% FRPL)
	\item high minority classrooms ($<75$\% white) and low minority classrooms ($\geq 75$\% white)
	\item more experienced teachers ($>10$ years) and less experienced teachers ($\leq10$ years)
	\item more educated teachers (graduate deg.) and less educated teachers (no graduate deg.)
	\item rural classrooms (in rural district) and urban classrooms  (in urban or suburban district)
\end{enumerate}

\noindent
The $p$-values from these hypothesis tests are shown in Table \ref{tab:star-misspecification}, with math scores in Panel A and reading scores in Panel B. For both outcome variables, I find no evidence to reject the null hypothesis that each parameter difference is equal across groups. These results support the cross-group restrictions on parameter differences, which are required for identification.

\vspace{-0.25cm}
\begin{singlespace}
	\begin{table}[ht!]
	\centering
	\caption{Misspecification Tests: Project STAR}
	\label{tab:star-misspecification}
	\small
	\begin{tabular}{lccccc}
		\toprule
		& Large Share & Large Share & Teacher Has & Teacher Has & Rural \\
		& Poverty & Minority & $>$10 Years Experience & Higher Degree & District \\
		\midrule
		\multicolumn{6}{l}{\underline{\textit{Panel A. Math Score Outcome}}} \\
		$h_m - h_f$ & 0.710 & 0.590 & 0.790 & 0.880 & 0.910 \\
		$J_{mm} - J_{fm}$ & 0.614 & 0.644 & 0.746 & 0.965 & 0.929 \\
		$J_{ff} - J_{mf}$ & 0.626 & 0.603 & 0.746 & 0.993 & 0.921 \\
		\midrule
		\multicolumn{6}{l}{\underline{\textit{Panel B. Reading Score Outcome}}} \\
		$h_m - h_f$ & 0.996 & 0.611 & 0.760 & 0.983 & 0.834 \\
		$J_{mm} - J_{fm}$ & 0.442 & 0.756 & 0.798 & 0.987 & 0.523 \\
		$J_{ff} - J_{mf}$ & 0.622 & 0.711 & 0.790 & 0.990 & 0.659 \\
		\bottomrule
	\end{tabular}
	\begin{minipage}{\textwidth}
		\renewcommand{\baselinestretch}{0}\selectfont\setlength{\spaceskip}{0em}\footnotesize
		\emph{Notes.} This table presents $p$-values from one-dimensional tests of whether $h_m-h_f$, $J_{ff}-J_{mf}$, and $J_{mm}-J_{fm}$ correlate with class covariates. Outcomes indicate scoring in the top 50\% statewide on the corresponding exam.
	\end{minipage}
\end{table}

\end{singlespace}

\subsubsection*{D.5.\hspace{0.1cm} Sensitivity of Estimates to the Partitioning Rule}

\noindent
When estimating the model, I construct internal instruments by randomly dividing students within each classroom and type into two subsamples, $a$ and $b$, and using average outcomes in subsample $b$ as instruments for the corresponding averages in subsample $a$. Under the model’s assumptions, these averages measure the same underlying choice probabilities with conditionally independent sampling errors. However, the estimates can nevertheless vary across partitions in a finite sample. To assess sensitivity to the partitioning rule, I re-estimate the model $M\!=\!1{,}000$ times, drawing a new random partition each time while holding the estimation specification fixed. I then plot histograms of the resulting parameter estimates.

Figure \ref{fig:star-partitions} reports the histograms of the parameter estimates for each outcome. The estimates appear approximately normally distributed, and variation across partitions preserves the qualitative patterns reported in Table \ref*{main-tab:star-estimates}. The baseline IV estimates also lie within the main mass of these distributions. Together, these results suggest that the main findings are not particularly sensitive to the random partition used to construct the internal instruments.

\subsubsection*{D.6.\hspace{0.1cm} Finite-Sample Performance of the IV Estimator}

\enlargethispage{2\baselineskip}
\noindent
Table \ref{tab:star-simulations} examines the IV estimator’s finite-sample performance in simulations that reproduce Project STAR’s number of classrooms, classroom sizes, and gender composition. I ca-\\[-0.75cm]

\noindent
librate the parameters to the estimates in Column (1) of Table 3, the specification without controls.\footnote{Since only differences in the parameters are identified, I set the remaining parameters, namely the levels of $h_k$ and $J_{k\ell}$ and the distribution of $\alpha_n$, to match the observed distribution of outcomes across classrooms in the data.} I use this specification throughout all simulations and apply the same partitioning rule and standard error calculation used to obtain the estimates in Table 3. To simulate data, I follow the procedure described in Appendix C.1: running a fixed-point iteration under the logistic model, and then drawing agents' choices conditional on the resulting probabilities.

Based on $2{,}000$ simulation draws, Table \ref{tab:star-simulations} reports root mean squared errors, average estimated standard errors, and 95\% confidence interval coverage rates. Panel A presents results for math scores, and Panel B presents results for reading scores. The table demonstrates that average estimated standard errors are close to the root mean squared errors, and all coverage rates are within 2.5 percentage points of the nominal 95\% level. These results support the reliability of inference in finite samples that resemble the empirical application.

\begin{figure}[h]
	\begin{subfigure}{.329\textwidth}
		\centering
		\includegraphics[width=\linewidth]{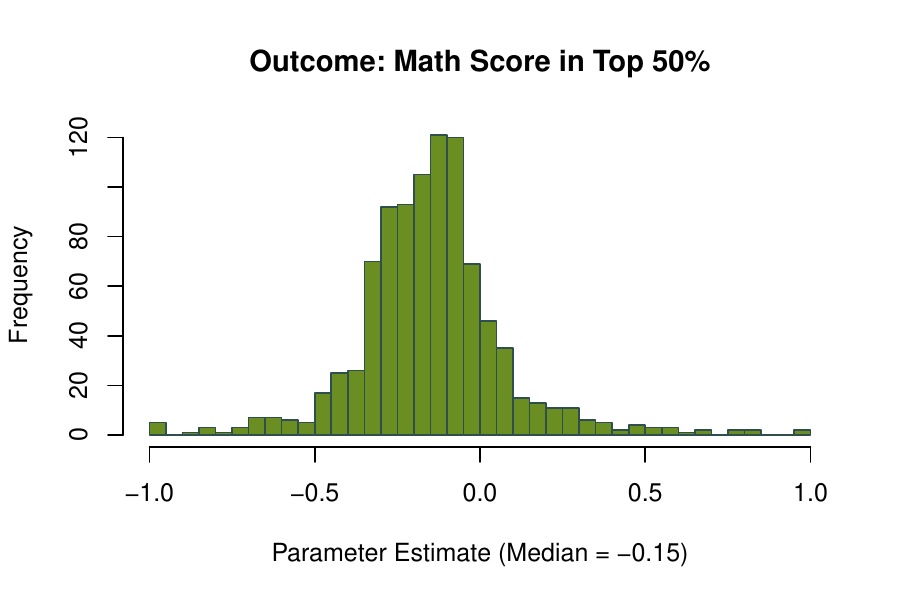}
		\caption{\scriptsize Parameter: $h_m - h_f$}
	\end{subfigure}
	\begin{subfigure}{.329\textwidth}
		\centering
		\includegraphics[width=\linewidth]{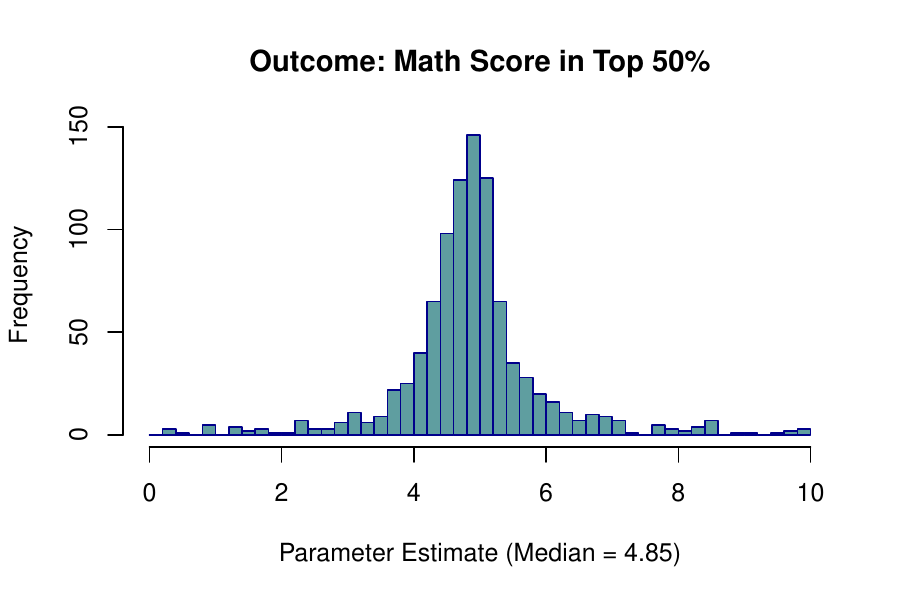}
		\caption{\scriptsize Parameter: $J_{mm} - J_{fm}$}
	\end{subfigure}
	\begin{subfigure}{.329\textwidth}
		\centering
		\includegraphics[width=\linewidth]{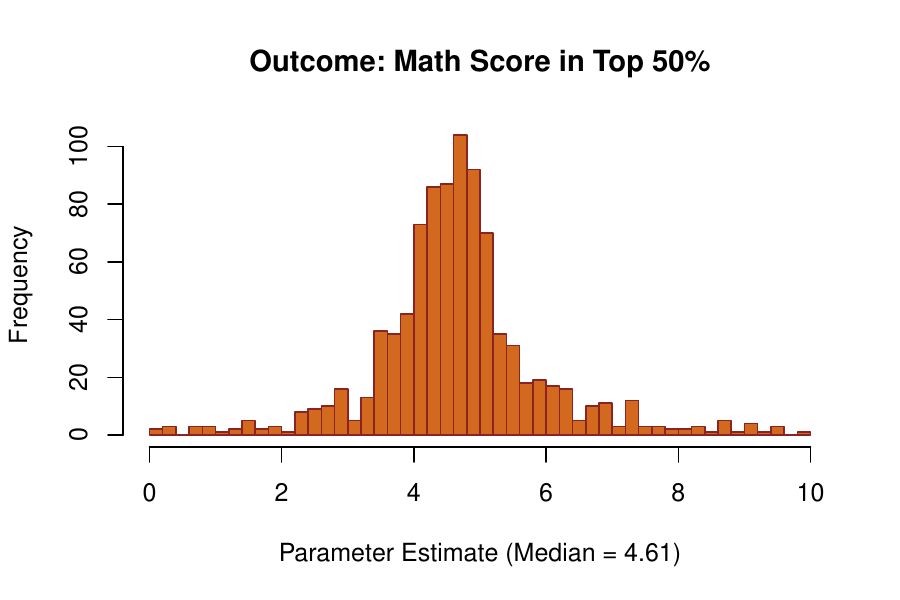}
		\caption{\scriptsize Parameter: $J_{ff} - J_{mf}$}
	\end{subfigure}
	\begin{subfigure}{.329\textwidth}
		\centering
		\includegraphics[width=\linewidth]{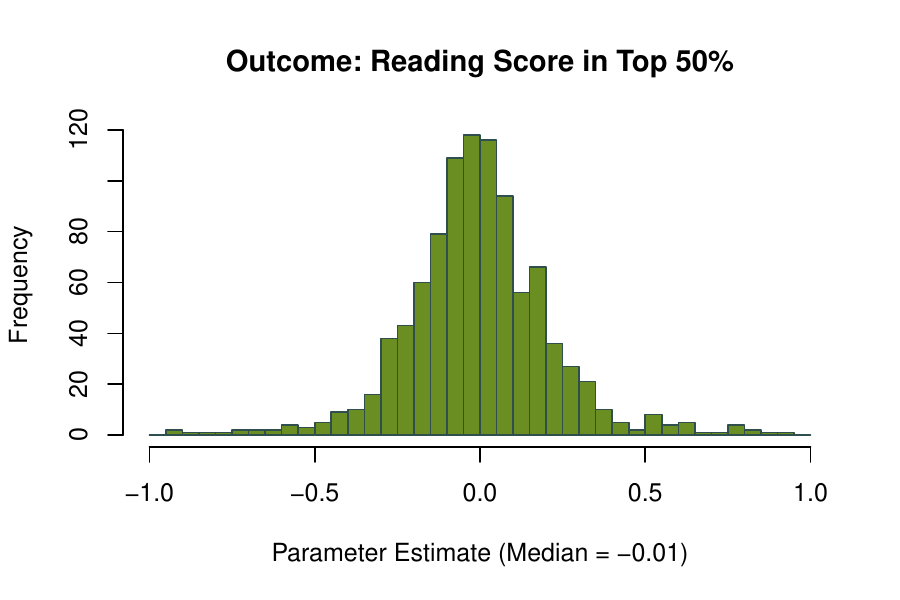}
		\caption{\scriptsize Parameter: $h_m - h_f$}
	\end{subfigure}
	\begin{subfigure}{.329\textwidth}
		\centering
		\includegraphics[width=\linewidth]{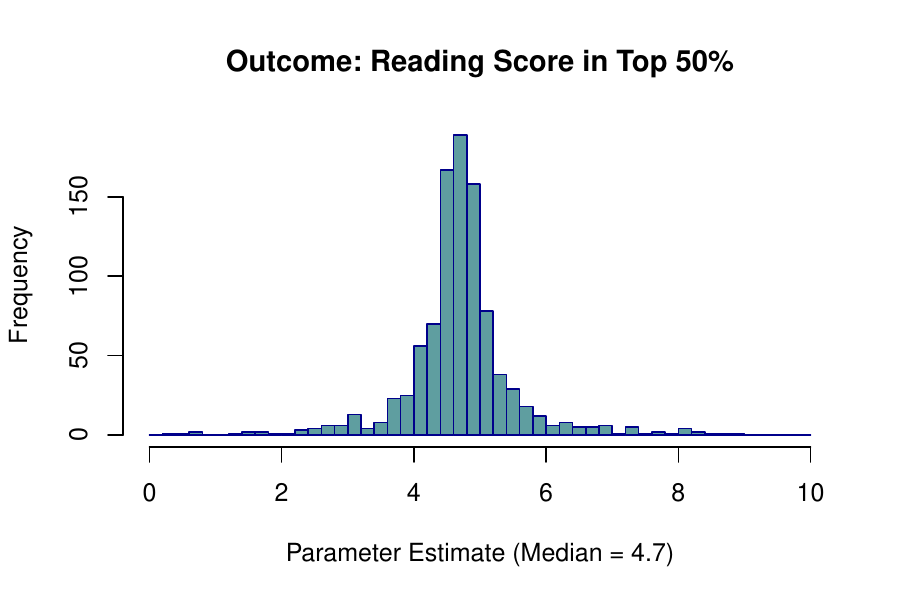}
		\caption{\scriptsize Parameter: $J_{mm} - J_{fm}$}
	\end{subfigure}
	\begin{subfigure}{.329\textwidth}
		\centering
		\includegraphics[width=\linewidth]{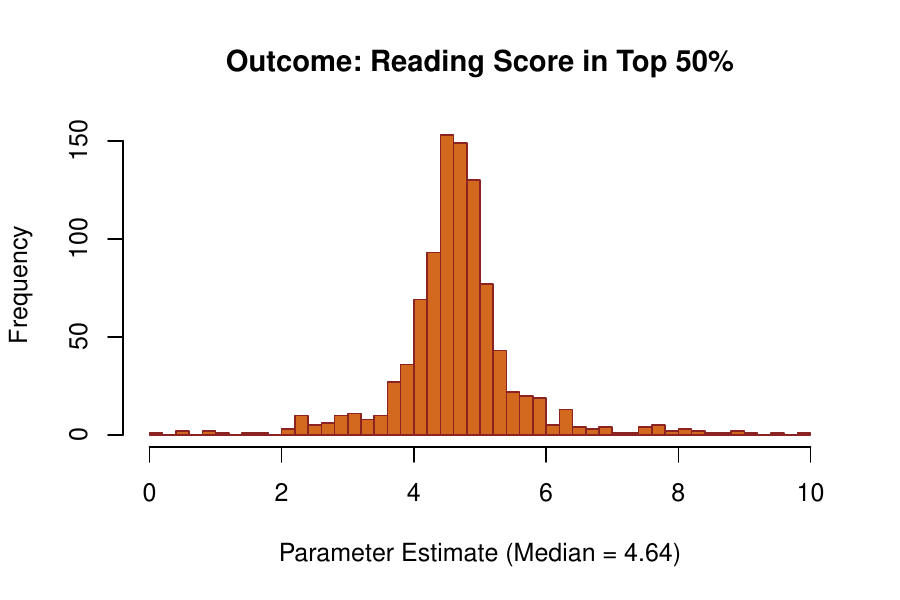}
		\caption{\scriptsize Parameter: $J_{ff} - J_{mf}$}
	\end{subfigure}
	\caption{\small Project STAR: Histograms of IV Estimates across $M = 1000$ Random Partitions}
	\label{fig:star-partitions}
\end{figure}

\enlargethispage{0.5\baselineskip}
\begin{singlespace}
	\begin{table}[ht!]
	\centering
	\caption{Project STAR: Finite-Sample Performance of the IV Estimator}
	\label{tab:star-simulations}
	\small
	\begin{tabular}{lcccc}
		\toprule
		 & Target Value & Root MSE & Avg.\ Standard Error & 95\% Coverage Rate \\
		\midrule
		\multicolumn{5}{l}{\underline{\textit{Panel A. Math Score Outcome}}} \\
		$h_m - h_f$ & -0.112 & 0.366 & 0.355 & 0.956 \\ 
		$J_{mm} - J_{fm}$ & 4.229 & 0.355 & 0.359 & 0.968 \\ 
		$J_{ff} - J_{mf}$ & 4.127 & 0.363 & 0.354 & 0.956 \\ 
		\midrule
		\multicolumn{5}{l}{\underline{\textit{Panel B. Reading Score Outcome}}} \\
		$h_m - h_f$ & -0.029 & 0.505 & 0.511 & 0.971 \\ 
		$J_{mm} - J_{fm}$ & 5.339 & 0.501 & 0.505 & 0.972 \\ 
		$J_{ff} - J_{mf}$ & 5.223 & 0.523 & 0.520 & 0.963 \\ 
		\bottomrule
	\end{tabular}
	\begin{minipage}{\textwidth}
		\renewcommand{\baselinestretch}{0}\selectfont\setlength{\spaceskip}{0em}\footnotesize
		\emph{Notes.} This table presents root mean squared errors, average estimated standard errors, and coverage rates for 95\% confidence intervals, based on $2{,}000$ simulation draws. The design exactly reproduces the number of classrooms and each classroom's size and gender composition in the Project STAR data. The simulations use the specification in Column (1) of Table \ref*{main-tab:star-estimates} (no school fixed effects), with target values set to the corresponding empirical estimates.
	\end{minipage}
\end{table}
\end{singlespace}

\section{Robustness Analyses for Primary School Deworming}

\subsubsection*{E.1.\hspace{0.1cm} Differences in Treatment Uptake within Schools}

\noindent
Figure \ref{fig:deworming-uptake-gaps} shows histograms of the within-school gaps in treatment uptake rates by student type. For any two types $k_1$ and $k_2$, the size and direction of these gaps vary greatly across schools, which suggests that factors at the school level shape disparities in treatment uptake.

\begin{figure}[!ht]
	\centering
	\begin{subfigure}{.62\textwidth}
		\centering
		\includegraphics[width=\linewidth]{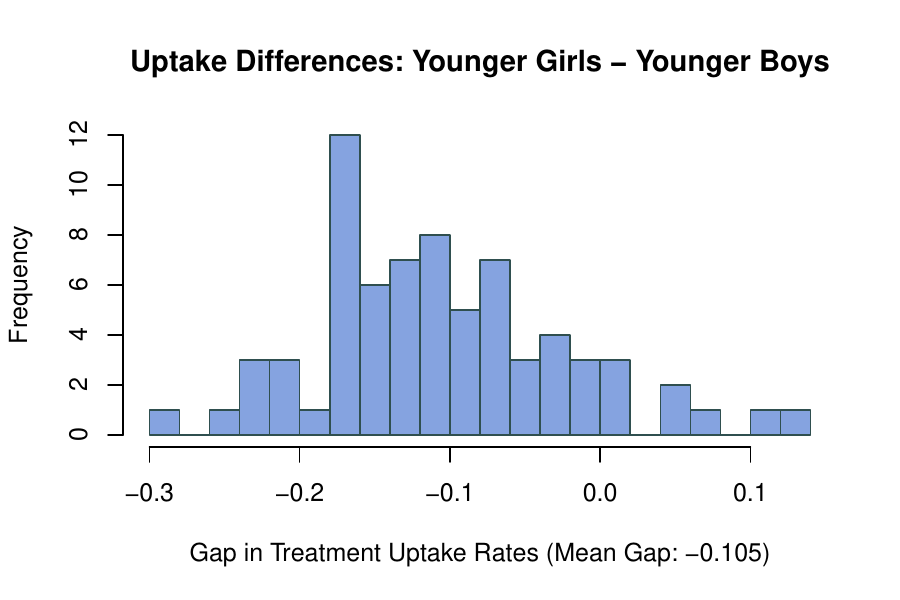}
	\end{subfigure}
	\begin{subfigure}{.62\textwidth}
		\centering
		\includegraphics[width=\linewidth]{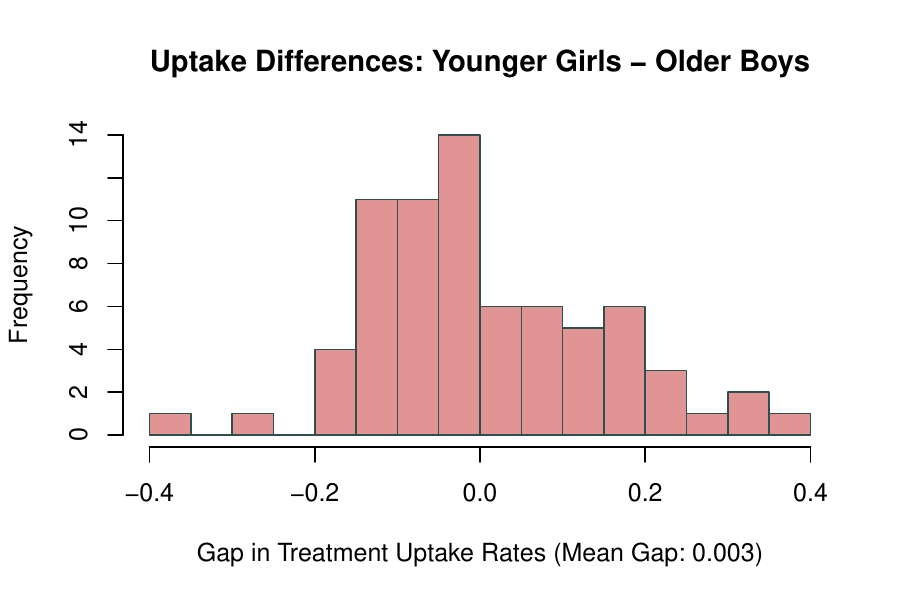}
	\end{subfigure}
	\begin{subfigure}{.62\textwidth}
		\centering
		\includegraphics[width=\linewidth]{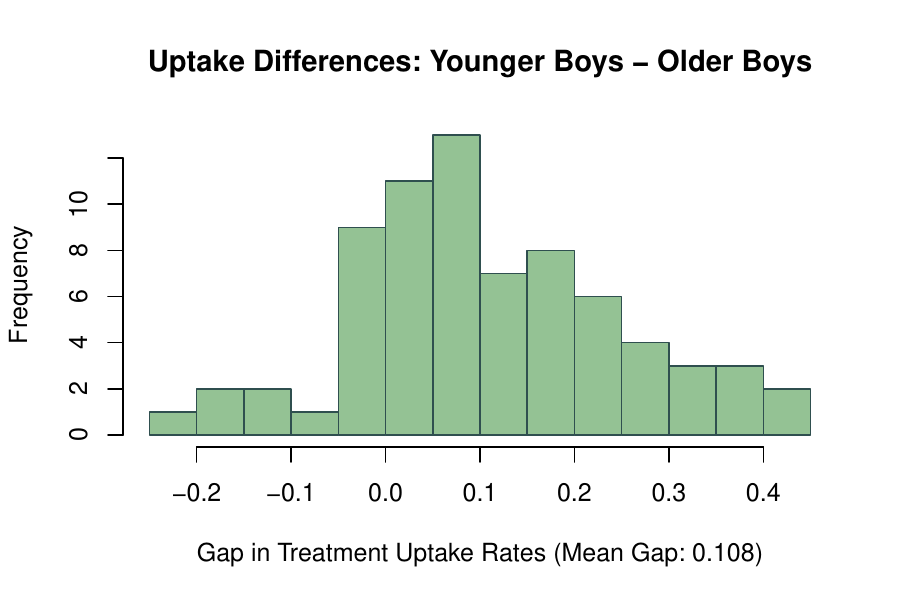}
	\end{subfigure}
	\caption{\small Within-School Differences in Treatment Uptake Rates Across Student Types}
	\label{fig:deworming-uptake-gaps}
\end{figure}

\subsubsection*{E.2.\hspace{0.1cm} IV Estimates for First-Time Eligible Students}

\noindent
To better isolate contemporaneous interactions from preexisting factors that can influence treatment decisions, Table \ref{tab:deworming-first-time} reports IV estimates for eligible students in schools receiving treatment for the first time. This sample includes Group 1 schools in 1998 and Group 2 schools in 1999. The resulting estimates are largely consistent with the full-sample results.

\begin{singlespace}
	\begin{table}[ht]
	\centering
	\caption{Deworming in Schools: IV Estimates for First-Time Eligible Students}
	\label{tab:deworming-first-time}
	\footnotesize
	\begin{tabular}{l@{\hskip 25pt}cccc}
		\toprule
		& (1) & (2) & (3) & (4) \\
		\midrule
		\multicolumn{2}{l}{\underline{\textit{Girls $\leq$ 13 and Boys $\leq$ 13}}} & & & \\
		$h_1 - h_2$ & \phantom{-}\phantom{1}0.338\hspace{0.2cm}(1.31) & \phantom{-}\phantom{1}0.227\hspace{0.2cm}(1.41) & \phantom{-}\phantom{1}0.779\hspace{0.2cm}(2.19) & \phantom{-}\phantom{1}1.184\hspace{0.2cm}(2.16) \\ 
		$J_{11} - J_{21}$ & \phantom{-}10.488\hspace{0.2cm}(3.29) & \phantom{-}10.777\hspace{0.2cm}(3.38) & \phantom{-}10.647\hspace{0.2cm}(2.22) & \phantom{-}10.973\hspace{0.2cm}(1.81) \\ 
		$J_{12} - J_{22}$ & -10.527\hspace{0.2cm}(2.57) & -11.105\hspace{0.2cm}(2.63) & -13.414\hspace{0.2cm}(2.77) & -14.942\hspace{0.2cm}(2.62) \\ 
		$J_{13} - J_{23}$ & \phantom{1}-0.402\hspace{0.2cm}(1.21) & \phantom{1}-0.001\hspace{0.2cm}(1.23) & \phantom{-}\phantom{1}1.369\hspace{0.2cm}(1.46) & \phantom{-}\phantom{1}1.452\hspace{0.2cm}(1.56) \\
		& & & & \\
		\multicolumn{2}{l}{\underline{\textit{Girls $\leq$ 13 and Boys $>$ 13}}} & & & \\
		$h_1 - h_3$ & \phantom{-}\phantom{1}0.494\hspace{0.2cm}(1.00) & \phantom{-}\phantom{1}0.062\hspace{0.2cm}(0.95) & \phantom{-}\phantom{1}1.079\hspace{0.2cm}(1.09) & \phantom{-}\phantom{1}2.193\hspace{0.2cm}(1.10) \\ 
		$J_{11} - J_{31}$ & \phantom{-}\phantom{1}5.215\hspace{0.2cm}(3.83) & \phantom{-}\phantom{1}6.338\hspace{0.2cm}(3.86) & \phantom{-}\phantom{1}5.645\hspace{0.2cm}(4.19) & \phantom{-}\phantom{1}5.226\hspace{0.2cm}(4.75) \\ 
		$J_{12} - J_{32}$ & \phantom{1}-1.251\hspace{0.2cm}(3.51) & \phantom{1}-3.497\hspace{0.2cm}(4.66) & \phantom{1}-5.055\hspace{0.2cm}(6.07) & \phantom{1}-5.898\hspace{0.2cm}(7.05) \\ 
		$J_{13} - J_{33}$ & \phantom{1}-4.206\hspace{0.2cm}(0.72) & \phantom{1}-2.647\hspace{0.2cm}(1.37) & \phantom{1}-1.248\hspace{0.2cm}(1.57) & \phantom{1}-1.086\hspace{0.2cm}(1.74) \\
		& & & & \\
		\multicolumn{2}{l}{\underline{\textit{Boys $\leq$ 13 and Boys $>$ 13}}} & & & \\
		$h_2 - h_3$ & \phantom{-}\phantom{1}0.156\hspace{0.2cm}(1.83) & \phantom{1}-0.165\hspace{0.2cm}(1.67) & \phantom{-}\phantom{1}0.300\hspace{0.2cm}(1.71) & \phantom{-}\phantom{1}1.009\hspace{0.2cm}(2.20) \\ 
		$J_{21} - J_{31}$ & \phantom{1}-5.273\hspace{0.2cm}(5.98) & \phantom{1}-4.439\hspace{0.2cm}(5.97) & \phantom{1}-5.003\hspace{0.2cm}(5.07) & \phantom{1}-5.747\hspace{0.2cm}(6.08) \\ 
		$J_{22} - J_{32}$ & \phantom{-}\phantom{1}9.276\hspace{0.2cm}(5.14) & \phantom{-}\phantom{1}7.608\hspace{0.2cm}(5.76) & \phantom{-}\phantom{1}8.359\hspace{0.2cm}(6.28) & \phantom{-}\phantom{1}9.043\hspace{0.2cm}(7.72) \\ 
		$J_{23} - J_{33}$ & \phantom{1}-3.804\hspace{0.2cm}(1.22) & \phantom{1}-2.647\hspace{0.2cm}(1.70) & \phantom{1}-2.617\hspace{0.2cm}(1.31) & \phantom{1}-2.538\hspace{0.2cm}(1.31) \\ 
		\midrule
		Year Fixed Effects & No & Yes & Yes & Yes \\
		Student Controls & No & No & Yes & Yes \\
		School Controls  & No & No & No & Yes \\
		1st Stage $F$ ($\bar\omega_n^1$) & 9.32 & 6.97 & 12.64 & 10.34 \\
		1st Stage $F$ ($\bar\omega_n^2$) & 10.02 & 8.05 & 12.08 & 9.74 \\
		1st Stage $F$ ($\bar\omega_n^3$) & 32.17 & 13.42 & 11.08 & 10.55 \\
		\bottomrule
	\end{tabular}
	\begin{minipage}{\textwidth}
		\renewcommand{\baselinestretch}{0}\selectfont\setlength{\spaceskip}{0em}\footnotesize
		\emph{Notes.}\ This table shows the estimated differences in spillover effects in treatment uptake for an estimation sample restricted to  eligible students in schools receiving treatment for the first time: Group 1 schools in 1998 and Group 2 schools in 1999. Apart from this sample restriction, all other specifications are the same as in Table \ref*{main-tab:deworming-estimates} of the paper. \\[-0.55cm]
	\end{minipage}
\end{table}

\end{singlespace}

\subsubsection*{E.3.\hspace{0.1cm} Misspecification Tests}

\noindent
For each type-specific parameter difference, $h_{k_1,t}-h_{k_2,t}$ and $J_{k_1\ell}-J_{k_2\ell}$, I test the null hypothesis of equality across groups that are defined by the following school characteristics:
\begin{enumerate}
	\item local infection rate (above vs.\ below median)
	\item treatment group (1998 group vs.\ 1999 group)
	\item population density within 3 km (above vs.\ below median)
	\item population density between 3 and 6 km (above vs.\ below median)
\end{enumerate}

\noindent
Table \ref{tab:deworming-misspecification} reports the $p$-values from the hypothesis tests. Overall, I find very little evidence that the parameter differences vary with observable school characteristics. These results are therefore consistent with the restrictions on parameter differences underlying identification.

\begin{singlespace}
	\begin{table}[ht!]
	\centering
	\caption{Misspecification Tests: Primary School Deworming} 
	\label{tab:deworming-misspecification}
	\footnotesize
	\begin{tabular}{lcccc}
		\toprule
		& High Local & Treatment & High Population & High Population  \\
		& Infection Rate & Group & Density ($<3$km) & Density (3-6km)  \\
		\midrule
		\multicolumn{1}{l}{\underline{\textit{Girls $\leq$ 13 and Boys $\leq$ 13}}} & & & & \\
		$h_1 - h_2$ & 0.954 & 0.914 & 0.053 & 0.748 \\
		$J_{11} - J_{21}$ & 0.304 & 0.722 & 0.583 & 0.231 \\
		$J_{12} - J_{22}$ & 0.120 & 0.772 & 0.381 & 0.122 \\
		$J_{13} - J_{23}$ & 0.983 & 0.735 & 0.852 & 0.650 \\
		& & & & \\
		\multicolumn{1}{l}{\underline{\textit{Girls $\leq$ 13 and Boys $>$ 13}}} & & & &\\
		$h_1 - h_3$ & 0.637 & 0.831 & 0.886 & 0.727 \\
		$J_{11} - J_{31}$ & 0.158 & 0.806 & 0.277 & 0.281 \\
		$J_{12} - J_{32}$ & 0.049 & 0.799 & 0.395 & 0.053 \\
		$J_{13} - J_{33}$ & 0.348 & 0.754 & 0.746 & 0.715 \\
		& & & & \\
		\multicolumn{1}{l}{\underline{\textit{Boys $\leq$ 13 and Boys $>$ 13}}} & & & & \\
		$h_2 - h_3$ & 0.343 & 0.751 & 0.367 & 0.999 \\
		$J_{21} - J_{31}$ & 0.580 & 0.934 & 0.682 & 0.886 \\
		$J_{22} - J_{32}$ & 0.866 & 0.866 & 0.993 & 0.709 \\
		$J_{23} - J_{33}$ & 0.278 & 0.835 & 0.795 & 0.433 \\
		\bottomrule
	\end{tabular}
	\begin{minipage}{\textwidth}
		\renewcommand{\baselinestretch}{0}\selectfont\setlength{\spaceskip}{0em}\footnotesize
		\emph{Notes.}\ This table presents $p$-values from one-dimensional tests of whether spillover effect differences and type fixed-effect differences correlate with observed school characteristics. The regression specification includes year fixed effects but does not add student controls. All other specifications are the same as in Table \ref*{main-tab:deworming-estimates} of the main text. \\
	\end{minipage}
\end{table}

\end{singlespace}

\subsubsection*{E.4.\hspace{0.1cm} Sensitivity of Estimates to the Partitioning Rule}

\noindent
As in the first empirical application, I evaluate the sensitivity of the IV estimates to the random partition used to construct the internal instruments. I re-estimate the model $M=1{,}000$ times, each time drawing a new random partition, and report histograms of the resulting estimates in Figure \ref{fig:deworming-partitions}. The baseline IV estimates in Table \ref*{main-tab:deworming-estimates} lie within the main mass of these distributions, and variation across partitions generally preserves the qualitative patterns in the baseline results. These findings suggest that the main conclusions in Section 4.2 of the main text are not especially sensitive to the partition that is used to construct instruments.

\clearpage
\begin{figure}[h]
	\begin{subfigure}{.329\textwidth}
		\centering
		\includegraphics[width=\linewidth]{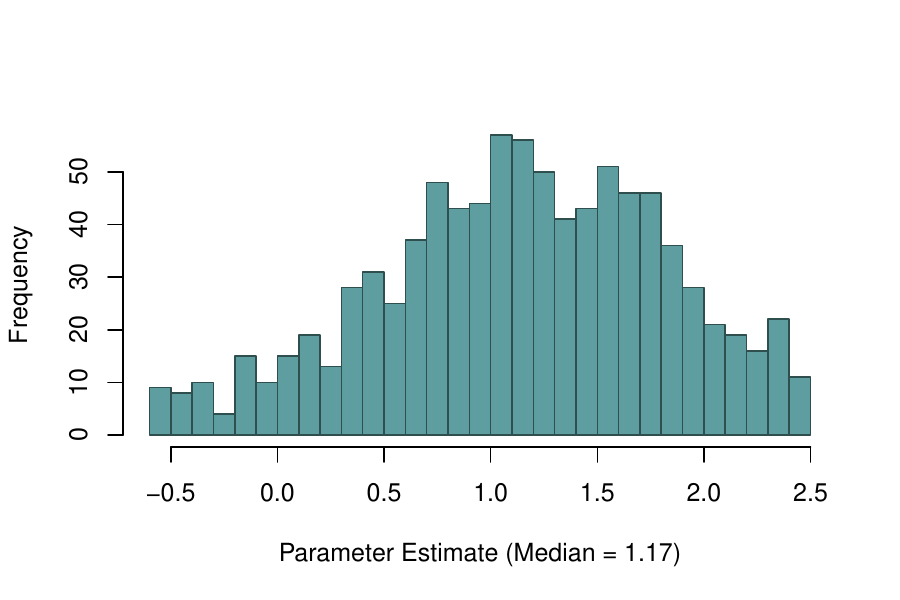} \vspace{-0.6cm}
		\caption{\scriptsize Parameter: $h_1 - h_2$}
	\end{subfigure}
	\begin{subfigure}{.329\textwidth}
		\centering
		\includegraphics[width=\linewidth]{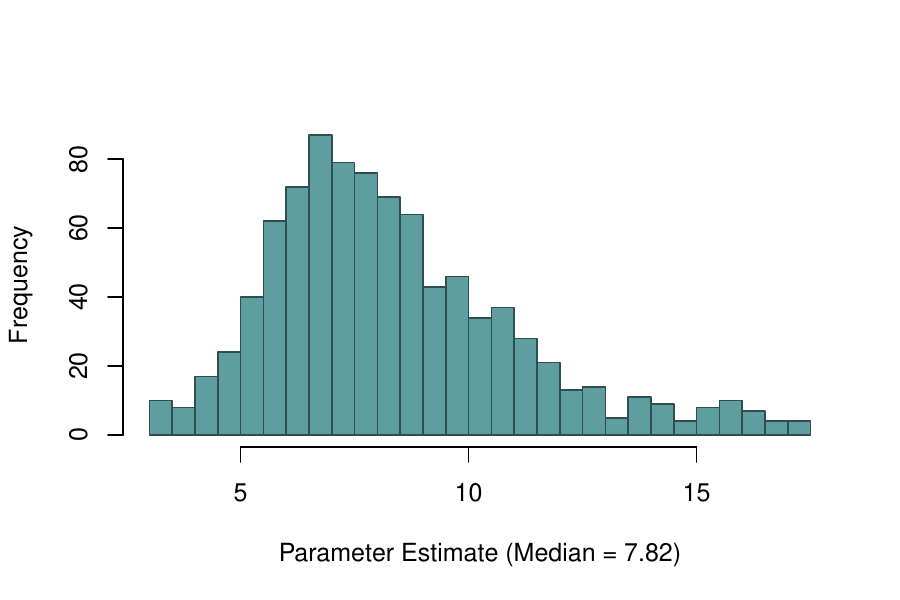} \vspace{-0.6cm}
		\caption{\scriptsize Parameter: $J_{11} - J_{21}$}
	\end{subfigure}
	\begin{subfigure}{.329\textwidth}
		\centering
		\includegraphics[width=\linewidth]{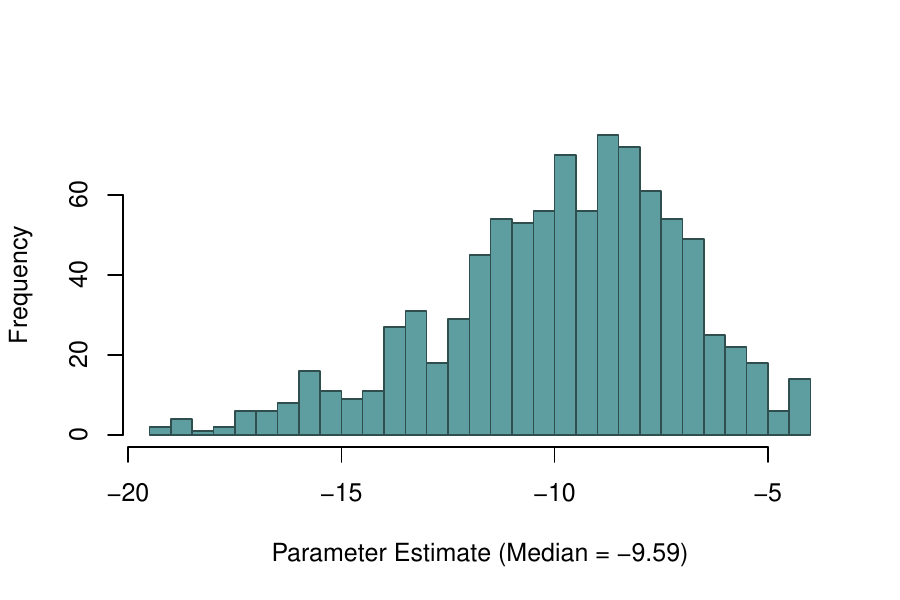} \vspace{-0.6cm}
		\caption{\scriptsize Parameter: $J_{12} - J_{22}$}
	\end{subfigure}
	\begin{subfigure}{.329\textwidth}
		\centering
		\includegraphics[width=\linewidth]{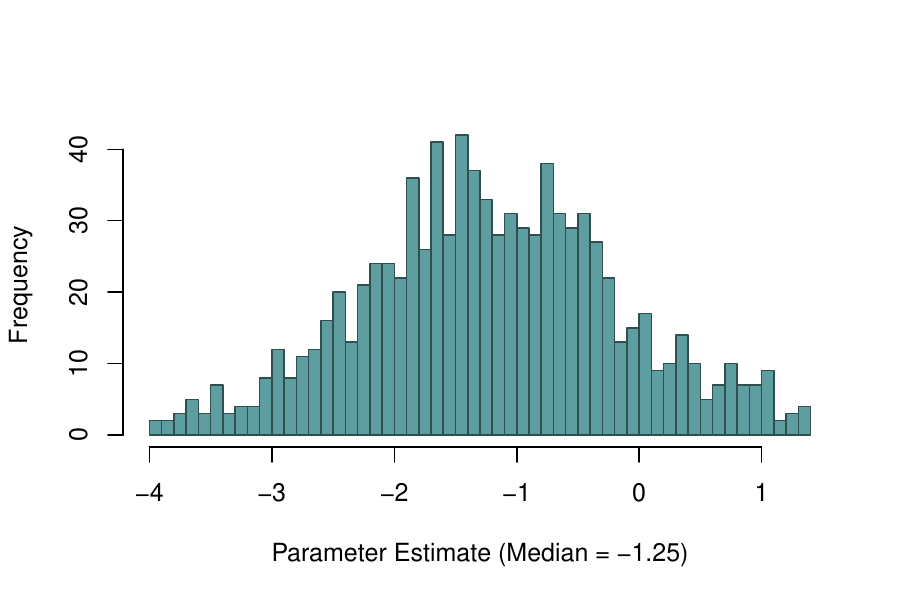} \vspace{-0.6cm}
		\caption{\scriptsize Parameter: $J_{13} - J_{23}$}
	\end{subfigure}
	\begin{subfigure}{.329\textwidth}
		\centering
		\includegraphics[width=\linewidth]{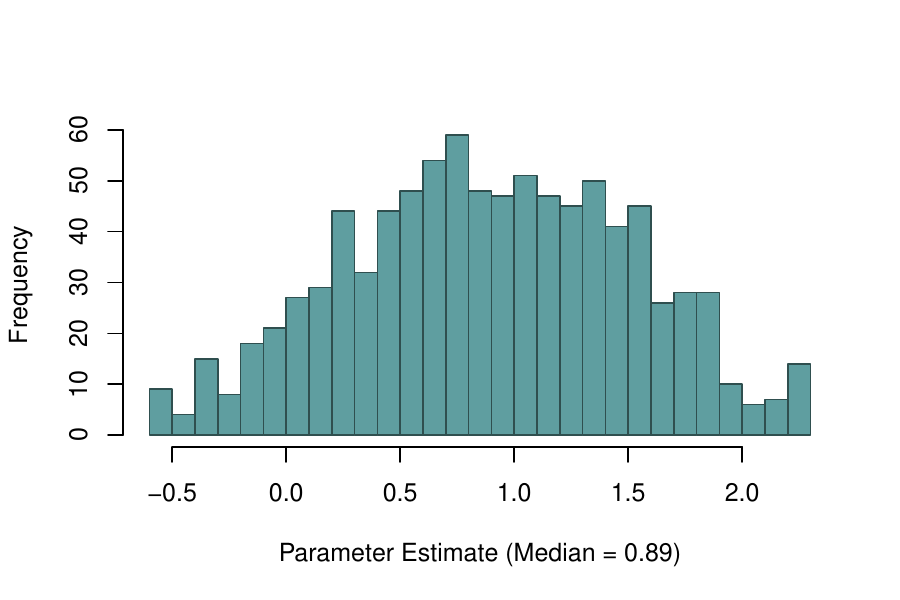} \vspace{-0.6cm}
		\caption{\scriptsize Parameter: $h_1 - h_3$}
	\end{subfigure}
	\begin{subfigure}{.329\textwidth}
		\centering
		\includegraphics[width=\linewidth]{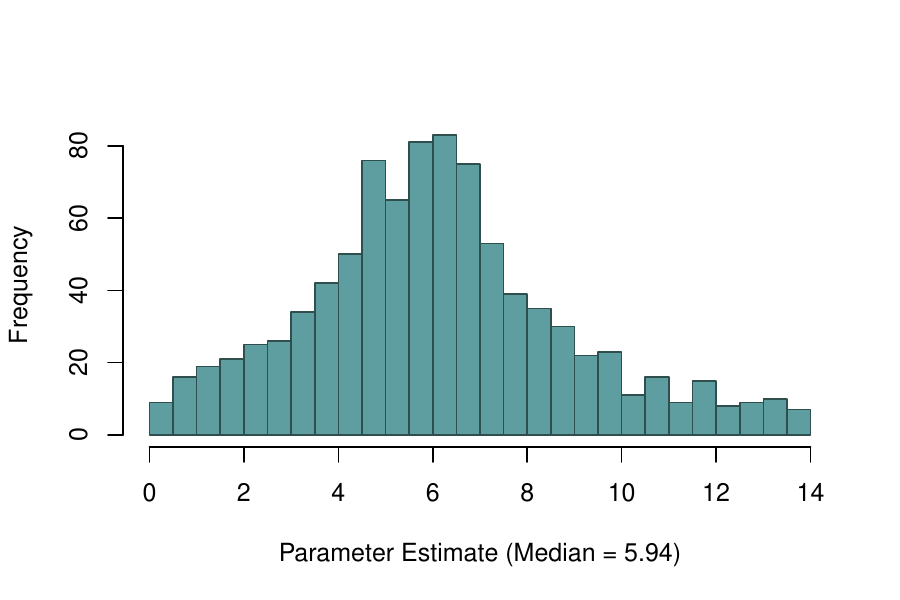} \vspace{-0.6cm}
		\caption{\scriptsize Parameter: $J_{11} - J_{31}$}
	\end{subfigure}
	\begin{subfigure}{.329\textwidth}
		\centering
		\includegraphics[width=\linewidth]{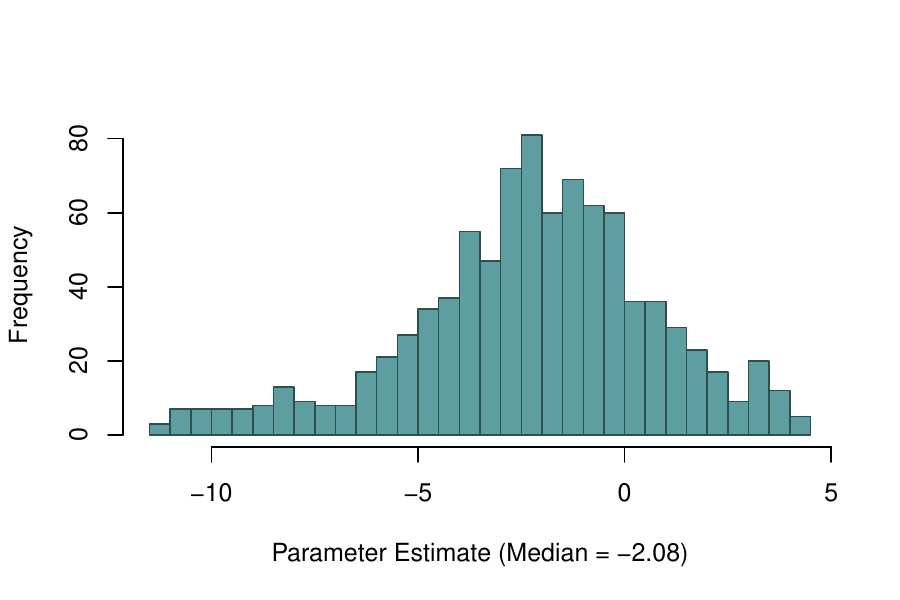} \vspace{-0.6cm}
		\caption{\scriptsize Parameter: $J_{12} - J_{32}$}
	\end{subfigure}
	\begin{subfigure}{.329\textwidth}
		\centering
		\includegraphics[width=\linewidth]{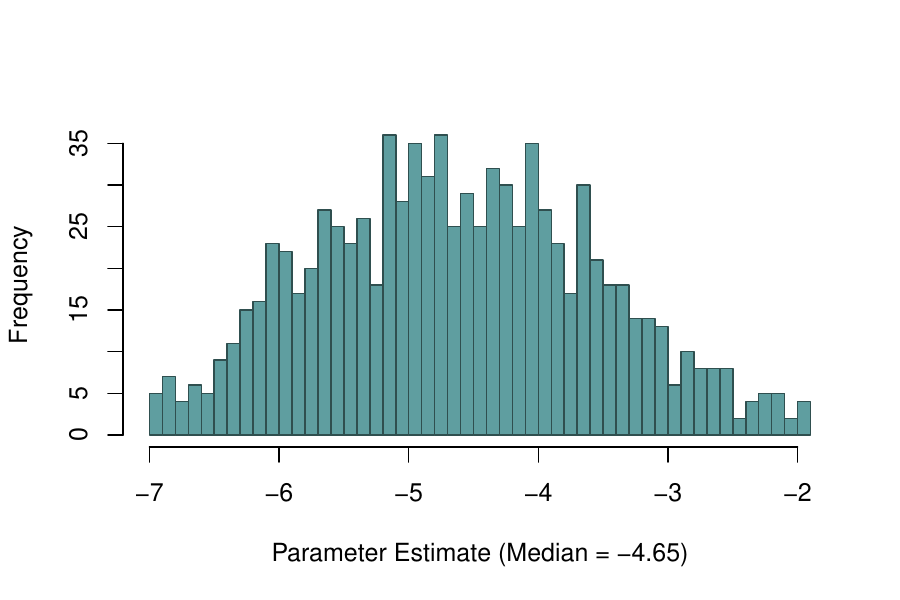} \vspace{-0.6cm}
		\caption{\scriptsize Parameter: $J_{13} - J_{33}$}
	\end{subfigure}
	\begin{subfigure}{.329\textwidth}
		\centering
		\includegraphics[width=\linewidth]{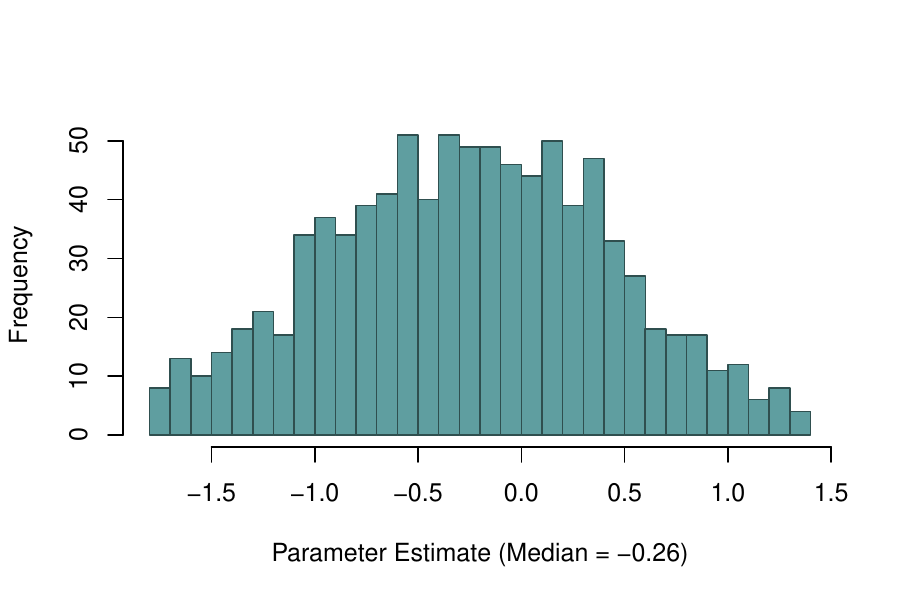} \vspace{-0.6cm}
		\caption{\scriptsize Parameter: $h_2 - h_3$}
	\end{subfigure}
	\begin{subfigure}{.329\textwidth}
		\centering
		\includegraphics[width=\linewidth]{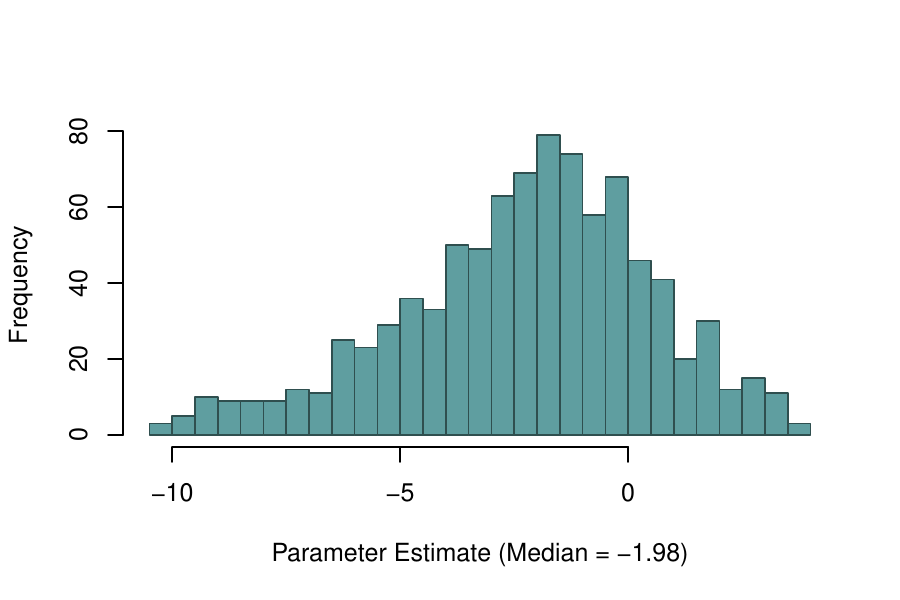} \vspace{-0.6cm}
		\caption{\scriptsize Parameter: $J_{21} - J_{31}$}
	\end{subfigure}
	\begin{subfigure}{.329\textwidth}
		\centering
		\includegraphics[width=\linewidth]{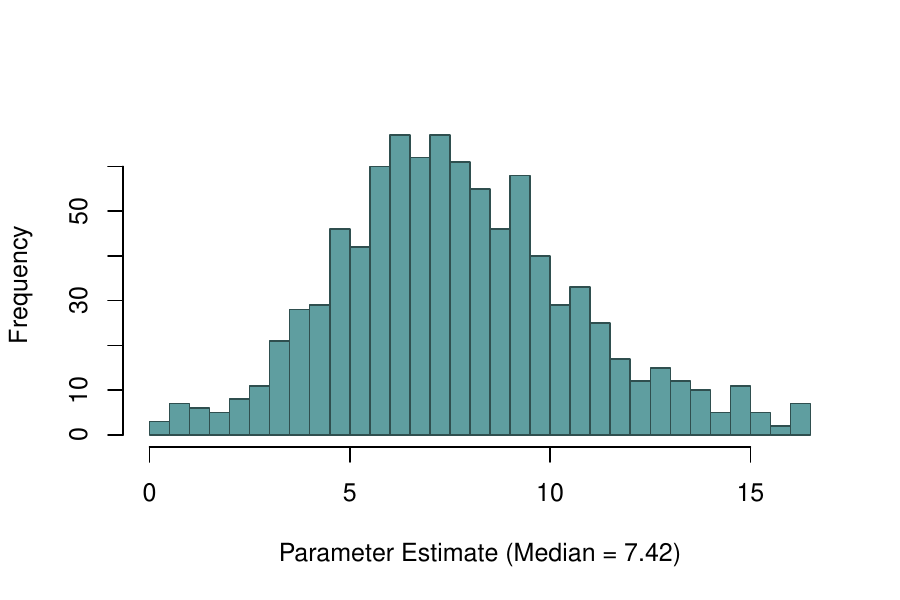} \vspace{-0.6cm}
		\caption{\scriptsize Parameter: $J_{22} - J_{32}$}
	\end{subfigure}
	\begin{subfigure}{.329\textwidth}
		\centering
		\includegraphics[width=\linewidth]{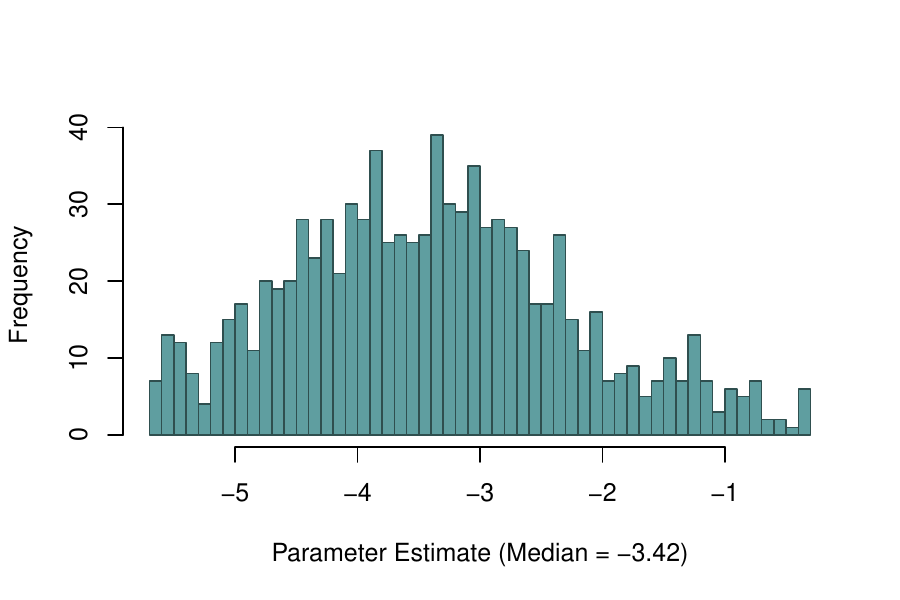} \vspace{-0.6cm}
		\caption{\scriptsize Parameter: $J_{23} - J_{33}$}
	\end{subfigure}
	\caption{\small Primary School Deworming: Histograms of IV Estimates across $M = 1000$ Random Partitions. The estimation procedure uses the same regression specification as in Table \ref*{main-tab:deworming-estimates}, Column 4.}
	\label{fig:deworming-partitions}
\end{figure}

\vspace{-0.75cm}
\subsubsection*{E.5.\hspace{0.1cm} Finite-Sample Performance of the IV Estimator}

\enlargethispage{1.5\baselineskip}
\noindent
Table \ref{tab:deworming-simulations} assesses the IV estimator’s finite-sample performance in simulations that reprod-\\uce the deworming application’s number of schools, school sizes, and gender and age composition. I calibrate the parameters to estimates from the specification without year or zone fixed effects or student or school controls.\footnote{As in Appendix D.6, only differences in the parameters are identified, so I set the levels of $h_k$ and $J_{k\ell}$, and the distribution of the network effects $\alpha_n$, to match the distribution of treatment uptake across schools in the data.} For this reason, the target values differ slightly from the estimates in Table 4, which includes these controls. I apply the same partitioning rule and standard error calculation used to obtain the paper's main estimates. To simulate data, I follow the procedure described in Appendix C.1, running a fixed-point iteration under the logistic model and then drawing agents' choices conditional on those probabilities.

Based on $2{,}000$ simulation draws, Table \ref{tab:deworming-simulations} reports root mean squared errors, average estimated standard errors, and coverage rates for 95\% confidence intervals. The three panels present results for each pair of student types. Average estimated standard errors are close to root mean squared errors, and coverage rates range from 91\% to 99\%. Overall, these results support the reliability of inference in finite samples that resemble the empirical application.

\enlargethispage{0.5\baselineskip}
\begin{singlespace}
	\begin{table}[ht!]
	\centering
	\caption{Primary School Deworming: Finite-Sample Performance of the IV Estimator}
	\label{tab:deworming-simulations}
	\small
	\begin{tabular}{lcccc}
		\toprule
		 & Target Value & Root MSE & Avg.\ Standard Error & 95\% Coverage Rate \\
		\midrule
		\multicolumn{5}{l}{\underline{\textit{Girls $\leq$ 13 and Boys $\leq$ 13}}} \\
		$h_1 - h_2$ & 0.846 & 0.494 & 0.509 & 0.965 \\ 
		$J_{11} - J_{21}$ & 8.069 & 1.210 & 1.147 & 0.965 \\ 
		$J_{12} - J_{22}$ & -9.244 & 1.579 & 1.475 & 0.961 \\ 
		$J_{13} - J_{23}$ & 0.130 & 1.519 & 1.396 & 0.963 \\ 
		\midrule
		\multicolumn{5}{l}{\underline{\textit{Girls $\leq$ 13 and Boys $>$ 13}}} \\
		$h_1 - h_3$ & -0.629 & 0.492 & 0.517 & 0.991 \\ 
		$J_{11} - J_{31}$ & 7.563 & 1.304 & 1.176 & 0.929 \\ 
		$J_{12} - J_{32}$ & -1.490 & 1.792 & 1.513 & 0.956 \\ 
		$J_{13} - J_{33}$ & -5.229 & 1.699 & 1.431 & 0.915 \\ 
		\midrule
		\multicolumn{5}{l}{\underline{\textit{Boys $\leq$ 13 and Boys $>$ 13}}} \\
		$h_2 - h_3$ & -1.475 & 0.530 & 0.519 & 0.944 \\ 
		$J_{21} - J_{31}$ & -0.506 & 1.343 & 1.181 & 0.957 \\ 
		$J_{22} - J_{32}$ & 7.754 & 1.721 & 1.521 & 0.932 \\ 
		$J_{23} - J_{33}$ & -5.358 & 1.719 & 1.439 & 0.912 \\ 
		\bottomrule
	\end{tabular}
	\begin{minipage}{\textwidth}
		\renewcommand{\baselinestretch}{0}\selectfont\setlength{\spaceskip}{0em}\footnotesize
		\emph{Notes.} This table presents root mean squared errors, average estimated standard errors, and coverage rates for 95\% confidence intervals, based on $2{,}000$ simulation draws. The design exactly reproduces the number of schools and each school's size and type composition in the deworming data. The simulations use the specification without year or zone fixed effects or student or school controls, with target values set to the corresponding empirical estimates.
	\end{minipage}
\end{table}
\end{singlespace}

\end{appendix}
\vspace{-0.65cm}
\onehalfspacing
\bibliographystyle{ecta-fullname}
\bibliography{References}